%% file: main.tex
\documentclass{article}

\usepackage{PRIMEarxiv}
\input{util}

\usepackage{parskip}

\title{Approximation Algorithms for Size-Constrained Non-Monotone Submodular Maximization in Deterministic Linear Time}

\author{
  Yixin Chen \\
  Department of Computer Science \& Engineering \\
  Texas A\&M University \\
  College Station, Texas, USA\\
  \texttt{chen777@tamu.edu} \\
   \And
  Alan Kuhnle \\
  Department of Computer Science \& Engineering \\
  Texas A\&M University \\
  College Station, Texas, USA\\
  \texttt{kuhnle@tamu.edu} \\
}
\begin{document}

\maketitle

\input{content}


\bibliographystyle{plainnat}
\bibliography{main}

\clearpage

\input{appendix}
\end{document}

%% file: util.tex
\usepackage[utf8]{inputenc} 
\usepackage[T1]{fontenc}    
\usepackage{url}            
\usepackage{booktabs}       
\usepackage{amsfonts}       
\usepackage{nicefrac}       
\usepackage{microtype}      
\pdfoutput=1
\usepackage[numbers]{natbib}
\usepackage{amsmath}
\usepackage{amsthm}
\usepackage{algorithm}
\usepackage{algorithmicx}
\usepackage[noend]{algpseudocode}
\usepackage{verbatim}
\usepackage{graphicx}
\usepackage[space]{grffile}
\usepackage{tikz}
\usepackage{subfigure}
\usepackage{balance}
\usetikzlibrary{arrows}
\newtheorem{theorem}{Theorem}
\newtheorem*{theorem*}{Theorem}
\newtheorem{lemma}{Lemma}

\theoremstyle{definition}

\newtheorem{claim}{Claim}

\graphicspath{{figures/}}
\usepackage{mathrsfs}
\usepackage{amsfonts}
\usepackage{booktabs}
\usepackage{changepage}
\usepackage{xspace}

\usepackage{tikz}
\usetikzlibrary{arrows}

\usepackage{xcolor}

\newcommand{\ie}{\textit{i.e.}\xspace}
\newcommand{\eg}{\textit{e.g.}\xspace}

\newcommand{\ex}[1]{ \mathbb{E} \left[ #1 \right] }

\newcommand{\epsi}[0]{ \varepsilon }
\newcommand{\reals}[0]{ \mathbb{R}^+ }
\newcommand{\func}[2]{ #1 \left( #2 \right) }
\newcommand{\ff}[1]{ f \left( #1 \right) }
\newcommand{\parth}[1]{ \left( #1 \right) }
\newcommand{\marge}[2]{ \delta_{#1} \left( #2 \right) }

\DeclareMathOperator*{\argmax}{arg\,max}

\newcommand{\opt}{\textsc{OPT}\xspace}

\newcommand{\oh}[1]{O\left( #1 \right)}
\newcommand{\ig}{\textsc{Inter\-lace\-Greedy}\xspace}
\newcommand{\lrig}{\textsc{Interpolated\-Greedy}\xspace}

\newcommand{\lrtig}{\textsc{Derandomized\-Interpolated\-Greedy}\xspace}

\newcommand{\add}{\texttt{ADD}\xspace}

\newcommand{\unc}{\textsc{UncMax}\xspace}
\newcommand{\uncalg}{\textsc{UncMaxCompetitiveAlg}\xspace}

\newcommand\numberthis{\addtocounter{equation}{1}\tag{\theequation}}
\newcommand{\ratio}{(2b + 4)(1 + \beta/b) + \epsi}
\newcommand{\ratioapprox}{23.313 + \epsi}
\newcommand{\lcratio}{(2b + 4)(1 + 1/b)}
\newcommand{\lcratioapprox}{11.657}
\newcommand{\um}{\text{um}}
\newcommand{\maxSize}{2\ell(k / b + 1)\log_2(k)}
\newcommand{\ratioPrime}{1 + \beta / b}

\newcommand{\fb}{$\texttt{fb}$\xspace}
\newcommand{\maxcut}{$\texttt{maxcut}$\xspace}
\newcommand{\revmax}{$\texttt{revmax}$\xspace}
\newcommand{\imgsum}{$\texttt{imgsum}$\xspace}
\newcommand{\ba}{$\texttt{ba}$\xspace}
\newcommand{\cif}{$\texttt{CIFAR-10}$\xspace}
\newcommand{\slashdot}{$\texttt{slashdot}$\xspace}
\newcommand{\er}{$\texttt{er}$\xspace}
\newcommand{\pokec}{$\texttt{pokec}$\xspace}

\newcommand{\uni}{\mathcal{U}}
\newcommand{\nmon}{\textsc{SMCC}\xspace}

\newcommand{\fkk}{\textsc{FKK}\xspace}

\newcommand{\aefns}{\textsc{AEFNS}\xspace}
\newcommand{\lrvz}{\textsc{LRVZ}\xspace}

\newcommand{\qs}{\textsc{Linear\-Stream}\xspace}
\newcommand{\qo}{\textsc{Linear\-UncMax}\xspace}
\newcommand{\lc}{\textsc{Linear\-Card}\xspace}
\newcommand{\uu}{\textsc{UncMax\-Update}\xspace}

\newcommand{\qss}{\textsc{LS}\xspace}
\newcommand{\qssp}{\textsc{LS+}\xspace}
\newcommand{\mpl}{\textsc{MultiPass\-Linear}\xspace}
\newcommand{\mpll}{$\textsc{MPL}$\xspace}

\newcommand{\qsmem}{\oh{ k \log(k) \log( 1 / \epsi )}}

\usepackage{pifont}
\newcommand{\cmark}{\ding{51}}%
\newcommand{\xmark}{\ding{55}}%



%% file: content.tex
\begin{abstract}
  In this work, we
  study the problem of
  finding the maximum value of a non-negative submodular function subject to a limit on the number of items selected,
  a ubiquitous problem that appears in many applications, such as data summarization and nonlinear regression.
  We provide the first deterministic, linear-time approximation algorithms for this problem that do not
  assume the objective is monotone. 
  We present three deterministic, linear-time algorithms:
  a single-pass streaming algorithm with a ratio of $23.313 + \epsilon$, which is the first linear-time streaming algorithm;
  a simpler deterministic linear-time algorithm with a ratio of $11.657$; and a $(4 + O(\epsilon ))$-approximation algorithm.
  Finally, we present a deterministic algorithm
  that obtains ratio of $e + \epsilon$
  in $O_{\epsilon}(n \log(n))$ time, close to the best known expected ratio
  of $e - 0.121$ in polynomial time.
\end{abstract}



\keywords{Submodular Maximization, Deterministic Algorithms, Linear Time}

\section{Introduction} \label{sec:intro}
Within discrete optimization, the submodularity property has
been shown to be a fundamental and useful property.
Intuitively, submodularity captures the idea of
diminishing returns, where the marginal gain in utility decreases
as the set becomes larger. 
Submodular objective functions arise in many learning objectives,
\eg interpreting neural networks \citep{Elenberg2017},
nonlinear sparse regression \citep{Elenberg2018},
among many others (see
\citet{Iyer2020} and references therein). In this work,
we study submodular maximization subject to
a size constraint, defined formally as follows.

\textbf{Submodularity and Problem Definition.} Formally,
a nonnegative, set function $f:2^{\mathcal U} \to \reals$,
where ground set $\mathcal U$ is of 
size $n$, is \textit{submodular}
if for all $S \subseteq T \subseteq \mathcal U$, $u \in \mathcal U \setminus T$, 
$\ff{ T \cup \{ u \} } - \ff{ T } \le \ff{ S \cup \{u \} } - \ff{S}$. 
A function $f$ is \textit{monotone} if $f(S) \le f(T)$ whenever
$S \subseteq T$.
In this work, we study the cardinality-constrained submodular maximization problem (\nmon): given
submodular $f$ and integer $k$, determine $\argmax_{|S| \le k} f(S).$ The function $f$ is not required to be monotone. 
We consider the value query model, in which the function $f$
is available to an algorithm as an oracle that returns, in a single
operation, the value $f(S)$ of any queried set $S$.
Since problem \nmon is NP-hard, we seek approximation algorithms that obtain a performance ratio with respect to an optimal solution.


\textbf{Challenges from Big Data.}
Because of ongoing exponential growth in data size \citep{Mislove2008,Libbrecht2017} over the past decades,
much effort has gone into the design of
algorithms for submodular optimization
with low time complexity, \eg
\citep{Badanidiyuru2014,Mirzasoleiman2014,Buchbinder2015a,Fahrbach2018,Kuhnle2019}.
In addition to the time complexity of an algorithm,
we also consider the number of oracle queries an
algorithm makes or the
query complexity; this information is important
as the function evaluation may be much more expensive
than arithmetic operations. 
Moreover, much effort has also gone into
the design of memory efficient algorithms
that do not need to store all of the data.
In this context, researchers have
studied \emph{streaming algorithms} 
for
submodular optimization \citep{Badanidiyuru2014a,Chakrabarti2015,Chekuri2015,Feldman2018,Mirzasoleiman2018,Alaluf2020a,Haba2020a,Liu2021}. 
A streaming algorithm
takes a constant number of passes 
through the ground set (preferably a single pass) while staying
within a small memory footprint
of $O(k \log n )$, where $k$ is the maximum size
of a solution and $n$ is the size of the ground set.

\textbf{Randomized vs. Deterministic Algorithms.}
Especially in the case that $f$ is non-monotone, randomization
has been shown to be very useful in designing approximation algorithms
for $\nmon$ \citep{Chan2017,Feldman2017,Amanatidis2020}.
Deterministic algorithms for $\nmon$ have been
much less common. In addition to the
theoretical question of the power of deterministic
vs. randomized algorithms, there
are practical disadvantages to randomized algorithms:
an approximation ratio that holds only in expectation
may result in a poor solution with constant probability.
To obtain the ratio with high probability,
$O( \log n )$ repetitions of the algorithm
are typically required, which
may be infeasible or undesirable in practice, especially
in a streaming context.
Moreover, the derandomization of algorithms for submodular
optimization has proven difficult, although
a method to derandomize some algorithms at the cost of a polynomial
increase in time complexity was 
given by \citet{Buchbinder2018}.

\textbf{Prior State-of-the-Art Deterministic Algorithms.}
The fastest
deterministic algorithm of prior literature is the
$4 + \epsi$-approximation in $O_\epsi( n \log k )$ of \citep{Kuhnle2019}.
On the other hand, the
best approximation ratio of a deterministic algorithm
is $e$ in $O(k^3n)$ time of \citep{Buchbinder2018}.
Since the constraint $k$ is not constant (in the worst case, it can be on the order of $n$), neither
of these algorithms are linear-time in the size
of the ground set.
Therefore, in this work we seek to
answer the following questions:
\textit{Q1: Does there exist a deterministic, linear-time algorithm for \nmon with constant approximation factor? If so, what is the best ratio obtainable in linear time? Q2: Does there exist a linear-time, single-pass streaming algorithm for \nmon?}
\begin{table*}[t] \caption{Single-pass streaming algorithms. Two entries for \citet{Alaluf2020a} are shown with different choice of post-processing algorithm, one to give it its best possible ratio and the other to give its fastest runtime.} \label{table:cmp}
\begin{center} 
\begin{tabular}{llllll} \toprule
Reference  &  Ratio & Deterministic? &  Time & Passes & Memory\\ 
  \midrule
  \citep{Alaluf2020a}+\citep{Buchbinder2016} & $3.597 + \epsi$ & \xmark &  $\oh{ \frac{1}{\epsi}\log \left( \frac{k}{\epsi} \right) \left( \frac{n}{\epsi}  + \text{poly}\left( \frac{k}{\epsi}\right) \right) }$ & 1 & $O( k / \epsi^2 )$ \\
  \citep{Alaluf2020a}+\citep{Kuhnle2019} & $5 + \epsi$ & \cmark & ${\scriptstyle \oh{ \frac{1}{\epsi}\log \left( \frac{k}{\epsi} \right) \left(\frac{n}{\epsi}  + \frac{k}{\epsi^2} \log \left( \frac{k}{\epsi^2} \right) \right)} }$ & 1 & $O( k / \epsi^2 )$ \\
  \citep{Liu2021} & $e + \epsi + o(1)$ & \xmark &  $\oh{ \frac{n}{\epsi^3}k^{2.5}}$ & 1 & $O(k/\epsi)$ \\ \midrule 
  \qs, Alg. \ref{alg:quickstream}  & $\ratioapprox$ & \cmark & ${O(n)}$ & 1 & $\oh{ k \log(k)\log\left( \frac{1}{\epsi} \right) }$ \\
  \bottomrule
\end{tabular}
\end{center}
\end{table*}
\begin{table*} \caption{Comparison with state-of-the-art algorithms in terms of ratio and time complexity.} \label{table:cmp2}
\begin{center} 
\begin{tabular}{llllll} \toprule
Reference  &  Ratio & Deterministic? &  Time \\
  \midrule
  Fastest, Randomized \citep{Buchbinder2015a} & $e + \epsi$ & \xmark & $\oh{ \frac{n}{\epsi^2} \log \left( \frac{1}{\epsi} \right) }$ \\ 
  Best Ratio, Randomized \citep{Buchbinder2016} & $e - 0.121$ & \xmark & $O( n^5 )$ \\ \midrule
  Fastest Deterministic \citep{Kuhnle2019} & $4 + \epsi$ & \cmark & $O\left( \frac{n}{\epsi} \log\left( \frac{k}{\epsi}\right) \right)$ \\
  Best Ratio, Deterministic \citep{Buchbinder2018} & $e$ & \cmark & $O(k^3 n)$ \\ 
  \midrule
  \lc, Alg. \ref{alg:linear-card} &  $\lcratioapprox$ & \cmark & $O(n)$ & \\
  \lc + \mpll, Alg. \ref{alg:boostratio}  & $4 + O( \epsi )$ & \cmark & $\oh{\frac{n}{\epsi} \log \left( \frac{1}{\epsi} \right)}$ \\
  Derandomized $\lrig$, Alg. \ref{alg:lrtig}  & $e + \epsi$ & \cmark & $O( \epsi^{-2/\epsi - 1} n \log (k) ) $ \\
  \bottomrule
\end{tabular}
\setcounter{footnote}{0}
\end{center}
\end{table*}

\textbf{Contributions.}
Our first contributions are
deterministic, linear-time
approximation algorithms for
\nmon. The first,
\lc, is
a linear-time, deterministic
approximation algorithm 
with ratio at most $\lcratioapprox$.
Once we have an initial approximation in
linear time, we use it as a subroutine in our
approximation algorithm \mpl, which obtains ratio
$4 + O( \epsi )$ in deterministic, linear time.
As its name suggests, \mpl is in addition a multi-pass
streaming algorithm.

Second, we provide
the first linear-time, single-pass streaming algorithm \qs
  for $\nmon$, with ratio $\ratioapprox$.
  The algorithm \qs requires as
  a subroutine a
  deterministic,
  linear-time algorithm for the unconstrained maximization problem
  that can update its solution in constant time to maintain
  a competitive ratio to the offline optimal.
  Since no algorithm satisfying these properties exists in the literature,
  we also provide \qo, a 4-competitive algorithm satisfying these properties.

  Finally, to obtain an algorithm closer to the
  best known deterministic ratio of $e$,
  we develop the
  deterministic algorithm \lrig, which
  obtains ratio $e + \epsi$ in time $O_{\epsi}(n \log n)$,
  for any $\epsi > 0$. This algorithm is a novel interpolation between
  the standard greedy algorithm \citep{Nemhauser1978}
  and the RandomGreedy algorithm of \citet{Buchbinder},
  as described further below; it
  is our only 
  superlinear-time algorithm. Although this is a significant
  theoretical improvement (from $O(n^4)$ of \citep{Buchbinder2018} to $O_\epsi(n \log n)$ with nearly the same ratio),
  the dependence of the runtime
  on the constant $\epsi$ is exponential, making
  our \lrig algorithm impractical. 

  Table \ref{table:cmp} shows how our algorithms compare
  theoretically to state-of-the-art streaming
  algorithms designed for \nmon, and Table
  \ref{table:cmp2} compares to state-of-the-art
  algorithms in terms of runtime or approximation ratio. 
  An empirical evaluation in Section \ref{sec:exp} shows improvement
  in query complexity and solution value of both our
  single-pass streaming algorithm 
  over the current state-of-the-art streaming algorithms
  on two applications of \nmon.
  

\subsection{Related Work} \label{sec:rw}
Because of the vast literature on submodular optimization,
we focus on the most
closely related works to ours. 

\textbf{The Single-Pass Algorithm of \citep{Kuhnle2021}.}
Our streaming algorithm \qs
may be viewed  as a generalization of the algorithm of
\citep{Kuhnle2021} to non-monotone submodular functions;
this generalization is accomplished by maintaining two
  disjoint candidate solutions $X$ and $Y$
  that compete for elements. The loss due to non-monotonicity can then be bounded using the
  inequality $f( S \cup X ) + f ( S \cup Y ) \ge f(S)$ for any set $S$, which follows from submodularity, nonnegativity, and the fact that $X \cap Y = \emptyset$. This strategy of managing the non-monotonicity has been used in the context of greedy
  algorithms previously \citep{Kuhnle2019,Feldman2020}.
  In addition to the above strategy,
  it becomes necessary to use an unconstrained maximization algorithm on the candidate sets $X, Y$; this is needed since deletion of elements from the set may cause the function value to increase due to non-monotonicity. To the best of our knowledge, this is a novel use of unconstrained maximization and requires a new procedure that can update its solution in constant time as discussed above.

\textbf{StandardGreedy and RandomGreedy.} The standard greedy algorithm
was analyzed by \citet{Nemhauser1978} and shown to obtain a $e/(e - 1) \approx 1.582$
approximation ratio for $\nmon$ when $f$ is monotone. Later, this ratio
was shown to be the best possible under the value query model \citep{Nemhauser1978a}.
Unfortunately, the non-monotone case of $\nmon$
is more difficult, and the standard
greedy algorithm may perform arbitrarily badly.
The RandomGreedy algorithm was introduced
by \citet{Buchbinder} and is a typical example
of how randomization can help non-monotone
algorithms. Instead of selecting an element
with the best marginal gain as StandardGreedy
does, RandomGreedy chooses a uniformly random
element from the top $k$ marginal gains.
\citet{Buchbinder} suggested that RandomGreedy
is a natural replacement for StandardGreedy
since it obtains the same ratio of $e/(e - 1)$
(in expectation) for monotone functions, but
also obtains a ratio of $e$ in expectation
for non-monotone $\nmon$. Later, \citet{Buchbinder2018}
derandomized RandomGreedy at the cost of additional
time, to obtain the ratio of $e$ in time $O(k^3 n)$.
To the best of our knowledge, this is the only
deterministic algorithm that obtains a ratio of $e$
for $\nmon$.

Our algorithm $\lrig$ is an interpolation between these two algorithms,
  the standard greedy algorithm \citep{Nemhauser1978} and the RandomGreedy algorithm \citep{Buchbinder}.
  Each of them may be recovered at certain parameter settings as discussed in Section \ref{sec:lrig}.
  Certain desirable properties of each algorithm
  are retained by the interpolation;
  namely, 1) $\lrig$ can be derandomized and sped up with a decreasing
thresholds approach, as StandardGreedy can be \citep{Badanidiyuru2014};
and 2) $\lrig$ obtains nearly ratio $e/(e-1)$ for monotone and $e$
for non-monotone $\nmon$, as RandomGreedy does.
To create it,
  we use as a subroutine an \ig algorithm; this
  is a generalization of the
  algorithm 
  and analysis of
  \citet{Kuhnle2019} from two greedy procedures to $\ell$ greedy procedures, where $\ell$
  is a constant. 

\textbf{The \ig Algorithm of \citep{Kuhnle2019}.}
The \ig algorithm of \citep{Kuhnle2019} maintains
two disjoint sets, each of which is the solution
of a greedy procedure. Each greedy procedure takes
a turn choosing an element into its respective set,
and then yields to the other procedure, until
both sets are of size $k$. The better of the two
sets is returned. This algorithm was shown
to obtain a $4$-approximation for \nmon. As described in
Section \ref{sec:lrig}, we generalize this
algorithm and analysis to $\ell$ greedy procedures,
each with its own disjoint candidate solution.
We show that after each set has size only $k/ \ell$,
we have nearly an $\ell$-approximation.
The generalized \ig is an important subroutine
for our \lrig algorithm. 

\textbf{Single-Pass Streaming: Adversarial Order.} \citet{Alaluf2020a} introduced a 
single-pass streaming algorithm that obtains
ratio $1 + \alpha + \epsi$, where $\alpha$ is the ratio of
an offline post-processing algorithm $\mathcal A$ for \nmon
with time complexity $\mathcal T (\mathcal A, m )$ on an input
of size $m$.
The time complexity of their algorithm is
$O( (\log (k / \epsi) / \epsi ) \cdot ( n / \epsi + \mathcal{T}(\mathcal{A}, k / \epsi ))$.
The currently best offline ratio that may be used for $\mathcal A$
is the $2.597$ algorithm of \citet{Buchbinder2016}, which
yields ratio $3.597 + \epsi$ in expectation
for \citet{Alaluf2020a} in polynomial time.
This is the state-of-the-art ratio for single-pass streaming under no assumptions on
the stream order. If the $4 + \epsi$ algorithm
of \citet{Kuhnle2019} is used for post-processing, the resulting algorithm
is a deterministic, single-pass algorithm with time complexity
$\oh{ \left(\frac{n}{\epsi^2}  + \frac{k}{\epsi^3}\right) \log \left( \frac{k}{\epsi} \right)}$
and ratio $5 + \epsi$; this is the
state-of-the-art time complexity for a single-pass streaming algorithm.
While we do not improve on the state-of-the-art ratio for a single-pass
algorithm in this paper,
we improve the state-of-the-art time complexity to $O(n)$ with
\qs.

\textbf{Single-Pass Streaming: Random Order.}
To the best of our knowledge
the only algorithm for the general case under
random stream order is that of
\citet{Liu2021}. Their algorithm achieves
ratio $e + \epsi + o(1)$ in expectation
with time complexity $\oh{ nk^{2.5} / \epsi^3 }$.
We compare with this algorithm empirically
in Section \ref{sec:exp} and find that
due to the large numbers of queries involved, this
algorithm only completes on very small instances.
\subsection{Preliminaries}
An alternative characterization of submodularity is
the following: $f$ is submodular iff. $\forall A, B \subseteq \mathcal U$, $f(A) + f(B) \ge f(A \cup B) + f( A \cap B )$.
We use the following notation of the marginal
gain of adding $T \subseteq \mathcal U$ to set $S \subseteq \mathcal U$: $\marge{T}{S} = \ff{S \cup T} - \ff{S}$. For element $x \in \mathcal U$, $\marge{x}{S} = \marge{ \{x\} }{S}$.

\textbf{Competitive Ratio.}
\citet{Buchbinder2015b} defined
a notion of online algorithm for
submodular optimization problems,
in which an algorithm must maintain an
(approximate) solution under dynamic
changes in the problem instance.
Our algorithms do not formally fit into
the notion of online algorithm for submodular
optimization defined by \citet{Buchbinder2015b}.
However, some of our algorithms do maintain
a \textit{competitive ratio} with respect to
the optimal solution. An algorithm has competitive
ratio $\alpha$ for problem $\Pi$ if, after 
having received elements $\{e_1,\ldots,e_l\}$,
the algorithm maintains solution $S$ such that
$\alpha f(S) \ge \opt_{\{e_1,\ldots,e_l\}}$,
where $\opt_X$ is the solution of $\Pi$
restricted to ground set $X$. 
In contrast, an approximation ratio only
ensures $\alpha f(S) \ge \opt_{\mathcal U}$,
where $S$ is the set returned after the algorithm terminates.

\textbf{Unconstrained Maximization.}
Given a submodular function $f$,
the \textit{unconstrained maximization problem} (\unc) is
to determine $\argmax_{S \subseteq \uni} f(S)$. This problem
is also NP-hard; and a $(2 - \epsi)$-approximation requires
exponentially many oracle queries \citep{Feige2011a}. \citet{Naor2012}
gave a $2$-approximation algorithm for \unc in linear time.
Our algorithms require a deterministic, linear-time algorithm
for \unc that can update its solution on receipt of a new element
in constant time. \citet{Buchbinder2015b} give two online algorithms
for \unc with competitive ratios $4$ and $e$,
but these algorithms do not meet our requirements.
Therefore, we present \qo in Section \ref{sec:unc}.

\textbf{Organization.} Our linear-time algorithms
can be viewed as starting
with a simple, linear-time algorithm for unconstrained maximization (Section \ref{sec:unc});
and adding increasingly
sophisticated components to specialize the algorithm to
cardinality constraint (Section \ref{sec:card}) and the
single-pass streaming setting (Section \ref{sec:qs}). 
In Section \ref{sec:brlarge}, we
leverage our existing constant factor algorithms
to get ratio $4$ in linear-time via
a multi-pass streaming algorithm.
In Section \ref{sec:lrig}, we give our
nearly linear-time algorithm with ratio
$e + \epsi$.
Finally, we empirically evaluate
our single-pass algorithm in Section \ref{sec:exp}. Proofs omitted from the main
text are provided in the Appendices. 
\section{Linear-Time Algorithm for \unc with Competitive Ratio} \label{sec:unc}
In this section, we present a simple, $4$-competitive linear-time algorithm
\qo for \unc. In addition to serving as the conceptual starting point for
\lc and \qs, \qo is required by our single-pass algorithm \qs as a
subroutine. Omitted proofs are provided in Appendix \ref{apx:unc}.
\begin{algorithm}[H]
   \caption{The $4$-competitive linear-time algorithm for \unc.}\label{alg:online-unc}
   \begin{algorithmic}[1]
     \Procedure{\qo}{$f$}
     \State \textbf{Input:} oracle $f$ 
     \State $X \gets \emptyset$, $Y \gets \emptyset$
     \For{ element $e$ received } \label{alg:unc-for}
     \State {$S \gets \underset{S \in \{X, Y\}}{\argmax} \delta_e ( S )$}\label{line:cmp-unc}
     \If{$\delta_e (S) > 0$}\label{line:add-unc}
     \State $S \gets S \cup \{ e \}$ 
     \EndIf
     \State $S' \gets \argmax \{ f(X), f(Y) \}$ \label{line:end-update}
     \EndFor
     \State \textbf{return} $S'$
     \EndProcedure
\end{algorithmic}
\end{algorithm}

\textbf{Algorithm Overview.} The algorithm (Alg. \ref{alg:online-unc}) maintains two candidate solutions $X$ and $Y$, which are initially empty. As each element $e$ is received, it is added to the set to which it gives the largest marginal gain, as long as such marginal gain is non-negative. Let $X_l, Y_l$ denote the value of $X,Y$, respectively, after receipt of $\{e_1,\ldots,e_l\}$.  Below, we show a 
competitive ratio of $4$ to the maximum on the set of elements
received thus far; \ie
\begin{equation*}
\begin{aligned}
4\max\{f(X_l),f(Y_l)\} \ge \max_{S \subseteq \{e_1, \ldots, e_l \}} f(S).
\end{aligned}
\end{equation*}
Although there exist other algorithms in the literature for \unc that maintain a competitive ratio \citep{Buchbinder2015b}, our algorithm \qs 
requires a subroutine that updates its solution in
constant time upon receipt of a new element, which no
algorithms in the literature satisfy. 
\begin{theorem} \label{thm:qo}
  \qo is a deterministic, linear-time algorithm for \unc
  with competitive ratio $4$, which runs in linear time in the number of received elements.
\end{theorem}
\textbf{Proof Overview.} The main idea is to maintain two candidate
solutions $X,Y$ that are disjoint. These candidate solutions bound the
gain of adding optimal elements to one set by the value of the other.
For example, consider any element $o$ of the
optimal solution $O$ that was not added to $X$ and gives a positive
gain to $X$. Because
of submodularity, $o$ must have been added to $Y$
and the gain of adding $o$ to $X$ is bounded by the actual gain
received when adding it to $Y$. In this way, we can
bound $f( O \cup X ) - f(X)$ by $f(Y)$; and similarly, we can bound
$f( O \cup Y ) - f(Y)$ by $f(X)$. Because of submodularity,
one of $f(O \cup X), f( O \cup Y)$ must be at least $f(O) / 2$,
which gives the result.


\section{Linear-Time Algorithm for \nmon} \label{sec:card}
In this section, we present a $\lcratioapprox$-competitive, linear-time algorithm
\lc for \nmon. This algorithm answers the question $Q1$ above affirmatively,
as we have given a deterministic, linear-time algorithm for \nmon with
constant approximation ratio. 
\begin{algorithm}[H]
   \caption{The $\lcratio$-competitive algorithm for \nmon.}\label{alg:linear-card}
   \begin{algorithmic}[1]
     \Procedure{\lc}{$f,b$}
     \State \textbf{Input:} oracle $f$, $b > 0$
     \State $X \gets \emptyset$, $Y \gets \emptyset$
     \For{ element $e$ received } \label{alg:linear-card-for}
     \State {$S \gets \underset{S \in \{X, Y\}}{\argmax} \delta_e ( S )$}\label{line:cmp-linear-card}
     \If{$\delta_e (S) > bf(S)/k$}\label{line:add-linear-card}
     \State $S \gets S \cup \{ e \}$
     \EndIf
     \State $X' \gets \{ k \text{ elements most recently added to } X \}$. 
     \State $Y' \gets \{ k \text{ elements most recently added to } Y \}$. 
     \EndFor
     \State \textbf{return} $S' \gets \argmax \{ f(X'), f(Y') \}$ \label{line:secondtolast-unc}
     \EndProcedure
\end{algorithmic}
\end{algorithm}

\textbf{Algorithm Overview.}
The algorithm has a strategy similar to \qo, with two differences. First, an element is only added to $S \in \{X,Y\}$ if its marginal gain is at least a threshold of $\tau = bf(S)/k$, where $b$ is a parameter and $k$ is the cardinality constraint; in the unconstrained version, this threshold was $0$. Second, the sets $X$, $Y$ are infeasible
in general. Therefore, instead of returning $\max\{ f(X), f(Y) \}$,
the algorithm instead considers the last $k$ elements added to $X$ or $Y$: $X'$ and $Y'$, respectively. The value of $\tau$ ensures that a constant fraction of
the value of the set has accumulated in the last $k$ elements. 
Both of theses changes result in a loss of approximation ratio as compared to
the unconstrained version. 
A competitive ratio is maintained, but the ratio
has worsened from $4$ to $\le\lcratioapprox$ which is achieved with $b = \sqrt{2}$.
However, in Section \ref{sec:brlarge}, we show how to improve any constant
ratio to $\approx 4$ in linear time. 
\begin{theorem} \label{thm:lc}
  \lc is a deterministic, linear-time algorithm for \nmon
  with competitive ratio $\lcratio$, which runs in linear time in the number of received elements.
\end{theorem}
\textbf{Proof Overview.} The strategy is similar to that for \qo:
two candidate, disjoint sets $X,Y$ are maintained and each is
used to bound the distance of the other
from $f( O \cup S )$, for $S \in \{ X, Y \}$.
However, the
mininum gain of $\tau = bf(S)/k$ for adding an element to
$S \in \{X,Y \}$ is important to bound both 1) the loss of value
of elements of the optimal solution $O$ that are added to neither set; and 2)
the loss of value from discarding all but the last $k$ elements added to
$X$ or $Y$, which is needed to obtain a feasible solution.
The value of $b$ balances these two competing interests against one another.

\section{Single-Pass Streaming Algorithm for \nmon} \label{sec:qs}
In this section, a linear-time, constant-factor algorithm
is described. This algorithm (\qs, Alg. \ref{alg:quickstream})
is a deterministic
streaming algorithm that
makes one pass through the ground set and two queries to $f$
per element received.

\textbf{\qs Overview.} \label{sec:qs-base}
\begin{algorithm}[t]
   \caption{A single-pass algorithm for \nmon.}\label{alg:quickstream}
   \begin{algorithmic}[1]
     \Procedure{\qs}{$f, k, \epsi, b$}
     \State \textbf{Input:} oracle $f$, cardinality constraint $k$, $\epsi > 0$, $b > 0$
     \State $\alpha \gets \ratioPrime$
     \State $\ell \gets \lceil \log ( (6\alpha)/\epsi + 1 )\rceil  + 3$\label{line:ell}     
     \State $A \gets \emptyset$, $B \gets \emptyset$, $\tau \gets \text{um}_A \gets \text{um}_B \gets f(\emptyset)$
     \For{ element $e$ received } \label{alg:qs-for}
     \State {$S \gets \underset{S \in \{X, Y\}}{\argmax} \delta_e ( S )$}\label{line:cmp}
     \If{$\delta_e (S) \ge b \tau / k$}\label{line:add}
     \State $S \gets S \cup \{ e \}$
     \State $\text{um}_S = \, \uu( e )$
     \If{$\max\{\text{um}_S, f(S)\} > \tau$}
     \State $\tau \gets \max\{\text{um}_S, f(S)\}$
     \EndIf
     \EndIf
     \If{$|S| > 2 \ell (k/b + 1) \log_2 (k)$}\label{alg:qs-delete}
     \State $S \gets \{ \ell (k/b + 1)\log_2 (k)$ elements most recently added to $S \}$\label{line:delete-A}
     \State $\text{um}_S \gets \uncalg( S )$
     \State $\tau \gets \max\{ f(A), f(B), \text{um}_A, \text{um}_B \}$.
     \EndIf
     \EndFor
     \State $A' \gets \{ k \text{ elements most recently added to } A \}$. 
     \State $B' \gets \{ k \text{ elements most recently added to } B \}$. 
     \State \textbf{return} $S' \gets \argmax \{ f(A'), f(B') \}$ \label{line:secondtolast}
     \EndProcedure
\end{algorithmic}
\end{algorithm}
The starting point of the algorithm is \lc (Alg. \ref{alg:linear-card}, Section \ref{sec:card});
several modifications are needed to ensure
the algorithm stays within $O_\epsi (k \log n)$
space. First, we add a deletion procedure
on Line \ref{line:delete-A}. The intuition is that
if the size of $A$ (resp. $B$) is large, then
because the threshold required to add elements
on Line \ref{line:add} depends on $f(A)$,
the value of the initial
elements is small. Therefore, deleting these elements
can cause only a small loss in the value of $f(A)$.
However, because $f$ may be
non-monotone, such deletion may
actually cause an \textit{increase}
in the value of $f(A)$, which interferes
with the concentration of value of $A$ into its last
$k$ elements.
Therefore, to ensure enough value accumulates in
the last $k$ elements, we need to ensure that
each addition adds $\Omega( \um (A) / k )$ value,
where $\um (A)$ is the solution to the \unc problem
with the domain of the function $f$ restricted to $A$.

To ensure each addition adds $\Omega( \um (A) / k )$ value,
we require a $\beta$-competitive algorithm for \unc;
and to ensure our algorithm stays linear-time, we need
to be able to update the estimate for $\um(A)$
to $\um(A \cup \{e \})$ in constant
time. Therefore, we need algorithms
\uncalg and \uu such that 1) both algorithms are
deterministic; 2) \uncalg is linear-time; 3) \uu
is constant time; 4) a $\beta$-competitive estimate
of $\um(A)$ is maintained. Observe that using \qo (Alg. \ref{alg:online-unc}, Section \ref{sec:unc}) for \uncalg; and
Lines \ref{line:cmp-unc}--\ref{line:end-update} of Alg. \ref{alg:online-unc} for \uu,
all of the above requirements are met with $\beta = 4$.

\textbf{Theoretical Guarantees.} Next,
  we prove the following theorem concerning the
  performance of \qs (Alg. \ref{alg:quickstream}).
  With $\beta = 4$, the ratio is optimized to $12 + 8\sqrt{2} + \epsi \ge \ratioapprox$ at $b = 2\sqrt{2}$.
\begin{theorem} \label{thm:qs}
  Let $\epsi, b \ge 0$, and let $(f,k)$ be an instance of \nmon;
  and suppose \uncalg and \uu satisfy the requirements discussed
  in Section \ref{sec:qs-base}. Then
  the solution $S'$ returned by $\qs(f,k,\epsi,b)$
  satisfies
  $$\opt \le \left( \ratio \right) f(S').$$
  Further, \qs has time complexity $O(n)$,
  memory complexity $\qsmem$, and makes one pass
  over the ground set.
\end{theorem}
\begin{proof}[Proof of Theorem \ref{thm:qs}]
  The time and memory complexities of \qs
  are immediate,
  so we
  focus on the approximation ratio.
  The first lemma (Lemma \ref{lemm:general})
  establishes basic facts about the
growth of the value in the sets $A$ and $B$ as elements
are received. Lemma \ref{lemm:general}
considers a general sequence of elements that satisfy
the same conditions on addition and deletion as elements of
$A$ or $B$, respectively.
The proof is deferred to Appendix \ref{apx:lemm-gen}
and depends on a condition 
to add elements and uses submodularity
of $f$ to bound the loss in value due to periodic deletions.
\begin{lemma} \label{lemm:general}
  Let $(c_0,\ldots,c_{m-1})$ be a sequence of elements,
  and $(C_0, \ldots, C_m)$ a sequence of sets, such that
  $C_0 = \emptyset$, and
  $C^+_i = C_i \cup \{ c_i \}$ satisfies 
  $f(C^+_i) \ge (1 + b/k)f(C_i)$,
  and $C_{i + 1} = C^+_i$, unless $|C^+_i| > \maxSize$,
  in which case $C_{i + 1} = C_i^+ \setminus C_j$,
  where $j = i - \ell(k / b + 1)\log_2(k)$.
  Then 1) $\ff{C_{i + 1}} \ge \ff{C_i}$, for any $i \in \{0, \ldots, m-1\}$; and
  2) Let $C^* = \{c_0, \ldots, c_{m - 1}\}$. Then $\ff{C^*} \le \left( 1 + \frac{1}{k^\ell - 1} \right) \ff{C_m}$.
\end{lemma}

\textbf{Notation.} Next, we define notation
used throughout the proof.
Let $A_i,B_i$ denote the respective values of variables
$A,B$ at the beginning
of the $i$-th iteration of \textbf{for} loop;
let $A_{n + 1},B_{n+1}$ denote their respective
final values. Also, let $A^* = \bigcup_{1 \le i \le n + 1} A_i$;
analogously, define $B^*$. 
Let $e_i$ denote the element
received at the beginning of iteration $i$. We refer
to line numbers of the pseudocode Alg. \ref{alg:quickstream}.
Notice that after deletion of duplicate entries,
the sequences $(A_n),(B_n)$ satisfy the hypotheses
of Lemma \ref{lemm:general},
with the sequence of elements in $A^*,B^*$, respectively.
Since many of the following lemmata are symmetric with
respect to $A$ and $B$, we state them generically, with variables
$C,D$ standing in for one of $A,B$, respectively. The notations
$C_i, C^*, D_i, D^*$ are defined analogously to $A_i,A^*$ defined
above. Finally,
if $D \in \{A,B\}$, define
$\Delta D_i = \ff{D_{i+1}} - \ff{D_i}.$ Observe that
$\sum_{i = 0}^n \Delta D_i = \ff{D_{n+1}} - \ff{\emptyset}.$

The analyses of both
\lc and \qo above use the fact that
the marginal gain of an element to one set
can be bounded by the increase in value
of the other because of the competition between
the sets
(\ie the comparison on Line \ref{line:cmp}).
If a deletion occurs after the comparison on
Line \ref{line:cmp}, this bound may no longer
hold. The next lemma shows that an approximate form of the bound holds.
\begin{lemma} \label{lemm:two}
  Let $C, D \in \{A,B\}$, such that $C \neq D$.
  Let $o \in O \cap D^*$. Let $\gamma = \frac{1}{k^{\ell} - 1}$.
  Let $i(o)$ denote the iteration in which $o$ was processed.
  Then
  $$\marge{o}{C_{i(o)}} \le \frac{\Delta D_{i(o)} + \gamma \ff{D_{n + 1}} }{1 + \gamma}. $$
\end{lemma}
The next lemma uses Lemma \ref{lemm:two}
to bound the gain of adding the entire
set $O$ into $C^*$.
\begin{lemma} \label{lemm:uni-bound}
  Let $C, D \in \{A,B\}$, such that $C \neq D$. Then
  $$\ff{C^* \cup O} - \ff{C^*} \le b \ff{C_{n+1}} + (1 + k\gamma) \ff{D_{n+1}}$$
\end{lemma}
As in the analysis of \lc, we need
to show a concentration of value in the
last $k$ elements added to our sets. The next
lemma accomplishes this by using that each
element gives a gain of $\Omega( \um( C ) / k )$
by using the $\beta$-competitive procedure for \unc.
 \begin{lemma} \label{lemm:prime}
   Let $C \in \{A,B\}$, and let $C' \subseteq C_{n+1}$
   be the set of $\min\{|C_{n+1}|, k \}$ elements most recently
   added to $C_{n+1}$. 
   Then $\ff{C_{n+1}} \le \left( \ratioPrime \right) \ff{C'} $.
 \end{lemma}
 \begin{proof}
    For simplicity of notation, let $C = C_{n+1}$.
  If $|C| \le k$, the result follows since $C' = C$. 
  So suppose $|C|>k$, and let $C' = \{c_1, \ldots, c_k\}$
  be ordered by the iteration in which each element was
  added to $C$. Also, let $C'_i = \{c_1, c_2, \ldots, c_i\}$,
  for $1 \le i \le k$, and let $C'_0 = \emptyset$.
  Let $C_i$ denote the value of $C$ at the beginning of
  the iteration in which $c_i$ is added. For any set $X$,
  let $\um(X)$ abbreviate $\unc(X)$.

  Observe that
  $C_i \setminus C' \supseteq C \setminus C'$ for all
  $1 \le i \le k$, regardless of whether a deletion occurs
  at any point during the addition of elements of $C'$.
  From this observation, submodularity, and the condition to
  add an element to $C$ on Line \ref{line:add} and the fact
  that \uncalg is a $\beta$-competitive algorithm, we have that
  \begin{align*}
    \ff{C'_i} - \ff{C'_{i-1}} &\ge \ff{(C_i \setminus C) \cup C_i'} 
                              - \ff{(C_i \setminus C) \cup C_{i-1}'} \\
                              & \ge \frac{b}{\beta k} \um( (C_i \setminus C') \cup C'_{i - 1} )
    \ge \frac{b}{\beta k} \um( C \setminus C' ).
  \end{align*}
  Therefore, $f(C') \ge (b/\beta)\um(C\setminus C')$.
  Hence, by submodularity, nonnegativity of $f$, we have
  \begin{align*}
    f(C) &\le f( C \setminus C' ) + \ff{C'} \\
         &\le \um( C \setminus C' ) + \ff{C'} \le (\ratioPrime)\ff{C'}. \qedhere
  \end{align*}
\end{proof}
By application of Lemma \ref{lemm:uni-bound} with $A=C$ and then again with $B=C$, we obtain
  \begin{align}
  \label{ineq:a} \delta_{O}(A^*) &\le b \ff{A_{n+1}} + (1 + k\gamma) \ff{B_{n+1}}, \text{ and} \\
  \label{ineq:b} \delta_{O}(B^*) &\le b \ff{B_{n+1}} + (1 + k\gamma) \ff{A_{n+1}}. 
\end{align}
Next, we have that
\begin{align*} 
  \ff{O} &\le \ff{A^* \cup O} + \ff{B^* \cup O} \numberthis \label{ineq:14}\\
         &\le \ff{A^*} + \ff{B^*} + 
         (b + 1 + k\gamma)(\ff{A_{n+1}} + \ff{B_{n+1}}), \numberthis \label{ineq:15}
\end{align*}
where Inequality \ref{ineq:14} follows from the fact that $A^* \cap B^* = \emptyset$ and submodularity and nonnegativity of $f$. Inequality
\ref{ineq:15} follows from the summation of Inequalities \ref{ineq:a} and \ref{ineq:b}.
By application of Property 2 of 
Lemma \ref{lemm:general}, we have from Inequality \ref{ineq:15}
\begin{align*}
  \ff{O} &\le (b + 2 + (k + 2)\gamma)(\ff{A_{n+1}} + \ff{B_{n+1}}) \\
  &\le (2b + 4 + 2(k+2)\gamma) \ff{C_{n+1}}, \numberthis \label{ineq:16} 
\end{align*}
where $C_{n + 1} = \argmax \{ f(A_{n + 1} ), f( B_{n + 1} ) \}$.
Observe that the choice of $\ell$ on Line \ref{line:ell} ensures that
$2(k + 2)\gamma < \epsi (\ratioPrime)^{-1}$, by Lemma \ref{lemm:ell-choice}.
Therefore, by application of Lemma \ref{lemm:prime}, we have from Inequality \ref{ineq:16}
$$f(O) \le \left( \ratio \right) \ff{C'}.\qedhere$$

\end{proof}
\subsection{Post-Processing: \qssp} \label{sec:qs-post}
In this section, we briefly describe a
modification to \qs that improves its
empirical performance. Instead of choosing,
on Line \ref{line:secondtolast}, the best of $A'$
and $B'$ as the solution; introduce a third
candidate solution as follows: use an offline algorithm
for \nmon in a post-processing procedure on the
restricted universe $A \cup B$ to select
a set of size at most $k$ to return. This method can only
improve the objective value of the returned solution and
therefore does not compromise the theoretical analysis
of the preceding section. The empirical solution value can
be further improved 
by lowering the parameter $b$ as this increases the size
of $A \cup B$, potentially
improving the quality of the solution found by the selected
post-processing algorithm.

\section{Multi-Pass Streaming Algorithm for \nmon} \label{sec:brlarge}
In this section, we describe a multi-pass streaming algorithm
for \nmon that can be used to improve any constant ratio
to $4 + O(\epsi )$ in linear time and $O(k)$ space. 
\begin{algorithm}[H]
   \caption{A multi-pass algorithm for \nmon.}\label{alg:boostratio}
   \begin{algorithmic}[1]
     \Procedure{\mpl}{$f, k, \epsi, \Gamma, \alpha$}
     \State \textbf{Input:} oracle $f$, cardinality constraint $k$, $\epsi > 0$, parameters $\Gamma, \alpha$, such that $\Gamma \le \opt \le \Gamma / \alpha$.
     \State $\tau \gets \Gamma / (4k\alpha)$ \Comment{Choice satisfies $\tau \ge \opt/(4k)$.}
     \While{ $\tau \ge \epsi \Gamma / (16k)$ }
     \For{ $u \in \mathcal U$ }
     \State {$S \gets \underset{S \in \{X, Y\} \text{ and } |S| < k}{\argmax} \delta_e ( S )$}
     \label{line:cmpbr} \Comment{If $\argmax$ is empty, break from loop.}
     \If{$\delta_u (S) \ge \tau$}\label{line:addbr}
     \State $S \gets S \cup \{ u \}$
     \EndIf
     \EndFor
     \State $\tau \gets \tau (1 - \epsi)$
     \EndWhile
     \State \textbf{return} $S \gets \argmax \{ f(A), f(B) \}$ 
     \EndProcedure
\end{algorithmic}
\end{algorithm}

\textbf{Algorithm Overview.}
The algorithm \mpl (Alg. \ref{alg:boostratio})
starts with $\Gamma$, an initial estimate of $\opt$ obtained by running a constant-factor $\alpha$-approximation algorithm. The value of $(\Gamma, \alpha)$
is used to compute an upper bound for $\opt / (4k)$. Then, a fast greedy approach
with descending thresholds is used, in which
two disjoint sets $A$ and $B$ compete
for elements with gain above a threshold.
To obtain our stated theoretical guarantees, \qs is used to obtain the initial value of $\Gamma, \alpha$.
\begin{theorem} \label{thm:br}
  Let $0 \le \epsi \le 1/2$, and let $(f,k)$ be an instance of \nmon,
  with optimal solution value $\opt$.
  Suppose $\Gamma, \alpha \in \mathbb{R}$
  satisfy $\Gamma \le \opt \le \Gamma / \alpha$.
  The solution $S$ returned by $\mpl(f,k,\epsi,\Gamma, \alpha)$
  satisfies
  $\opt \le (4 + 6\epsi) f(S).$
  Further, $\mpl$ has time and query
  complexity $O \left( \frac{n}{\epsi} \log \left( \frac{1}{\alpha \epsi} \right) \right)$,
  memory complexity $O(k)$, and makes
  $O( \log(1/(\alpha \epsi)) / \epsi )$ passes
  over the ground set.
\end{theorem}

\section{Nearly Linear-Time Algorithm for \nmon with $e+\epsi$ Ratio} \label{sec:lrig}
\begin{algorithm}[t!]
    \caption{An $(e+\epsi)$-approximation algorithm for \nmon} \label{alg:lrig}
    \begin{algorithmic}[1]
    \Procedure{\lrig}{$f, k, \epsi$}
    \State \textbf{Input:} oracle $f$, constraint $k$, $\epsi$
    \State \textbf{Initialize } 
    $\ell \gets \frac{2e}{\epsi}+1$,
    $G_0 \gets \emptyset$
    \For{$m \gets 1$ to $\ell$}
      \State $G_m \gets \ig(f, k, \ell, G_{m-1}$)
    \EndFor
    \State \textbf{return} $G_\ell$
    \EndProcedure
\end{algorithmic}
\end{algorithm}
\begin{algorithm}[t!]
    \caption{A $({\approx}\ell)$-approximation that interlaces $\ell$ greedy procedures together and uses only $1/\ell$ fraction of the budget. This is the main subroutine required by $\lrig$. } \label{alg:rig}
    \begin{algorithmic}[1]
    \Procedure{\ig}{$f, k, \ell, G$}
    \State \textbf{Input:} oracle $f$, constraint $k$, set size $\ell$,
    starting set $G$
      \State $\{a_1, \ldots , a_\ell\} \gets$ top $\ell$ elements \label{line:rig-maxl}
      in $\mathcal{U} \setminus G$ with respect to marginal gains on $G$ 
      \For{$u \gets 0$ to $\ell$ in parallel} \label{line:rig-outerfor} \Comment{$\ell+1$ guesses of max singleton in $O$}
      \If{ $u = 0 $}
        \State $A_{u,l} \gets G\cup \{a_l\}$, for all $1\le l \le \ell$
      \Else
      \State $A_{u,l} \gets G \cup \{a_u\}$, for any $1\le l \le \ell$
    \EndIf
    \While{$j \gets 1$ to $k/\ell-1$} \label{line:lrig_j} \label{line:rig-while-begin}
            \For{$i \gets 1$ to $\ell$}
              \State $x_{j,i} \gets \argmax_{
              x \in \mathcal{U} \setminus \parth{\cup_{l=1}^\ell A_{u,l}^j }} 
              \marge{x}{A_{u,i} }$\label{line:max}
              \State $A_{u,i} \gets A_{u,i} \cup \{x_{j,i}\}$
              \EndFor
              \State $j \gets j + 1$
          \EndWhile \label{line:rig-while-end}
      \EndFor
    \State \textbf{return} $A \gets $ a random set in $\{A_{u,i}:1 \le i \le \ell, 0 \le u \le \ell\}$
    \EndProcedure
\end{algorithmic}
\end{algorithm}
\begin{figure}[t]
\centering
\includegraphics[width=0.45\textwidth]{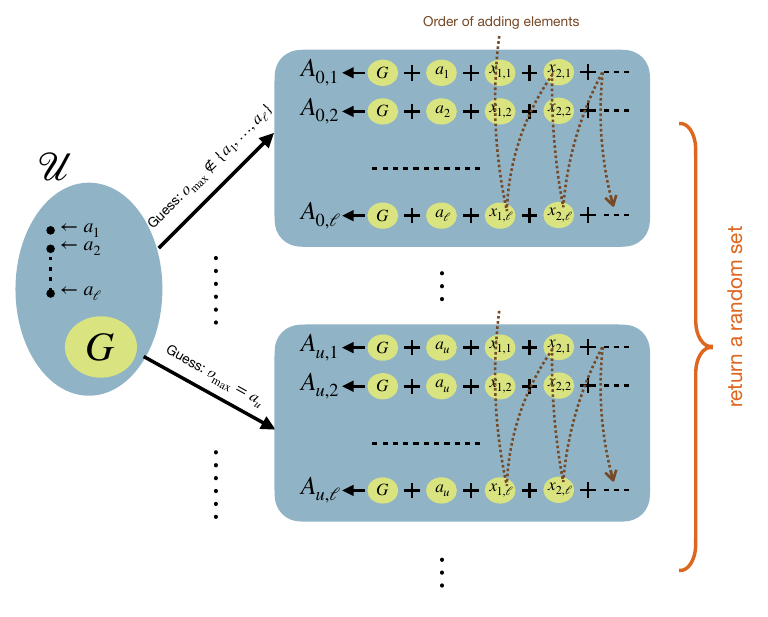}
\caption{This figure depicts the way elements are added to the solution set $A_{u, l}$
in $\ig$, Alg.~\ref{alg:rig}.} \label{fig:rig}
\label{fig:IGframework}
\end{figure}
In this section, we propose a deterministic, nearly linear time
algorithm for \nmon with $e+\epsi$ approximation ratio.
For simplicity, we present and analyze a slower, randomized
version of the algorithm in this section,
which contains the main algorithmic
ideas. We discuss how to derandomize
and speed up the algorithm, with full details provided in
Appendix \ref{sec:derand-alg}.

\textbf{Algorithm Overview.}
The algorithm $\lrig$ (Alg. \ref{alg:lrig}) may be thought of as
an interpolation between
two algorithms in the literature: the standard greedy algorithm
\citep{Nemhauser1978} and the RandomGreedy algorithm of
\citet{Buchbinder} -- see the Related Work section for
a discussion of these algorithms. The degree of interpolation
is controlled by a parameter $\ell$:
If $\ell = 1$, $\lrig$ 
reduces to standard greedy, while if $\ell = k$, it
can be shown that our algorithm
reduces to RandomGreedy.
Since we set $\ell$ to a constant value ($2e/\epsi$),
the algorithm can
considered to be closer to StandardGreedy than RandomGreedy. 

The main idea is to have an algorithm
that we can analyze similarly to
the analysis of RandomGreedy,
except that a constant number $\ell$ of iterations are used,
instead of $k$, as in RandomGreedy.
In order for this
to work, we need a way to add $k / \ell$ elements each
iteration, in such a way that we reduce the
distance to the optimal value\footnote{Precisely, when $f$ is non-monotone, it is not the optimal value but the value of the optimal solution unioned with the current solution of the algorithm.}
by at least a factor of $1/\ell$.
That is, we need an $\ell$-approximation algorithm that
uses only $(1 / \ell)$-fraction of the budget,
which seems like a difficult proposition as this means the algorithm
must use the $(1/\ell)$-fraction of the budget optimally.
If one considers each iteration of the standard greedy algorithm,
it is precisely a $k$-approximation that requires a $(1/k)$-fraction
of the budget -- it is this characteristic of the greedy selection
that must be generalized to selecting a constant fraction of the budget
at once. 
Nevertheless, we show that
the algorithm $\ig$ (Alg. \ref{alg:rig}) is able to satisfy this requirement. 

The operation of $\ig$ is illustrated
in Fig. \ref{fig:rig}. For technical reasons related to the nonmonotonicity of the
function, the actual pool of candidate sets used by $\ig$
consists of $\ell + 1$ pools of $\ell$ candidates,
and there is only a $( \ell + 1 )^{-1}$ probability of success
(that we hit the right pool). This low probability of success
at each iteration leads to a large constant factor
(exponential dependence on $\epsilon$)
in the runtime when
we derandomize the algorithm. Finally,
$\lrig$ works for both monotone and non-monotone
cases: if the function is monotone,
$\lrig$ gets nearly the optimal $e/(e - 1)$ ratio,
as shown in Appendix~\ref{sec:mon-lrig}.

\textbf{Speedup and Derandomization.}
One can use a descending thresholds greedy
approach (a common strategy first seen
in submodular optimization with \citet{Badanidiyuru2014})
to speed up $\ig$, which replaces the factor
of $k$ in the runtime with $O( \log n )$.
Since the randomization is over
the selection of the set from
a pool of $\ell (\ell + 1)$ candidates at each of
the $\ell$ iterations, and $\ell$ is
a constant, there are only $(\ell(\ell + 1))^\ell$
  possible computation paths, which is a constant.
  Hence, the algorithm can be derandomized at
  the cost of a constant factor by
  following each of these paths and selecting
  the best. 

\textbf{Overview of Proof.} The main difficulty of the
proof is showing that $\ig$ is nearly an $\ell$-approximation
(Theorem \ref{thm:ig}).
Once this is established, the proof is similar to
the RandomGreedy analysis, except with $\ell$ iterations
instead of $k$. To show $\ig$ is an $\ell$-approximation,
it is necessary to order the optimal solution $O$ in a
certain way, such that each of the $\ell$ greedy
procedures gets a marginal gain that dominates the
gain of any of the next $\ell$ elements of $O$. If
one of the greedy procedures adds an element of $O$
in the first iteration, this may not be possible;
hence, we have to guess which of the first $\ell$
elements, if any, is the best element that intersects with $O$
and give each greedy procedure the opportunity
to start with this element. This guessing procedure
($\ell + 1$ guesses in total) is responsible for
the low (but constant) probability of success
of each iteration, as only one of the guesses can
be correct. 
\begin{theorem}\label{thm:ig}
  Let $O \subseteq \uni$ be any set of size at most $k$, and suppose
  $\ig$ is called with $(G, f, k, \ell)$.
  Then $\ig$ outputs a set $A$
with $\oh{\ell nk}$ queries and probability $(\ell+1)^{-1}$ such that:
1) $\ex{f(O \cup A)} \ge \left(1-\frac{1}{\ell}\right) f(O \cup G)$;
\\2) $\ex{f(A)} \ge \frac{\ell}{\ell+1} f(G)+
     \frac{1}{\ell+1} \parth{1-\frac{1}{\ell}}f(O \cup G)$.
\end{theorem}
\begin{proof}[Proof of Theorem~\ref{thm:ig}]
  Let $o_{\text{max}} = \argmax_{o \in O\setminus G}\marge{o}{G}$,
  and let $\{ a_1, \ldots, a_\ell \}$ be the largest $\ell$ elements
  of $\{ \marge{x}{G} : x \in \uni \setminus G \}$, as chosen on
  Line \ref{line:rig-maxl}. We consider the following two cases.

\textbf{Case  $(O\setminus G) \cap \{a_1, \ldots, a_\ell\}=\emptyset$.}
Then, $o_{\text{max}} \not \in \{a_1, \ldots, a_\ell\}$ 
which implies that
$\marge{a_u}{G} \ge \marge{o}{G}$,
for all $1 \le u \le \ell$ and $o \in O\setminus G$;
and, after the first iteration of the \textbf{while} loop on Line~\ref{line:lrig_j},
no element of $O \setminus G$ is added into any of
$\{A_{0,i}\}_{i=1}^\ell$. We will analyze the iteration
of the \textbf{for} loop on Line \ref{line:rig-outerfor} with $u = 0$.

Since no element of $O\setminus G$ is added into the collection
when $j=0$, we can order $O\setminus G= \{o_1, o_2,\ldots\}$ such that the first $\ell$ elements
are not selected in any set before we get to $j = 1$,
the next $\ell$ elements are not selected in any set before we get
to $j = 2$,
and so on. Let $i \in \{1, \ldots, \ell \}$.
Let $A_{0,i}^{j}$ be the value of $A_{0,i}$ after $j$ elements are added into it,
and define $A_{0,i} = A_{0,i}^{k/\ell}$, the final value.
Finally, denote by $\delta_j$ the value $\delta_{x_{j, i}}(A_{0,i}^{j})$.
Then,
\begin{align*}
&\func{f}{O \cup A_{0,i}} - \func{f}{A_{0,i}}
\le \sum_{o \in O\setminus A_{0,i}}\marge{o}{A_{0,i}}\tag{submodularity}\\
&\le \sum_{o \in O\setminus G}\marge{o}{A_{0,i}}\tag{$G \subseteq A_{0,i}$}\\
&\le \sum_{l = 1}^{\ell} \marge{o_l}{A_{0,i}^{0}}
+ \sum_{l = \ell+1}^{2\ell} \marge{o_l}{A_{0,i}^{1}}+\ldots \tag{submodularity}\\
&\le \ell \sum_{j=1}^{k/\ell}\delta_j
=\ell (\func{f}{A_{0,i}}-\func{f}{G}),\numberthis \label{inq:ig-rec}
\end{align*}
where the last inequality follows from the ordering of $O$ and the
selection of elements into the sets.
Since, $A_{u,i_1} \cap A_{u,i_2} = G$
for any $1 \le i_1\neq i_2 \le \ell$,
it holds that 
$\parth{O \cup A_{u,i_1}} \cap \parth{O \cup A_{u,i_2}} = O\cup G$.
Then, by repeated application of
submodularity and nonnegativity of $f$, it can be shown that
\begin{equation}\label{inq:ig-deg}
\sum_{i=1}^\ell f(O\cup A_{0,i})\ge (\ell-1)\func{f}{O \cup G} + 
\func{f}{O\cup \parth{\cup_{i=1}^\ell A_{0,i}}}
\ge (\ell-1)\func{f}{O \cup G}.
\end{equation}
By summing up Inequality~\ref{inq:ig-rec} with all $1\le i \le \ell$,
it holds that
\begin{equation}\label{inq:ig-recur}
\frac{\ell+1}{\ell} \sum_{i=1}^\ell f(A_{0,i}) \ge 
\frac{1}{\ell}\sum_{i=1}^\ell f(O\cup A_{0,i}) + \ell f(G)
\ge \frac{\ell-1}{\ell}\func{f}{O \cup G}+\ell \func{f}{G},
\end{equation}
where the second inequality follows from Inequality~\ref{inq:ig-deg}.
Therefore, if we select a random set from $\{A_{0,i}: 1\le i \le \ell\}$,
by Inequalities~\ref{inq:ig-deg} and~\ref{inq:ig-recur},
the two inequalities in the Lemma hold
and we have probability $1/(\ell + 1 )$ of this happening.

\begin{figure}[t]
\centering
\subfigure[\imgsum value, \cif]{ \label{mt:obj-img}
    \includegraphics[width=0.3\textwidth,height=0.165\textheight]{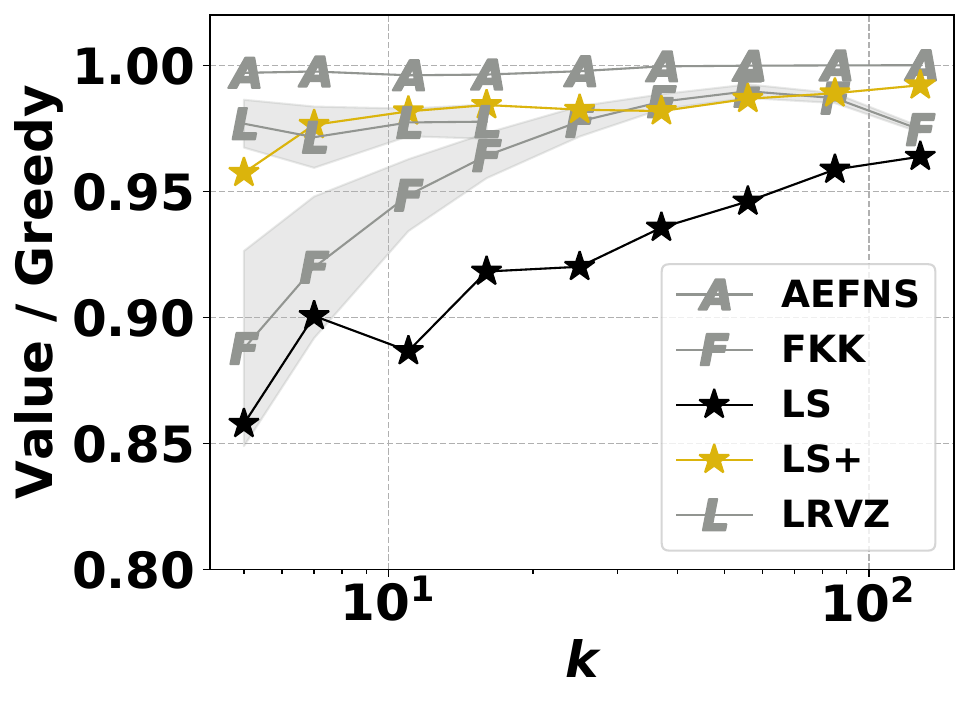}
  }
  \subfigure[\imgsum queries, \cif]{ \label{mt:query-img}
    \includegraphics[width=0.3\textwidth,height=0.165\textheight]{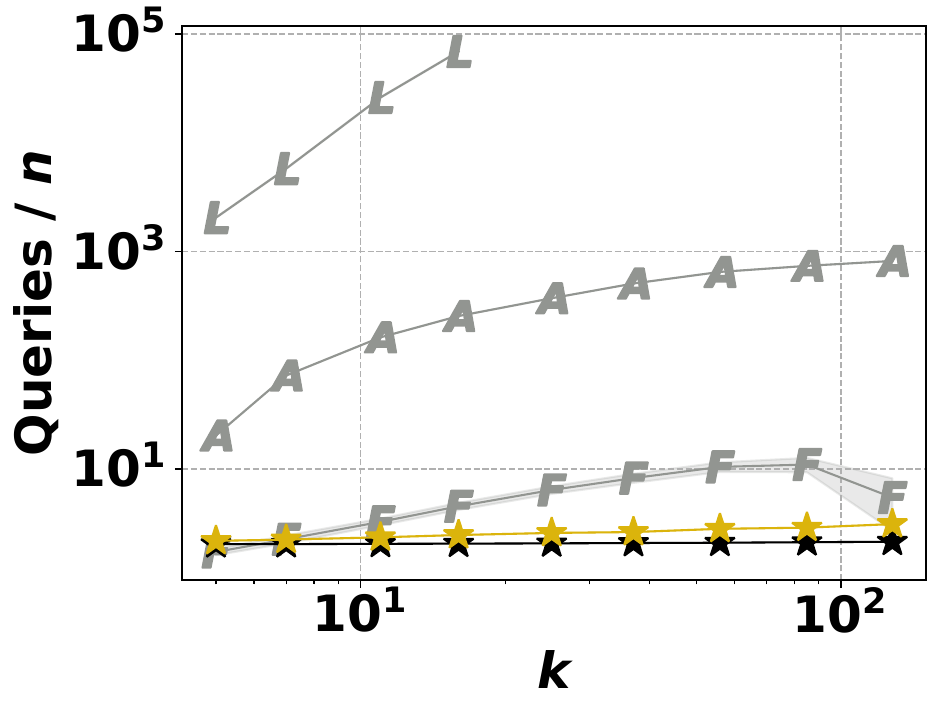}
  }
  \subfigure[\maxcut value, \ba]{ \label{mt:obj-ba}
    \includegraphics[width=0.3\textwidth,height=0.165\textheight]{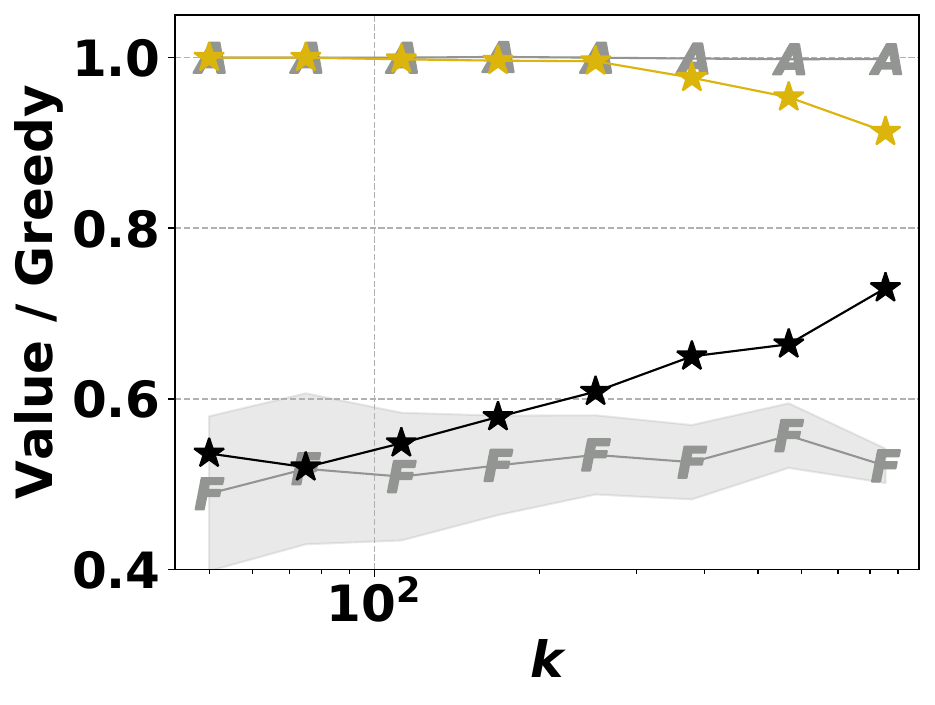}
  }
  \subfigure[\maxcut queries, \ba]{ \label{mt:query-ba}
    \includegraphics[width=0.3\textwidth,height=0.165\textheight]{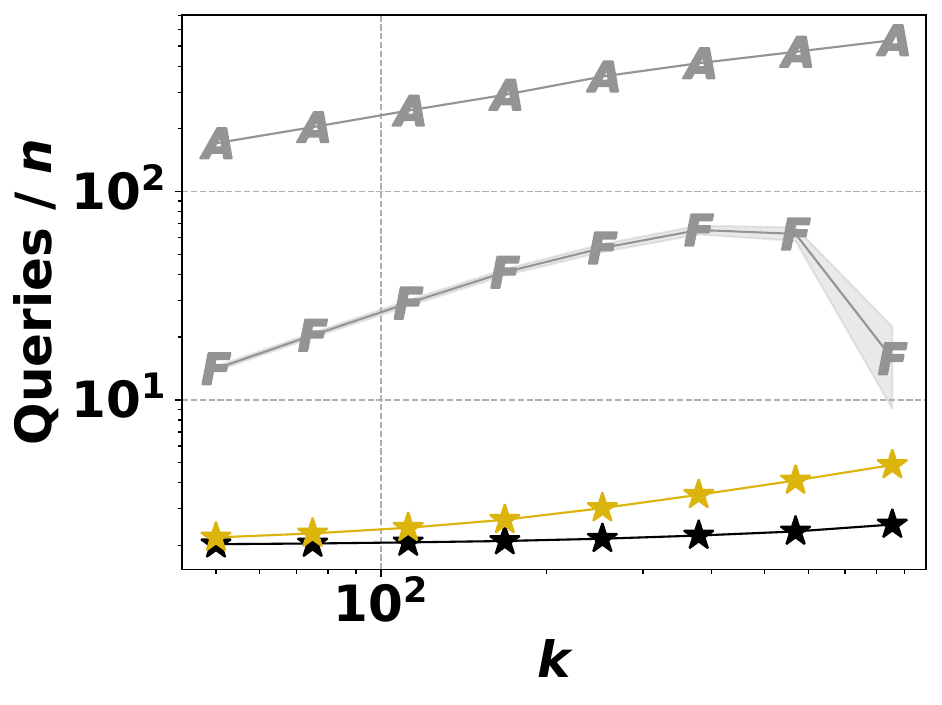}
  }
  \subfigure[\revmax value, \fb]{ \label{mt:obj-fb}
    \includegraphics[width=0.3\textwidth,height=0.165\textheight]{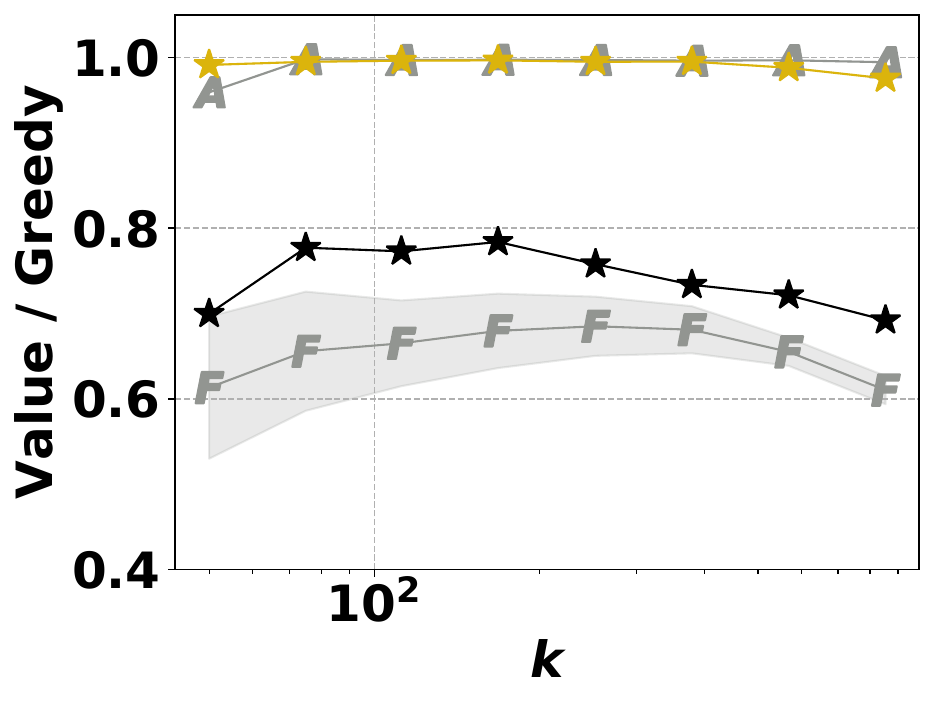}
  }
  \subfigure[\revmax queries, \fb]{ \label{mt:query-fb}
    \includegraphics[width=0.3\textwidth,height=0.165\textheight]{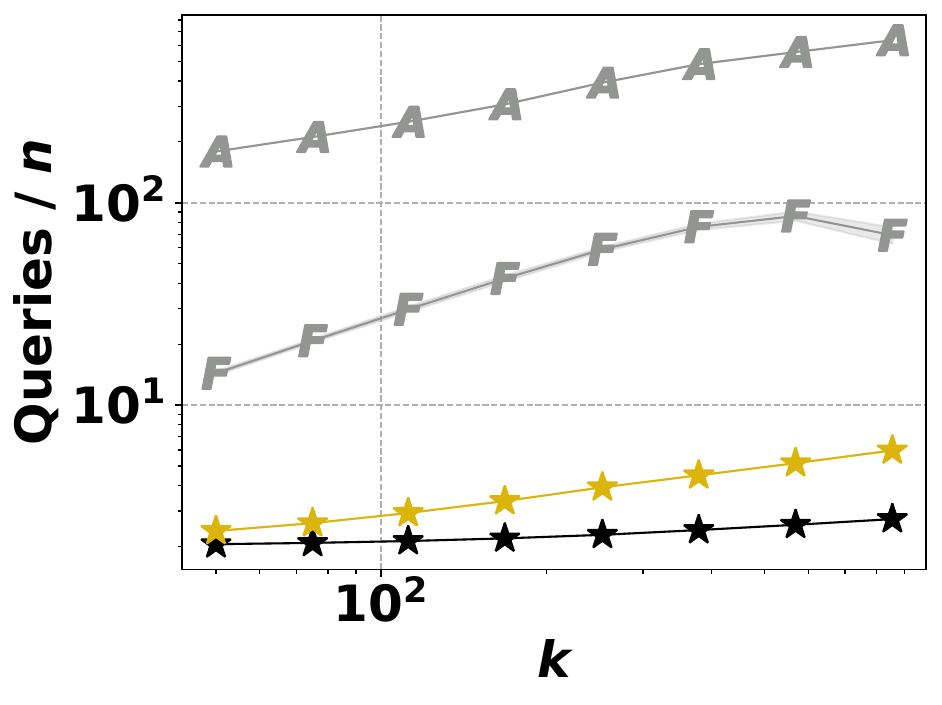}
  }
  \caption{Evaluation of single-pass streaming algorithms on the {\cif, \ba, and \fb datasets with the \imgsum [\textbf{(a),(b)}], \maxcut objectives [\textbf{(c),(d)}], and \revmax [\textbf{(e),(f)}] respectively.
  The legend in \textbf{(a)} applies to all the other subfigures.}}
  \label{fig:single-pass}
\end{figure}
\textbf{Case  $(O\setminus G) \cap \{a_1, \ldots, a_\ell\}\neq\emptyset$.}
Then $o_{\text{max}} \in \{a_1, \ldots, a_\ell\}$,
so $a_u = o_{\text{max}}$, for some $u \in 1, \ldots, \ell$.
we analyze the iteration $u$ of the \textbf{for}
loop on Line \ref{line:rig-outerfor}.
Similarly to the previous case,
let $i \in \uni$, 
define $A_{u,i}^{j}$ be the value of $A_{u,i}$ after we add $j$ elements into it,
and we will use $A_{u,i}$ for $A_{u,i}^{k/\ell}$,
Also, let $\delta_j = \marge{x_{j, i}}{A_{u,i}^{j-1}}$.
Finally, let $x_{1,i} = a_u$ and 
observe $A_{u,i}^{(1)} = G \cup \{a_u\}$.

Then, we can order $O\setminus G= \{o_1, o_2,\ldots\}$
such that: 1) for the first $\ell$ elements $\{o_l\}_{l=1}^\ell$, 
$\marge{o_l}{G} \le \marge{o_{\text{max}}}{G}=\delta_1$;
2) the next $\ell$ elements
$\{o_l\}_{l=\ell+1}^{2\ell}$ are not selected by any set before we get 
to $j = 2$, which implies that $\marge{o_l}{A_{u,i}^{1}} \le \delta_2$,
and so on.
Therefore, analagous to the the previous case,
we have that 
\begin{equation}\label{inq:ig-rec-2}
\func{f}{O \cup A_{u,i}} - \func{f}{A_{u,i}}
\le \ell (\func{f}{A_{u,i}}-\func{f}{G}).
\end{equation}
Since, $A_{u,i_1} \cap A_{u,i_2} = G\cup \{a_u\}$
for any $1 \le i_1\neq i_2 \le \ell$,
and $a_u \in O$,
it holds that 
$\parth{O \cup A_{u,i_1}} \cap \parth{O \cup A_{u,i_2}} = O \cup G$.
Then, by submodularity and nonnegativity of $f$, it holds that
\begin{equation}\label{inq:ig-deg-2}
\sum_{i=1}^\ell f(O\cup A_{u,i}) \ge (\ell-1)\func{f}{O \cup G} + 
\func{f}{O\cup \parth{\cup_{i=1}^\ell A_{u,i}}}
\ge (\ell-1)\func{f}{O \cup G}. 
\end{equation}
By summing up Inequality~\ref{inq:ig-rec-2} with all $1\le i \le \ell$,
it holds that 
\begin{equation}\label{inq:ig-recur-2}
\frac{\ell+1}{\ell} \sum_{i=1}^\ell f(A_{u,i}) \ge 
\frac{1}{\ell}\sum_{i=1}^\ell f(O\cup A_{u,i}) + \ell f(G)
\ge \frac{\ell-1}{\ell}\func{f}{O \cup G}+\ell \func{f}{G},
\end{equation}
where the second inequality follows from Inequality~\ref{inq:ig-deg-2}.
Therefore, if we select a random set from $\{A_{u,i}: 1\le i \le \ell\}$,
by Inequalities~\ref{inq:ig-deg-2} and~\ref{inq:ig-recur-2},
the two inequalities in the lemma holds, and
this happens with probability $(\ell+1)^{-1}$.
\end{proof}
Once we have Theorem \ref{thm:ig}, the analysis
of the next theorem is similar to the RandomGreedy
analysis and is relegated to Appendix \ref{apx:lrig}.
\begin{theorem}\label{thm:lrig}
Let $\epsi \ge 0$, and
$(f,k)$ be an instance of \nmon,
with optimal solution value \opt.
Algorithm $\lrig$ outputs a set $G_\ell$
with $\oh{\epsi^{-2}nk}$ queries such that
$\opt \le \parth{e+\epsi}\ex{f(G_\ell)}$ with
probability $(\ell+1)^{-\ell}$, where $\ell = \frac{2e}{\epsi}+1$.
\end{theorem}

\section{Empirical Evaluation} \label{sec:exp} 
In this section, we evaluate our single-pass
algorithm \qs in two variants:
without post-processing (\qss) and
using our algorithm \mpl for
post-processing (\qssp), as described
in Section \ref{sec:qs-post}. We compare with
1) Algorithm 3 (\lrvz) of \citet*{Liu2021},
  which achieves ratio
  $e + \epsi + o(1)$ in expectation
  in time $\oh{ \frac{n}{\epsi^3}k^{2.5}}$, if the stream is in random order.
  This is the best ratio achieved by a single-pass streaming algorithm,
  although if the stream is in adversarial order it has no ratio.
  Our implementation is an idealized implementation that uses more memory
  to run faster than the actual algorithm does, as discussed
  in Appendix \ref{apx:exp}. This optimization only advantages
  \lrvz in the comparison. 
2) Algorithm 2 (\fkk) of \citet*{Feldman2018};
  this algorithm achieves ratio $5.828$ in expectation
  and has $O(kn)$ time complexity.
3) Algorithm 1 (\aefns) of \citet*{Alaluf2020a};
  the implementation of this algorithm
  requires choice of a post-processing algorithm.
  For fair comparison, \aefns and \qssp used the same post-processing algorithm \mpl as discussed in Appendix \ref{sec:apx-impl}, which for \aefns
  yields
  ratio $5 + \epsi$ and time complexity $O_{\epsi}(n \log k )$.

  Randomized algorithms were repeated 40 times;
  plots show sample mean (symbol) and standard deviation (shaded region) of each metric.
  A timeout of four hours was used for
  each repetition. \lrvz received the stream in uniformly random order for each repetition; all other algorithms used the stream order determined by the data representation.

  
\textbf{Applications and Datasets.}
The algorithms were evaluated on
{three} applications of \nmon: 
cardinality constrained maximum cut
(\maxcut),
revenue maximization on social networks
(\revmax),
{and image summarization (\imgsum).}
A variety of network topologies from the
Stanford Large Network Dataset Collection
\citep{snapnets} were used, as well as
synthetic random graphs.
For more details on the applications and datasets, see
Appendix \ref{apx:exp}
\footnote{The source code is available at 
\url{https://gitlab.com/luciacyx/dtm-linear-code.git}.}.
  
\textbf{Results.}
Results for the objective value (normalized
by the standard greedy value)
and total queries (normalized by the number of vertices $n$ in the graph) for each application
are shown in Fig. \ref{fig:single-pass} as the cardinality
constraint $k$ varies, for the \ba dataset, a synthetic
BA random graph with $n=5000$, the \fb dataset,
a small section of the Facebook social network
with $n=4039$,
{and the \cif dataset, 
a random collection of images with $n=500$. }
These datasets are chosen since
all algorithms are able to complete within the time
limit on at least some
instances. Further exploration of the
scalability of our algorithms is given
in Appendix \ref{apx:exp}, where we show results
on datasets with $n > 10^6$. For the cardinality constraint,
we used a range of $k$ values
{which increased by a factor of 1.5.
The starting value of $k$ for \ba and \fb datasets was set to 50,
while for \cif, it was set to 5.}

\textbf{Discussion.} {On almost all instances}, \qssp returned nearly the greedy value
({typically $\ge 95\%$, except one instance in Fig.~\ref{mt:obj-ba}}) while using 
{ $< 3n$ oracle queries for \cif and \ba datasets,
and $< 6n$ oracle queries for \fb datasets.}
The only algorithm with competitive objective value
to \qssp is \aefns, which 
requires more than an order of magnitude more queries than \qssp.
Even without post-processing,
the objective value of \qss is competitive with 
{\fkk on some instances; see Fig.~\ref{mt:obj-ba} and~\ref{mt:obj-fb}}.
The algorithm \lrvz, despite having the best
theoretical performance ratio,
{returns similar solutions compared to \qssp; see Fig.~\ref{mt:obj-img};}
moreover, \lrvz is
the least scalable algorithm, often requiring more than
$1000n$ queries of the objective. 

Our algorithms \qss and \qssp used the fewest oracle queries on every evaluated
instance of \nmon. Moreover, our algorithms exhibited the best scaling
with the size constraint $k$; \qss has no $k$ dependence at all.
After \qss and \qssp, the next most scalable algorithm on the
evaluated instances is \fkk, which scales linearly with $k$
and used more than an order of
magnitude more queries than our algorithms. 
In addition, \fkk consistently returns the lowest objective value
of any algorithm. 
Finally, the post-processing of \qssp adds only a small
amount of extra queries over \qss, as shown in 
{Figs.~\ref{mt:query-img},}
\ref{mt:query-ba}, and \ref{mt:query-fb}, but
results in a large improvement in objective value.

\section{Conclusion and Future Work}
In this work, we have presented deterministic,
linear-time algorithms for \nmon, which are the first linear-time
algorithms for this problem that yield
constant approximation ratio with high probability.
A natural question for future work
is if the ratio of $4 + \epsi$ could be improved in deterministic,
linear time; can we remove the $\log(n)$ factor of
\lrig ? Also, could the constant factors of \lrig be improved? 
Finally, heuristic improvements to our single-pass algorithm
obtained nearly state-of-the-art objective value empirically
while using fewer oracle queries than any other streaming algorithm,
frequently by an order of magnitude.


%% file: appendix.tex
\clearpage
\onecolumn
\appendix


\section{Proofs for Section \ref{sec:unc}} \label{apx:unc}
\begin{proof}[Proof of Theorem \ref{thm:qo}.]
  Suppose \qo has received elements $\{e_1,\ldots,e_l\}$; and let
  $X, Y$ have their values after processing these elements; let
  $S = \argmax \{ f(X), f(Y) \}$. Let $O \subseteq \{ e_1, \ldots, e_l \}$ satisfy $f(O) = \unc\left( \{e_1, \ldots, e_l \} \right)$.
  We will show that $4f(S) \ge f(O)$.
  For each $o \in O \cap Y$, let $Y_{i(o)}$ denote the value of
  $Y$ at the beginning of the iteration in which $o$ was added to $Y$.
  We have
  \begin{align*}
    f( O \cup X ) - f(X) &\le \sum_{o \in O \setminus (X \cup Y)} \delta_o ( X ) + \sum_{o \in O \cap Y} \delta_o(X) \\
    &\le 0 + \sum_{o \in O \cap Y} \delta_o( Y_{i(o)} ) \le f(Y),
  \end{align*}
  where the second inequality follows from submodularity and the comparison on Line \ref{line:cmp-unc}.
  Analagously, $f(O \cup Y) - f(Y) \le f(X)$. Hence
  $$f(O) \le f(O \cup X) + f(O \cup Y) \le 2(f(X) + f(Y)) \le 4f(S),$$
  where the first inequality follows by submodularity, nonnegativity of $f$ and the fact that $X \cap Y = \emptyset$.
\end{proof}
\section{Proofs for Section \ref{sec:card}}
\begin{proof}[Proof of Theorem \ref{thm:lc}]
  Let $X, Y$ have their values after receiving a set $\uni$ of
  elements; let
  $O \subseteq \uni$ be an optimal solution to \nmon$(f, \uni, k)$.
  First, we will bound $f(O)$ in terms of $M_1 = \max \{ f(X), f(Y) \}$;
  subsequently, we will bound $M_1$ in terms of $M_2 = \max \{ f(X'), f(Y') \}$. Observe that $f(X),f(Y)$ do not decrease during the execution of the algorithm; so that, at any point during the execution, we have $\max \{ f(X), f(Y) \} \le M_1$.
  
  For each $o \in O \cap Y$, let $Y_{i(o)}$ denote the value of
  $Y$ at the beginning of the iteration in which $o$ was added to $Y$.
  We have
  \begin{align*}
    f( O \cup X ) &- f(X) \le \sum_{o \in O \setminus (X \cup Y)} \delta_o ( X ) + \sum_{o \in O \cap Y} \delta_o(X) \\
    &\le bM_1 + \sum_{o \in O \cap Y} \delta_o( Y_{i(o)} ) \le bM_1 + f(Y),
  \end{align*}
  where the second inequality follows from submodularity and the comparisons on Lines \ref{line:cmp-linear-card} and \ref{line:add-linear-card}.
  Analagously, $f(O \cup Y) - f(Y) \le bM_1 + f(X)$. Hence
  \begin{align*} 
    f(O) &\le f(O \cup X) + f(O \cup Y) \\
    &\le 2(f(X) + f(Y)) + 2bM_1 \\
    &\le (4 + 2b)M_1, \numberthis \label{ineq:lc-part1}
  \end{align*}
  where the first inequality follows by submodularity, nonnegativity of $f$ and the fact that $X \cap Y = \emptyset$.

  Next, we turn to the bound of $M_1$ in terms of $M_2$.
  Consider $X' \subseteq X$: we will show that
  $f(X) \le (1 + 1/b)f(X')$.
  First, if $|X'| < k$, then $X' = X$. So assume that
  $|X'| = \{x_1,\ldots,x_k\}$, where the order is by when
  these elements were added to $X$. Let $X'_i = \{x_1, \ldots, x_i \}$.
  For $i \in \{1,\ldots,k\}$, we have
  \begin{align*}
    f(X'_i) - f(X'_{i - 1}) &\ge \ff{(X \setminus X') \cup X'_i} 
                             -  \ff{(X \setminus X') \cup X'_{i - 1}} \numberthis \label{ineq:lc-1}\\
                            &\ge \frac{b}{k}\ff{ (X \setminus X') \cup X'_{i-1} } \numberthis \label{ineq:lc-2}\\
                            &\ge \frac{b}{k}\ff{X\setminus X'}, \numberthis \label{ineq:lc-3}
  \end{align*}
  where Inequality \ref{ineq:lc-1} follows from submodularity of $f$; Inequality \ref{ineq:lc-2} follows from the addition of $x_i$ on Line \ref{line:add-linear-card}; and Inequality \ref{ineq:lc-3} follows from the fact that $f(X)$ does not decrease during the execution of the algorithm. The summation of these inequalities yields $f(X') - f( \emptyset ) \ge b f(X \setminus X')$. By submodularity and nonnegativity of $f$,
  \begin{align*}
    \ff{X} \le \ff{X \setminus X'} + \ff{X'} \le (1 + 1/b)\ff{X'}.
  \end{align*}
  Symmetrically, $\ff{Y} \le (1 + 1/b) \ff{Y'}$, so we have
  $M_1 \le (1 + 1/b)M_2$. Together with Inequality \ref{ineq:lc-part1}, we have
  $$f(O) \le (4 + 2b)M_1 \le (4+2b)(1 + 1/b)M_2.$$
\end{proof}
\section{Proofs for Section \ref{sec:qs}} 
\subsection{Proof of Lemma \ref{lemm:general}} \label{apx:lemm-gen}
\begin{claim} \label{claim:elementary}
For any $y \ge 1$, $b > 0$, 
if $i \ge (k / b + 1) \log y$, then $(1 + b/k)^i \ge y$.
\end{claim}
\begin{proof}
Follows directly
from the inequality $\log x \ge 1 - 1/ x$ for $x > 0$.
\end{proof}
\begin{proof}[Proof of Property 1 of Lemma \ref{lemm:general}]
  If no deletion is made at element $i$
  of the sequence, then the result follows directly from
  $f(C_i^+) \ge (1 + b/k)f(C_i)$. So suppose 
  deletion of set $C_j$ from $C_i$ occurs. Observe that $C_{i + 1} = (C_i \setminus C_j) \cup \{ c_i \}$,
  because the deletion is triggered by the addition of $c_i$ to $C_i$.
  \begin{claim}
    From index $j$ to index $i$,
    there have been $\ell(k / b + 1) \log_2 (k) - 1 \ge (\ell - 1)(k / b + 1) \log_2 (k)$ additions and no deletions in the sequence.
  \end{claim}
  \begin{proof}
    The criterion for deletion at index $l$ is $|C_{l-1}^+| > 2 \ell( k/b + 1 )\log_2(k)$.
    Since initially $C_0 = \emptyset$, a deletion occurs only at
    indices $l$ for which $|C_{l-1}^+| = 2 \ell( k/b + 1 )\log_2(k) + 1$; so $|C_{l}| = \ell( k/b + 1 )\log_2(k) + 1$.
    Therefore, there are at least $\ell(k / b + 1) \log_2 (k) - 1$ indices between successive deletions.
  \end{proof}
    If $\ff{C_i \setminus C_j} \ge \ff{C_i}$, the lemma follows
    from submodularity and the condition
    $f(C_i^+) \ge (1 + b/k)f(C_i)$.
    Therefore, for the rest of the proof, suppose $\ff{C_i \setminus C_j} < \ff{C_i}$.

  It holds that
  $$\ff{ C_i \setminus C_j } \overset{(a)}{\ge} \ff{ C_i } - \ff{ C_j } \overset{(b)}{\ge} \left( 1 + \frac{b}{k} \right)^{(\ell - 1)(k / b + 1) \log k} \cdot f( C_j ) - f (C_j) \overset{(c)}{\ge} (k^{\ell - 1} - 1 ) f( C_j ),
  $$
  where Inequality a follows from submodularity and nonnegativity of $f$,
  Inequality b follows from the fact that each addition from $C_j$ to $C_i$ increases
  the value of $f(C)$ by a factor of at least $(1 + b/k)$, and Inequality c
  follows from Claim
  \ref{claim:elementary}.
  Therefore
  \begin{equation} \label{ineq:1}
    f(C_i) \le \ff{C_i \setminus C_j} + \ff{ C_j } \le \left( 1 + \frac{1}{k^{\ell - 1} - 1} \right) \ff{C_i \setminus C_j}.
  \end{equation}
  Next,
  \begin{equation} \label{ineq:2}
    \ff{ (C_i \setminus C_j) \cup \{ c_i \} } - \ff{ C_i \setminus C_j } \overset{(d)}{\ge} \ff{ C_i \cup \{ c_i \} } - \ff{ C_i } \overset{(e)}{\ge} b\ff{C_i} / k  \overset{(f)}{\ge} \ff{C_i \setminus C_j } / k, \end{equation}
  where Inequality d follows from submodularity; Inequality e is by
  the condition $f(C_i^+) \ge (1 + b/k)f(C_i)$;
  and Inequality f holds since $b \ge 1$ and $\ff{C_i} > \ff{C_i \setminus C_j}$.
  Finally, using Inequalities (\ref{ineq:1}) and (\ref{ineq:2}) as indicated below, we have
  \begin{align*}
    \ff{ C_{i + 1}} = \ff{ C_i \setminus C_j \cup \{ e_i \} } \overset{\text{By (\ref{ineq:2})}}{\ge} \left( 1 + \frac{1}{k} \right) \ff{ C_i \setminus C_j } \overset{\text{By (\ref{ineq:1})}}{\ge} \frac{1 + \frac{1}{k}}{1 + \frac{1}{k^{\ell - 1} - 1}} \cdot f(C_i) \ge f( C_i ), 
  \end{align*}
  where the last inequality follows since $k \ge 2$ and $\ell \ge 3$.
\end{proof}
\begin{proof}[Proof of Property 2 of Lemma \ref{lemm:general}]
\begin{lemma} \label{lemm:aplus}
  $\ff{C^*} \le \left( 1 + \frac{1}{k^\ell - 1} \right)\ff{ C_{n + 1} }. $
\end{lemma}
\begin{proof}
  Observe that $C^* \setminus C_{n + 1}$ may be written as the
  union of pairwise disjoint sets, each of which is size
  $\ell(k/b+1) \log_2(k)$.
  Suppose there were $m$ sets deleted during the sequence; write 
  $C^* \setminus C_{n + 1} = \{ D^i : 1 \le i \le m \}$, ordered such that $i < j$
  implies $D^i$ was deleted after $D^j$ (the reverse order in which they were deleted); finally, let $D^0 = C_{n + 1}$.
  \begin{claim} \label{claim:deletion}
    Let $0 \le i \le m$. Then $\ff{D^i} \ge k^{\ell} \ff{D^{i+1}}$.
  \end{claim}
  \begin{proof}
    There are at least $\ell(k/b + 1) \log k + 1$ elements added to $C$ 
    and exactly one deletion
    event during the period between starting when $C = D^{i + 1}$ until $C = D^i$. 
    Moreover, each addition except possibly one (corresponding to the deletion event)
    increases $f(C)$ by a factor
    of at least $1 + b/k$. Hence, by Lemma \ref{lemm:general} and Claim \ref{claim:elementary},
  $\ff{D^i} \ge k^{\ell} \ff{D^{i+1}}$.
  \end{proof}
  By Claim \ref{claim:deletion}, for any $0 \le i \le m$,
  $\ff{ C_{n + 1} = D^0 } \ge k^{\ell i} \ff{ D^i }$.
  Thus, 
  \begin{align*}
    \ff{ C^* } \le \ff{ C^* \setminus C_{n + 1} } + \ff{ C_{n + 1} } &\le \sum_{i = 0}^m \ff{ D^i } && \tag{Submodularity, Nonnegativity of $f$}\\
    &\le \ff{ C_{n + 1} } \sum_{i = 0}^\infty k^{-\ell i} \tag{Claim \ref{claim:deletion}} \\
                                                                     &=  \ff{ C_{n + 1} } \left( \frac{1}{1 - k^{-\ell }} \right) \tag{Sum of geometric series} 
  \end{align*} \renewcommand{\qedsymbol}{}
\end{proof} \end{proof}
\subsection{Proof of Lemma \ref{lemm:two}}
\begin{proof}
  Since $o \in O \cap D^*$, we know that
  $o$ is added to the set $D$ during iteration $i(o)$;
  therefore, by
  the comparison on Line \ref{line:cmp} of Alg. \ref{alg:quickstream},
  it holds that
  \begin{equation} \label{eq:marge}
    \marge{o}{C_{i(o)}} \le \marge{o}{D_{i(o)}}.
  \end{equation}
  If no deletion from $D$ occurs
  during iteration $i(o)$, the lemma follows from
  the fact that
  $\Delta D_{i(o)} = \marge{o}{D_{i(o)}}.$

  For the rest of the proof, suppose that a deletion from $D$
  does occur during iteration $i(o)$.
  For convenience, denote by $D^-$ 
  the value of $D$ after the deletion from $D_{i(o)}$.
  By Inequality \ref{ineq:1} in the proof of Lemma
  \ref{lemm:general}, 
  it holds that
  \begin{equation} \label{eq:gamma}
    \ff{D_{i(o)}} \le (1 + \gamma ) \ff{D^-}
  \end{equation}
  Hence,
  \begin{align*}
    \Delta D_{i(o)} &+ \gamma \ff{D_{n + 1}} = \ff{D^- + o } - \ff{ D_{i(o)} } + \gamma \ff{D_{n + 1}} \\
    &\ge (1 + \gamma)\ff{ D^- + o } - \ff{D_{i(o)}} \numberthis \label{ineq:aa} \\
    &\ge (1 + \gamma)\ff{ D^- + o } - (1 + \gamma) \ff{D^-} \numberthis \label{ineq:bb} \\
    &= (1 + \gamma) \marge{o}{D^-} \\
    &\ge (1 + \gamma) \marge{o}{D_{i(o)}} \numberthis \label{ineq:cc}\\
    &\ge (1 + \gamma) \marge{o}{C_{i(o)}}. \numberthis \label{ineq:dd}
  \end{align*}
  where Inequality
  \ref{ineq:aa} follows from Lemma \ref{lemm:general},
  Inequality \ref{ineq:bb} follows from Inequality \ref{eq:gamma},
  Inequality \ref{ineq:cc} follows from submodularity of $f$,
  and Inequality \ref{ineq:dd} follows from Inequality \ref{eq:marge}.
\end{proof}
\subsection{Proof of Lemma \ref{lemm:uni-bound}}
\begin{proof}
  \begin{align*}
    &\ff{ C^* \cup O } - \ff{C^*} \le \sum_{o \in O \setminus C^*} \marge{o}{C^*} \numberthis \label{iq:a}\\
    &\le \sum_{o \in O \setminus C^*} \marge{o}{C_{i(o)}} \numberthis \label{iq:b}\\
    &\le \sum_{o \in O \setminus C^*} \frac{b \ff{C_{i(o)}}}{k} + \frac{\Delta D_{i(o)} + \gamma \ff{D_{n + 1}}}{1 + \gamma} \numberthis \label{iq:c}\\
    &\le b \ff{C_{n + 1}} + \frac{1 + k\gamma}{1 + \gamma} \ff{D_{n+1}} \numberthis \label{ineq:ten}\\
    &\le b\ff{C_{n + 1}} + (1 + k\gamma) \ff{D_{n+1}},
  \end{align*}
  where Inequalities \ref{iq:a}, \ref{iq:b} follow
  from submodularity of $f$.
  Inequality \ref{iq:c} holds by the following
  argument: let $o \in O \setminus C^*$. If $o \not \in D^*$,
  then it holds that $\marge{o}{C_{i(o)}} < b \ff{ C_{i(o)} } / k$
  by Line \ref{line:cmp}. Otherwise, if $o \in D^*$, Lemma 
  \ref{lemm:two} yields $\marge{o}{C_{i(o)}} \le
  \frac{\Delta D_{i(o)} + \gamma \ff{D_{n + 1}} }{1 + \gamma}.$
  Inequality \ref{ineq:ten} follows from the fact that
  $|O \setminus C^*| \le k$, Lemma \ref{lemm:general},
  and the fact that
  $\ff{D_{n+1}} - \ff{\emptyset} = \sum_{i=0}^{n} \Delta D_i$,
  where each $\Delta D_i \ge 0$.
\end{proof}
\subsection{Justification of choice of $\ell$}
\begin{lemma} \label{lemm:ell-choice}
  Let $\epsi > 0$, and let $\alpha = \ratioPrime$.
  Choose $\ell \ge 1 + \log ( (6\alpha)/\epsi + 1 )$, 
  and let $\gamma = 1 / (k^\ell - 1)$.
Then
$$2(k + 2) \gamma < \epsi \alpha^{-1} .$$
\end{lemma}
\begin{proof}
  First, one may verify that $\ell > \frac{\log \left( \frac{(2k + 4)\alpha }{\epsi} + 1 \right) }{ \log k } \implies 2(k + 2) \gamma < \epsi \alpha^{-1} $. 
  Next, since $k \ge 1$,
  \begin{align*}
    \frac{1}{\log k} \left( \log \left( \frac{(2k + 4)\alpha }{\epsi} + 1 \right)\right)
    &\le \frac{1}{\log k} \left( \log \left( \left( \frac{(2 + 4)\alpha}{\epsi } + 1 \right) k \right) \right) \\
    &= \frac{1}{\log k} \left( \log \left( \frac{6 \alpha }{\epsi } + 1\right) + \log(k)\right) \\
    &\le 1 + \log\left( \frac{6 \alpha}{\epsi } + 1 \right).
  \end{align*}
  Hence it suffices to take $\ell$ greater than the last expression.
\end{proof}
\section{Proofs for Section \ref{sec:brlarge}} \label{apx:proofs-br}
\begin{proof}
  To establish the approximation ratio, consider
  first the
  case in which $C \in \{A,B\}$ satisfies $|C| = k$ 
  after the first iteration of the \textbf{while}
  loop. Let $C = \{c_1,\ldots,c_k\}$ be ordered
  by the order in which elements were added to $C$
  on Line \ref{line:addbr}, let $C_i = \{c_1,\ldots,c_i\}$,
  $C_0 = \emptyset$, 
  and let $\Delta C_i = \ff{C_i} - \ff{C_{i - 1}}$.
  Then
  $f(C) = \sum_{i=1}^k \Delta C_i \ge \Gamma / (4 \alpha) \ge \opt / 4,$
  and the ratio is proven.

  Therefore, for the rest of the proof,
  suppose $|A| < k$ and $|B| < k$
  immediately after the execution of the first iteration
  of the \textbf{while} loop.
  First, let $C, D \in \{A,B\}$, such that $C \neq D$
  have their values at the termination of the algorithm.
  For the definition of $D'$ and the proofs of the next two
  lemmata, see Appendix \ref{apx:proofs-br}. These lemmata
  together establish an upper bound on $\marge{O}{C}$ in
  terms of the gains of elements added to $C$ and $D$.
  \begin{lemma} \label{lemm:five}
    $$\sum_{o \in O \setminus (C \cup D')} \marge{o}{C} \le (1 + 2\epsi)\sum_{i : c_i \not \in O} \Delta C_i + \epsi \opt / 16.$$
  \end{lemma}

  \begin{lemma} \label{lemm:six}
    $$\marge{O}{C} \le \sum_{i: d_i \in O} \Delta D_{i} + (1 + 2 \epsi )\sum_{i : c_i \not \in O} \Delta C_i + \epsi \opt / 16.$$
  \end{lemma}
  Applying Lemma \ref{lemm:six} with $C=A$ and separately with
  $C=B$ and summing the resulting inequalities yields
  \begin{align*}
    \marge{O}{A} + \marge{O}{B} &\le (1 + 2\epsi) \left[ \sum_{i = 1}^k \Delta B_{i} + \sum_{i = 1}^k \Delta A_{i} \right] + \frac{\epsi \opt}{8}\\
    &= (1 + 2\epsi)\left[\ff{A} + \ff{B}\right] + \epsi \opt / 8.
  \end{align*}
  Thus,
  \begin{align*}
    \ff{O} &\le \ff{O \cup A} + \ff{O \cup B} \\
           &\le (2 + 2\epsi)( \ff{A} + \ff{B} ) + \epsi \opt / 8,
  \end{align*}
  from which the result follows.
\end{proof}

  Let $C, D \in \{A,B\}$, such that $C \neq D$
  have their values at the termination of the algorithm.
  If $|C| = k$, let $D'$ have the value of its corresponding variable 
  when the $k$th element is added to $C$; otherwise, if $|C| < k$
  let $D' = D$.
  \begin{proof}[Proof of Lemma \ref{lemm:five}]
    Suppose $|C| = k$. Let $\tau'$ be the value of $\tau$
    during the iteration of the \textbf{while} loop in which
    the last element was added to $C$.
    Let $o \in O \setminus (C \cup D')$. Then, since $o$ was not added to $C$ or $D'$
    during the previous iteration of the \textbf{while} loop, $\marge{o}{C} < \tau' / (1 - \epsi)$. Further, $\Delta C_i \ge \tau'$ for all
    $i$. Hence,
    \begin{align*}
      \sum_{o \in O \setminus (C \cup D')} \marge{o}{C} &\le \frac{1}{1- \epsi}\sum_{i : c_i \not \in O} \Delta C_i \\
                                                         &\le (1 + 2\epsi) \sum_{i : c_i \not \in O} \Delta C_i.
    \end{align*}
    Next, suppose that $|C| < k$. In this case, the last threshold
    $\tau$ of the \textbf{while} loop ensures that $\sum_{o \in O \setminus (C \cup D')} \marge{o}{C} < \epsi \Gamma / 16 \le \epsi \opt / 16.$  
  \end{proof}
  \begin{proof}[Proof of Lemma \ref{lemm:six}]
  Observe that
  \begin{align}
    \marge{O}{C} &\le \sum_{o \in O \cap D'} \marge{o}{C} + \sum_{o \not \in (C \cup D')} \marge{o}{C} \label{ineq:19}\\
    &\le \sum_{i : d_i \in O} \Delta D_i + \sum_{o \not \in (C \cup D')} \marge{o}{C}, \label{ineq:20}
  \end{align}
  where Inequality \ref{ineq:19} follows from submodularity
  and Inequality \ref{ineq:20} follows from submodularity
  and the comparison on Line \ref{line:cmpbr} for each element
  $o \in O \cap D'$. From Inequality \ref{ineq:20}, the lemma
  follows from application of Lemma \ref{lemm:five}.
  \end{proof}

  \section{Proofs for Section \ref{sec:lrig}} \label{apx:lrig}
  \begin{proof}[Proof of Theorem~\ref{thm:lrig}]
Suppose that, at each iteration $m$, we select a random set from the 
correct guess of $u$ in Alg.~\ref{thm:ig}.
Then, the two inequalities of Theorem~\ref{thm:ig} holds for each iteration,
and it happens with probability $(\ell+1)^{-\ell}$.

By unfixing $G_{m-1}$ and utilizing Inequality (1) of Theorem~\ref{thm:ig}
multiple times, it holds that
\begin{align*}
\ex{\ff{O \cup G_m}} &\ge \left(1-\frac{1}{\ell}\right)\ex{\ff{O \cup G_{m-1}}}\\
&\ge \left(1-\frac{1}{\ell}\right)^m f(O).
\end{align*}
Then, the approximation ratio can be bounded as follows
\begin{align*}
\ex{f(G_\ell)} &\ge \parth{1-\frac{1}{1+\ell}}
\ex{f(G_{\ell-1})} + \frac{1}{1+\ell}
\parth{1-\frac{1}{\ell}} \ex{f(O\cup G_{\ell-1})}\\
& \ge \parth{1-\frac{1}{\ell}}
\ex{f(G_{\ell-1})} + \frac{1}{1+\ell}
\parth{1-\frac{1}{\ell}}^\ell f(O)\\
&\ge \frac{\ell}{1+\ell}
\parth{1-\frac{1}{\ell}}^\ell f(O)\\
&\ge \parth{1-\frac{2}{\ell+1}}e^{-1}f(O) \tag{$\parth{1-\frac{1}{\ell}}^{\ell-1}\ge e^{-1}$}\\
&= \frac{1}{e+\epsi}f(O).\tag{$\ell = \frac{2e}{\epsi}+1$}
\end{align*}

\end{proof}
  
\section{Derandomized and Fast Algorithm for \lrig}\label{sec:derand-alg}
  \begin{algorithm}[t!]
    \caption{An $(e+\epsi)$-approximation algorithm for \nmon}\label{alg:lrtig}
    \begin{algorithmic}[1]
    \Procedure{\lrtig}{$f, k, \epsi$}
    \State \textbf{Input:} oracle $f$, cardinality constraint $k$, $\epsi>0$,
    \State \textbf{Initialize } 
    $\epsi' \gets \frac{\epsi}{2(e+1)(e+\epsi)}$,
    $\ell \gets \frac{4e}{\epsi}+3$,
    $G_0 \gets \{\emptyset\}$
    \For{$m \gets 1$ to $\ell$}
      \State $G_{m} \gets \emptyset$,
      \For{$S_{m-1} \in G_{m-1}$},
      \State $\{a_1, \ldots , a_\ell\} \gets$ top $\ell$ elements
      in $\mathcal{U} \setminus S_{m-1}$ with respect to marginal gains on $S_{m-1}$ 
      \For{$u \gets 0$ to $\ell$ in parallel} \Comment{$\ell+1$ guesses of max singleton in $O$}
        \If{$u = 0$}
        \State $M \gets \marge{a_\ell}{S_{m-1}}$
        \State $A_{u,l} \gets S_{m-1}\cup \{a_l\}$, for any $l \in [\ell]$ \label{line:A_ini}
        \Else 
        \State $M \gets \marge{a_u}{S_{m-1}}$
        \State $A_{u,l} \gets S_{m-1} \cup \{a_u\}$, for any $l \in [\ell]$ \label{line:A_ini_2}
        \EndIf
        \State $\tau_l \gets M$, $I_l\gets \textbf{true}$, for any $l \in [\ell]$
        \While{$\lor_{l=1}^\ell I_l $}
          \For{$i \gets 1$ to $\ell$}
          \If{$I_i$}
          \State $V\gets \uni\setminus \cup_{l=1}^\ell A_{u,l}$
            \State $A_{u,i}, \tau_i \gets \add(f, V, A_{u,i},\epsi', \tau_i, \frac{\epsi' M}{k\ell})$
            \If{$|A_{u,i}\setminus S_{m-1}|=k/\ell \lor \tau_i < \frac{\epsi' M}{k\ell}$}
             \State $I_i \gets \textbf{false}$
            \EndIf
            \EndIf
          \EndFor
          \EndWhile
          \State $G_m \gets G_m \cup \{A_{u, 1}, A_{u, 2}, \ldots, A_{u, \ell}\}$
      \EndFor
    \EndFor
    \EndFor
    \State \textbf{return} $S^* \gets \argmax \{f(S_\ell): S_\ell \in G_\ell\}$
    \EndProcedure
\end{algorithmic}
\end{algorithm}
\begin{algorithm}[t!]
    \caption{Add one element from $V$ to $A$ with threshold in between $[\tau_{\min},\tau]$}
    \begin{algorithmic}[1]
    \Procedure{\add}{$f, V, A, \epsi, \tau, \tau_{\min}$}
    \State \textbf{Input:} oracle $f$,
    candidate set $V$,
    solution set $A$,
    $\epsi$, threshold $\tau$,
    and its lower bound $\tau_{\min}$
    \While{$\tau \ge \tau_{\text{min}}$}
    \For{$x \in V$}
      \If{$\marge{x}{A} \ge \tau$}
      \textbf{return} $A \gets A \cup \{x\}$, $\tau$
      \EndIf
    \EndFor
    \State $\tau \gets (1-\epsi)\tau$
    \EndWhile
    \State \textbf{return} $A$, $\tau$
    \EndProcedure
\end{algorithmic}
\end{algorithm}
In this section, we introduce the derandomized and fast version of \lrig,
which gives a deterministic $(e+\epsi)$-approximation with $\oh{n\log (k)}$ query complexity.
\begin{theorem}\label{thm:lrtig-de}
Let $\epsi \ge 0$, and
$(f,k)$ be an instance of \nmon,
with optimal solution value \opt.
Algorithm \lrtig (Alg.~\ref{alg:lrtig}) outputs a set $S^*$
with $\oh{\epsi^{-2/\epsi-1}n\log (k)}$ queries such that
$\opt \le (e+\epsi)f(S^*)$.
\end{theorem}
\textbf{Algorithm Speedup.}
Inspired by \textsc{AcceleratedGreedy} in~\cite{Badanidiyuru2014},
we replace the greedy selection in Alg.~\ref{alg:rig} with the threshold selection.
Instead of querying $\oh{nk/\ell}$ times for each candidate set $A_{u,\ell}$
during the while loop on Line~\ref{line:rig-while-begin}-\ref{line:rig-while-end},
threshold procedure only queries $\oh{n}$ times for each set with $\oh{\log(k)}$ iterations.

\textbf{Derandomization.}
Unlike typical randomized algorithms that make choices
from $\oh{n}$ possibilities,
the only randomization in
\lrig is that it randomly picks a set from $\ell(\ell+1)$ candidates.
Therefore, 
it is possible to derandomize the algorithm easily
by retaining all the candidate sets.
Then, the algorithm returns the solution with deterministic $(e+\epsi)-$approximation
according to $\oh{\ell^{2\ell}}$ candidate sets.

\begin{proof}[Proof of Theorem~\ref{thm:lrtig-de}]
\begin{lemma}\label{lemma:lrtig_rec1}
At iteration $m$, given $S_{m-1}$,
there exists $0 \le u \le \ell$ such that, for any $1 \le l \le \ell$,
it holds that 
\begin{align*}
\left(\frac{\ell}{1-\epsi'}+1\right)\func{f}{A_{u,l}}
\ge \func{f}{O \cup A_{u,l}}+
\frac{\ell}{1-\epsi'}\func{f}{S_{m-1}}
-\frac{\epsi'}{\ell} f(O).
\end{align*}
and for any $1 \le l_1 \neq l_2 \le \ell$,
\[\parth{O \cup A_{u,l_1}} \cap \parth{O \cup A_{u,l_2}} = O \cup S_{m-1}.\]
\end{lemma}
\begin{proof}
At iteration $m$ of the first for loop,
condition on $S_{m-1}$ of the second for loop.
For any $0 \le u \le \ell$, by submodularity, it holds that 
$\marge{a_u}{S_{m-1}}\le f(a_u) \le f(O)$.
Let $o_{\text{max}} = \argmax_{o \in O\setminus S_{m-1}}\marge{o}{S_{m-1}}$.
We consider the following two cases.

\textbf{Case  $(O\setminus S_{m-1}) \cup \{a_1, \ldots, a_\ell\}=\emptyset$.}
In this case, we analyze that
lemma holds with sets $\{A_{0,l}\}_{l=1}^\ell$.

For any $l \in [\ell]$, 
let $A_{0,l}^{(j)}$ be $A_{0,l}$ after we add $j$ elements into it,
$\tau_l^{(j)}$ be $\tau_l$ when we adopt $j$-th elements into $A_{0,l}$,
and $\tau_l^{(1)} = M$.
By Line~\ref{line:A_ini}, it holds that $A_{0,l}^{(1)}=S_{m-1}\cup \{a_l\}$.
Since $(O\setminus S_{m-1}) \cup \{a_1, \ldots, a_\ell\}=\emptyset$,
and we add elements to each set in turn,
we can order $O\setminus S_{m-1} = \{o_1, o_2, \ldots\}$ 
such that the first $\ell$ elements are not selected by any set 
before we get $A_{0,l}^{(1)}$,
the next $\ell$ elements are not selected in any set before we get $A_{0,l}^{(2)}$,
and so on.
Therefore, 
for any $j \le |A_{0,l} \setminus S_{m-1}|$ 
and $\ell (j-1)+1 \le i \le \ell j$,
$o_i$ are filtered out by $A_{0,l}$
with threshold $\tau_l^{(j)}/(1-\epsi')$, 
which follows that 
$\marge{o_i}{A_{0,l}^{(j)}} < \tau_l^{(j)}/(1-\epsi') \le (f(A_{0,l}^{(j)})- f(A_{0,l}^{(j-1)}))/(1-\epsi')$;
for any $\ell |A_{0,l} \setminus S_{m-1}| < i \le |O \setminus S_{m-1}|$,
$o_i$ are filtered out by $A_{0,l}$
with threshold $\frac{\epsi' M}{k\ell}$, 
which follows that 
$\marge{o_i}{A_{0,l}} < \epsi' M/k$.
Thus,
\begin{align*}
&\func{f}{O \cup A_{0,l}} - \func{f}{A_{0,l}}
\le \sum_{o \in O\setminus A_{0,l}}\marge{o}{A_{0,l}}\tag{submodularity}\\
&\le \sum_{o \in O\setminus S_{m-1}}\marge{o}{A_{0,l}}\tag{$S_{m-1} \subseteq A_{0,l}$}\\
&\le \sum_{i = 1}^{\ell} \marge{o_i}{A_{0,l}^{(1)}}
+ \sum_{i = \ell+1}^{2\ell} \marge{o_i}{A_{0,l}^{(2)}}+\ldots 
+ \sum_{i > \ell |A_{0,l} \setminus S_{m-1}|} \marge{o_i}{A_{0,l}}
\tag{submodularity}\\
&\le \ell \cdot \frac{\func{f}{A_{0,l}}-\func{f}{S_{m-1}}}{1-\epsi'}
+\frac{\epsi' M}{\ell}\\
&\le \ell \cdot \frac{\func{f}{A_{0,l}}-\func{f}{S_{m-1}}}{1-\epsi'}
+\frac{\epsi'}{\ell}f(O).
\end{align*}
Since, $A_{u,l_1} \cap A_{u,l_2} = S_{m-1}$
for any $1 \le l_1 \neq l_2 \le \ell$,
it holds that 
$\parth{O \cup A_{u,l_1}} \cap \parth{O \cup A_{u,l_2}} = O\cup S_{m-1}$.

\textbf{Case  $(O\setminus S_{m-1}) \cup \{a_1, \ldots, a_\ell\}\neq\emptyset$.}
Then $o_{\text{max}} \in \{a_1, \ldots, a_\ell\}$.
Suppose that $a_u = o_{\text{max}}$.
We analyze that lemma holds with sets $\{A_{u,l}\}_{l=1}^\ell$.

Similar to the analysis of the previous case, 
let $A_{0,l}^{(j)}$ be $A_{0,l}$ after we add $j$ elements into it,
$\tau_l^{(j)}$ be $\tau_l$ when we adopt $j$-th elements into $A_{0,l}$,
and $\tau_l^{(1)} = M$.
By Line~\ref{line:A_ini_2}, it holds that $A_{u,l}^{(1)}=S_{m-1}\cup \{a_u\}$.
Then, we can order $O\setminus \left(S_{m-1}\cup \{a_u\}\right) = \{o_1, o_2, \ldots\}$ 
such that the first $\ell$ elements are not selected by any set 
before we get $A_{u,l}^{(1)}$,
the next $\ell$ elements are not selected in any set before we get $A_{u,l}^{(2)}$,
and so on.
Therefore, we can get the same result as the previous case for sets
$\{A_{u,l}\}_{l=1}^\ell$.
Since, $a_u \in O$, and
$A_{u,l_1} \cap A_{u,l_2} = S_{m-1}\cup \{a_u\}$
for any $1 \le l_1\neq l_2 \le \ell$,
it holds that 
$\parth{O \cup A_{u,i_1}} \cap \parth{O \cup A_{u,i_2}} = O \cup S_{m-1}$.

Overall, since either one of the above cases happens,
the lemma holds.
\end{proof}
\begin{lemma}\label{lemma:lrtig_rec2}
For any $1 \le m \le \ell$,
there exists $G_m' \subseteq G_m$ such that
\begin{align*}
&\frac{\sum_{S_m \in G_m'} f(S_m)}{|G_m'|} \ge \parth{\frac{m}{\frac{\ell}{1-\epsi'}+1}
\parth{1-\frac{1}{\ell}}^m - \epsi' \parth{1-\parth{1-\frac{1}{\ell}}^m}}\opt,\\
&\frac{\sum_{S_m \in G_m'} f(O \cup S_m)}{|G_m'|} \ge \parth{1-\frac{1}{\ell}}^m \opt.
\end{align*}
\end{lemma}
\begin{proof}
Given any $S_{m-1} \in G_{m-1}$, let Lemma~\ref{lemma:lrtig_rec1} holds with $u_0$.
Then,
\begin{align*}
&\frac{1}{\ell}\sum_{l=1}^\ell \func{f}{O \cup A_{u_0,l}}
\ge \frac{1}{\ell} \parth{(\ell-1)\func{f}{O \cup S_{m-1}} + 
\func{f}{O\cup \parth{\cup_{i=1}^\ell A_{u_0,l}}}}
\ge \left(1-\frac{1}{\ell}\right)\func{f}{O \cup S_{m-1}}
\numberthis \label{inq:Ounion}\\
\Rightarrow &\frac{1}{\ell}\sum_{l=1}^{\ell}f(A_{u_0, l})
\ge \frac{\frac{\ell}{1-\epsi'}}{\frac{\ell}{1-\epsi'}+1}f(S_{m-1}) +
\frac{1}{\frac{\ell}{1-\epsi'}+1}\frac{1}{\ell}\sum_{l=1}^\ell \func{f}{O \cup A_{u_0,l}}
- \frac{\epsi'}{\ell}f(O)\\
&\ge \parth{1-\frac{1}{\ell}}f(S_{m-1}) +
\frac{1-\frac{1}{\ell}}{\frac{\ell}{1-\epsi'}+1}f(O \cup S_{m-1})
- \frac{\epsi'}{\ell}f(O) \numberthis \label{inq:sol}
\end{align*} 
When $m=1$ and $S_0 = \emptyset$, based on the analysis above,
there exists $\ell$ sets in $G_1$ as $G_1'$ such that
\begin{align*}
 &\frac{\sum_{S_1 \in G_1'} f(O \cup S_1)}{\ell} \ge \left(1-\frac{1}{\ell}\right)\opt,\\
 & \frac{\sum_{S_1 \in G_1'} f(S_1)}{\ell} \ge \parth{\frac{1-\frac{1}{\ell}}{\frac{\ell}{1-\epsi'}+1}-\frac{\epsi'}{\ell}}\opt.
 \end{align*}
 Lemma holds with $m=1$ immediately.

 Suppose that lemma holds with iteration $1, \ldots, m-1$.
For any $S_{m-1} \in G_{m-1}'$,
there exists $\ell$ sets which follow Inequality~\ref{inq:Ounion} and~\ref{inq:sol},
and form the set $G_m'$.
Then,
\begin{align*}
\frac{\sum_{S_m \in G_m'} f(O \cup S_m)}{|G_m'|} &\ge 
\left(1-\frac{1}{\ell}\right) \frac{\sum_{S_{m-1} \in G_{m-1}'} f(O \cup S_{m-1})}{|G_{m-1}'|}
\ge \parth{1-\frac{1}{\ell}}^m \opt,
\end{align*}
and
\begin{align*}
\frac{\sum_{S_m \in G_m'} f(S_m)}{|G_m'|} &\ge 
\parth{1-\frac{1}{\ell}}\frac{\sum_{S_{m-1} \in G_{m-1}'} f(S_{m-1})}{|G_{m-1}'|}
 +\frac{1-\frac{1}{\ell}}{\frac{\ell}{1-\epsi'}+1}\frac{\sum_{S_{m-1} \in G_{m-1}'} f(O \cup S_{m-1})}{|G_{m-1}'|}
- \frac{\epsi'}{\ell} f(O)\\
&\ge \parth{\frac{m}{\frac{\ell}{1-\epsi'}+1}
\parth{1-\frac{1}{\ell}}^m - \epsi' \parth{1-\parth{1-\frac{1}{\ell}}^m}}\opt.
\end{align*}
\end{proof}
By Lemma~\ref{lemma:lrtig_rec2}, the approximation ratio for \lrtig can be bounded as follows,
\begin{align*}
f(S^*) &\ge \frac{\sum_{S_\ell \in G_\ell'} f(S_\ell)}{|G_\ell'|} \\
&\ge \parth{\frac{\ell}{\frac{\ell}{1-\epsi'}+1}
\parth{1-\frac{1}{\ell}}^\ell - \epsi' \parth{1-\parth{1-\frac{1}{\ell}}^\ell}}\opt\\
&\ge \parth{\frac{\ell-1}{\frac{\ell}{1-\epsi'}+1}
e^{-1} - \epsi' }\opt \tag{$\parth{1-\frac{1}{x}}^{x-1}\ge e^{-1}$}\\
&\ge \frac{1}{e+\epsi} \opt \numberthis \label{inq:num-cal}
\end{align*}
\end{proof}
\begin{proof}[Proof of Inequality~\ref{inq:num-cal}]
With $\epsi'= \frac{\epsi}{2(e+1)(e+\epsi)}$
and $\ell= \frac{4e}{\epsi}+3$,
\begin{align*}
\frac{\ell-1}{\frac{\ell}{1-\epsi'}+1} e^{-1} - \epsi' -\frac{1}{e+\epsi}
&= \frac{1}{\frac{1}{1-\epsi'}+\frac{1}{\ell}}\parth{\parth{1-\frac{1}{\ell}}e^{-1}-\epsi' \parth{\frac{1}{1-\epsi'}+\frac{1}{\ell}}-\frac{\frac{1}{1-\epsi'}+\frac{1}{\ell}}{e+\epsi}}\\
&= \frac{1}{\frac{1}{1-\epsi'}+\frac{1}{\ell}}
\parth{e^{-1}-\frac{\epsi'}{1-\epsi'}-\frac{1}{(1-\epsi')(e+\epsi)}-\frac{1}{\ell}\parth{e^{-1}+\epsi'+\frac{1}{e+\epsi}}}\\
&= \frac{1}{\frac{1}{1-\epsi'}+\frac{1}{\ell}}
\parth{\frac{\epsi-(e+1)(e+\epsi)\epsi'}{e(1-\epsi')(e+\epsi)} - \frac{2e+\epsi+e(e+\epsi)\epsi'}{\ell e(e+\epsi)}}\\
&= \frac{1}{\frac{1}{1-\epsi'}+\frac{1}{\ell}}
\parth{\frac{\epsi}{2e(1-\epsi')(e+\epsi)} - \frac{\frac{\epsi}{2}\parth{\frac{4e}{\epsi}+2+\frac{e}{e+1}}}{\ell e(e+\epsi)}}\\
&\ge \frac{1}{\frac{1}{1-\epsi'}+\frac{1}{\ell}}
\parth{\frac{\epsi}{2e(e+\epsi)} - \frac{\epsi}{2e(e+\epsi)}}=0.
\end{align*}
\end{proof}
\section{Analysis of \lrig for monotone \nmon}\label{sec:mon-lrig}
In this section, we show that \lrig also works for 
monotone objectives and obtains the best ratio of $\frac{e}{e-1}+\epsi$
in expectation with probability $(\ell + 1)^{-\ell}$. The derandomization
and speed up also works for the monotone analysis, although we omit
the proofs here. 
\begin{theorem}\label{thm:ig-mon}
  Let $O \subseteq \uni$ be any set of size at most $k$, and suppose
  $\ig$ is called with $(G, f, k, \ell)$, where $f$ is monotone.
  Then $\ig$ outputs a set $A$
with $\oh{\ell nk}$ queries and probability $(\ell+1)^{-1}$ such that:
\begin{align*}
\ex{f(A)} \ge \frac{\ell}{\ell+1} f(G)+
    \frac{1}{\ell+1} f(O).
\end{align*}
\end{theorem}
\begin{proof}[Proof of Theorem~\ref{thm:ig-mon}]
  Let $o_{\text{max}} = \argmax_{o \in O\setminus G}\marge{o}{G}$,
  and let $\{ a_1, \ldots, a_\ell \}$ be the largest $\ell$ elements
  of $\{ \marge{x}{G} : x \in \uni \setminus G \}$, as chosen on
  Line \ref{line:rig-maxl}. We consider the following two cases.

\textbf{Case  $(O\setminus G) \cap \{a_1, \ldots, a_\ell\}=\emptyset$.}
Then, $o_{\text{max}} \not \in \{a_1, \ldots, a_\ell\}$ 
which implies that
$\marge{a_u}{G} \ge \marge{o}{G}$,
for all $1 \le u \le \ell$ and $o \in O\setminus G$;
and, after the first iteration of the \textbf{while} loop on Line~\ref{line:lrig_j},
no element of $O \setminus G$ is added into any of
$\{A_{0,i}\}_{i=1}^\ell$. We will analyze the iteration
of the \textbf{for} loop on Line \ref{line:rig-outerfor} with $u = 0$.

Since no element of $O\setminus G$ is added into the collection
when $j=0$, we can order $O\setminus G= \{o_1, o_2,\ldots\}$ such that the first $\ell$ elements
are not selected in any set before we get to $j = 1$,
the next $\ell$ elements are not selected in any set before we get
to $j = 2$,
and so on. Let $i \in \{1, \ldots, \ell \}$.
Let $A_{0,i}^{j}$ be the value of $A_{0,i}$ after $j$ elements are added into it,
and define $A_{0,i} = A_{0,i}^{k/\ell}$, the final value.
Finally, denote by $\delta_j$ the value $\delta_{x_{j, i}}(A_{0,i}^{j})$.
Then,
\begin{align*}
&\func{f}{O \cup A_{0,i}} - \func{f}{A_{0,i}}
\le \sum_{o \in O\setminus A_{0,i}}\marge{o}{A_{0,i}}\tag{submodularity}\\
&\le \sum_{o \in O\setminus G}\marge{o}{A_{0,i}}\tag{$G \subseteq A_{0,i}$}\\
&\le \sum_{l = 1}^{\ell} \marge{o_l}{A_{0,i}^{0}}
+ \sum_{l = \ell+1}^{2\ell} \marge{o_l}{A_{0,i}^{1}}+\ldots \tag{submodularity}\\
&\le \ell \sum_{j=1}^{k/\ell}\delta_j
=\ell (\func{f}{A_{0,i}}-\func{f}{G}),
\end{align*}
where the last inequality follows from the ordering of $O$ and the
selection of elements into the sets.
By summing up the above inequality with all $1\le i \le \ell$,
and the repeated application of monotonicity,
it holds that
\begin{align*}
\frac{\ell+1}{\ell} \sum_{i=1}^\ell f(A_{0,i}) &\ge 
f(O) + \ell f(G)
\end{align*}
Therefore, if we select a random set from $\{A_{0,i}: 1\le i \le \ell\}$,
by the above inequality,
Lemma holds
and we have probability $1/(\ell + 1 )$ of this happening.

\textbf{Case  $(O\setminus G) \cap \{a_1, \ldots, a_\ell\}\neq\emptyset$.}
Then $o_{\text{max}} \in \{a_1, \ldots, a_\ell\}$,
so $a_u = o_{\text{max}}$, for some $u \in 1, \ldots, \ell$.
we analyze the iteration $u$ of the \textbf{for}
loop on Line \ref{line:rig-outerfor}.
Similarly to the previous case,
let $i \in \uni$, 
define $A_{u,i}^{j}$ be the value of $A_{u,i}$ after we add $j$ elements into it,
and we will use $A_{u,i}$ for $A_{u,i}^{k/\ell}$,
Also, let $\delta_j = \marge{x_{j, i}}{A_{u,i}^{j-1}}$.
Finally, let $x_{1,i} = a_u$ and 
observe $A_{u,i}^{(1)} = G \cup \{a_u\}$.

Then, we can order $O\setminus G= \{o_1, o_2,\ldots\}$
such that: 1) for the first $\ell$ elements $\{o_l\}_{l=1}^\ell$, 
$\marge{o_l}{G} \le \marge{o_{\text{max}}}{G}=\delta_1$;
2) the next $\ell$ elements
$\{o_l\}_{l=\ell+1}^{2\ell}$ are not selected by any set before we get 
to $j = 2$, which implies that $\marge{o_l}{A_{u,i}^{1}} \le \delta_2$,
and so on.
Therefore, analagous to the the previous case,
we have that 
\begin{equation*}
\func{f}{O \cup A_{u,i}} - \func{f}{A_{u,i}}
\le \ell (\func{f}{A_{u,i}}-\func{f}{G}).
\end{equation*}
By summing up Inequality~\ref{inq:ig-rec-2} with all $1\le i \le \ell$,
and the repeated application of monotonicity,
it holds that 
\begin{align*}
\frac{\ell+1}{\ell} \sum_{i=1}^\ell f(A_{u,i}) &\ge 
f(O) + \ell f(G)
\end{align*}
where the second inequality follows from Inequality~\ref{inq:ig-deg-2}.
Therefore, if we select a random set from $\{A_{u,i}: 1\le i \le \ell\}$,
by the above inequality,
Lemma holds, and
this happens with probability $(\ell+1)^{-1}$.
\end{proof}
\begin{theorem}\label{thm:lrig-mon}
Let $\epsi \ge 0$, and
$(f,k)$ be an instance of \nmon,
with optimal solution value \opt.
Algorithm $\lrig$ outputs a set $G_\ell$
with $\oh{\epsi^{-2}nk}$ queries such that
$\opt \le \parth{\frac{e}{e-1}+\epsi}\ex{f(G_\ell)}$ with
probability $(\ell+1)^{-\ell}$, where $\ell = \frac{1}{1+\log\parth{1-\frac{1}{\frac{e}{e-1}+\epsi}}}-1$.
\end{theorem}
\begin{proof}[Proof of Theorem~\ref{thm:lrig-mon}]
Suppose that, at each iteration $m$, we select a random set from the 
correct guess of $u$ in Alg.~\ref{thm:ig}.
Then, the inequality in Theorem~\ref{thm:ig} holds for each iteration,
and it happens with probability $(\ell+1)^{-\ell}$.

By unfixing $G_{m-1}$ and utilizing the recursion of Theorem~\ref{thm:ig}, 
it holds that
\begin{align*}
\ex{\ff{O \cup G_m}} &\ge \left(1-\frac{1}{\ell}\right) \ff{O \cup G_{m-1}}\\
&\ge \left(1-\frac{1}{\ell}\right)^m f(O).
\end{align*}
Then, the approximation ratio can be bounded as follows
\begin{align*}
\ex{f(G_\ell)} &\ge \parth{1-\frac{1}{1+\ell}}
\ex{f(G_{\ell-1})} + \frac{1}{1+\ell}
 \ex{f(O)}\\
 &\ge \parth{1-\parth{1-\frac{1}{\ell+1}}^\ell}f(O) \\
&\ge \parth{1-e^{-1+\frac{1}{\ell+1}}}f(O) \tag{$1-x \le e^{-x}$}\\
&\ge \frac{1}{\frac{e}{e-1}+\epsi}f(O) \tag{$\ell = \frac{1}{1+\log\parth{1-\frac{1}{\frac{e}{e-1}+\epsi}}}-1$}
\end{align*}

\end{proof}
  \section{Empirical Evaluation} \label{apx:exp}
\subsection{Environment}
All experiments were run on a linux server
running Ubuntu 20.04, with 2 $\times$ Intel(R) Xeon(R) Gold 5218R CPU @ 2.10GHz and 504 GB RAM.
\subsection{Implementation and Parameter Settings} \label{sec:apx-impl}
All algorithms were implemented in C++ and used
the same code for evaluation of the application
oracle. An optimized marginal gain computation
was available to the algorithms that could
benefit from such optimization and when the application permitted such optimization.

All algorithms used lazy evaluations whenever possible as follows.
Suppose $\marge{x}{S}$ has already been computed, and the algorithm
needs to check if $\marge{x}{T} \ge \tau$, for some
$\tau \in \mathbb R$ and $T \supseteq S$. Then if $\marge{x}{S} < \tau$, this evaluation may be safely skipped due to the submodularity
of $f$. The single-pass streaming algorithms evaluated do not benefit from lazy evaluations, except
for those algorithms (\qssp and \aefns)
that use post-processing.

The accuracy parameter
$\epsi$ of each algorithm is set to $0.1$.
The parameter $b$ of \qss is set to $1$;
while for \qssp, $b$ is set to $0.1$.
These choices for $b$ worked well empirically,
although they yield worse theoretical
guarantees than choosing $b = 2\sqrt 2$
as discussed in Section \ref{sec:qs}.
The smaller value of $b$ for \qssp
yields a larger universe for post-processing,
which empirically improves the solution value. 

The algorithm \lrvz is an idealized implementation
that uses more memory (that is, exceeds the memory
bound required for a streaming
algorithm) and is faster than the
actual algorithm described in \citet{Liu2021}; the
same optimization is used in the experimental
evaluation of \citet{Liu2021}. Briefly, the histories
required by the algorithm are saved instead of computed
on the fly. For more information, see comments
in the source code
released with the paper of \citet{Liu2021}.
Since we do not compare the memory usage of the
algorithms, this optimization only gives an
advantage to \lrvz in our experimental comparison.

As mentioned in Section \ref{sec:exp},
both \qssp and \aefns used \mpl
for post-processing. Recall that
\mpl requires an input of $\Gamma$
and $\alpha$. \qssp used its solution value (before
  post-processing) and its approximation ratio
  for $\Gamma$ and $\alpha$, respectively (which means that \mpl will run in linear-time for its post-processing). However, \aefns does not have an approximation ratio before post-processing, so the maximum singleton value and $k$ were used for $\Gamma$
  and $\alpha$, respectively.
  Both algorithms used their respective value for accuracy parameter $\epsi$ for the same parameter in \mpl.

\begin{figure}
  \centering
  \subfigure[\er]{ 
    \includegraphics[width=0.30\textwidth,height=0.18\textheight]{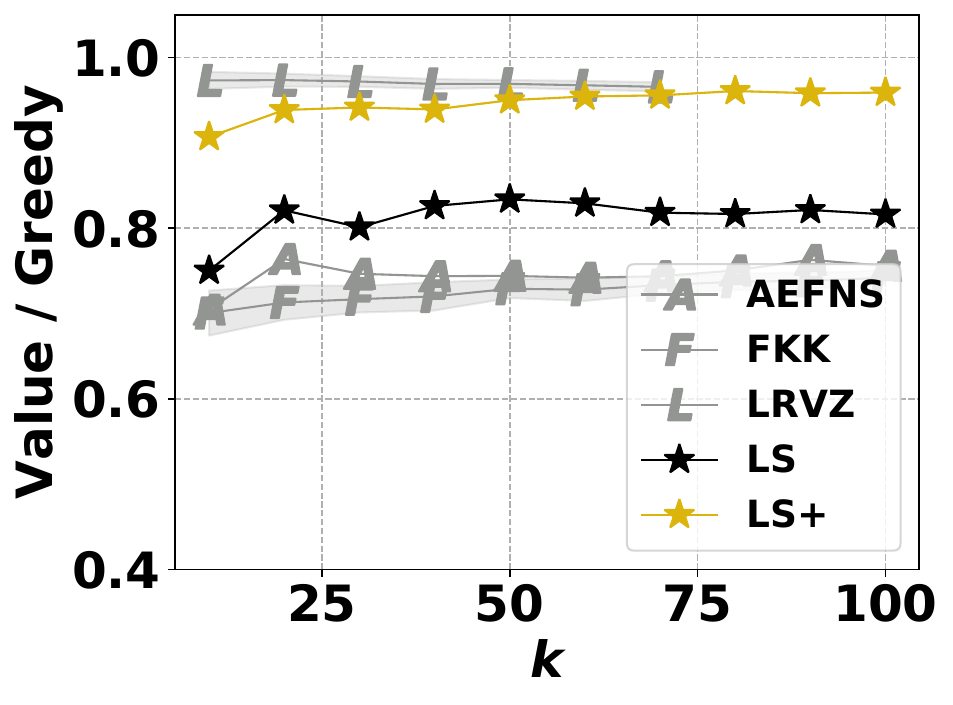}
  }
  \subfigure[\ba]{ 
    \includegraphics[width=0.30\textwidth,height=0.18\textheight]{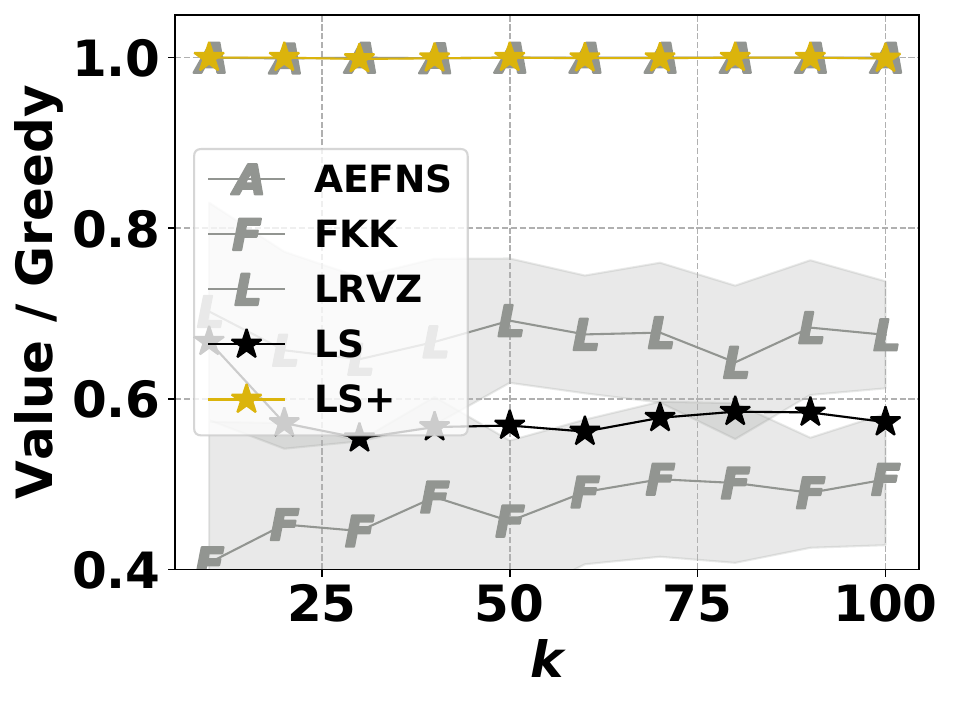}
  }
  \subfigure[\fb]{ 
    \includegraphics[width=0.30\textwidth,height=0.18\textheight]{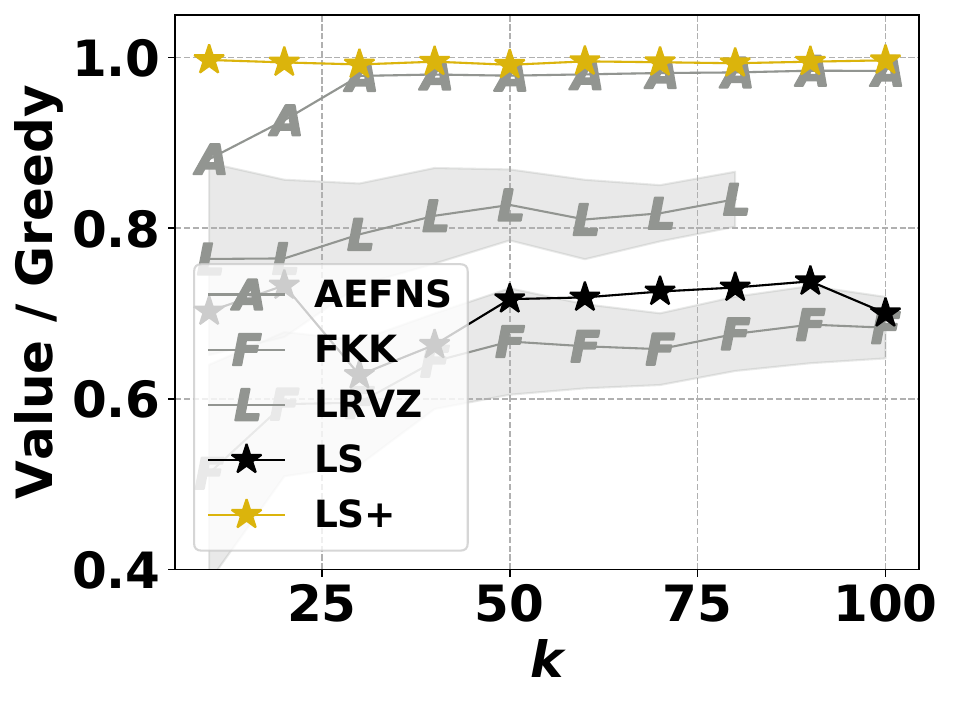}
  }
  
  \subfigure[\slashdot]{ 
    \includegraphics[width=0.30\textwidth,height=0.18\textheight]{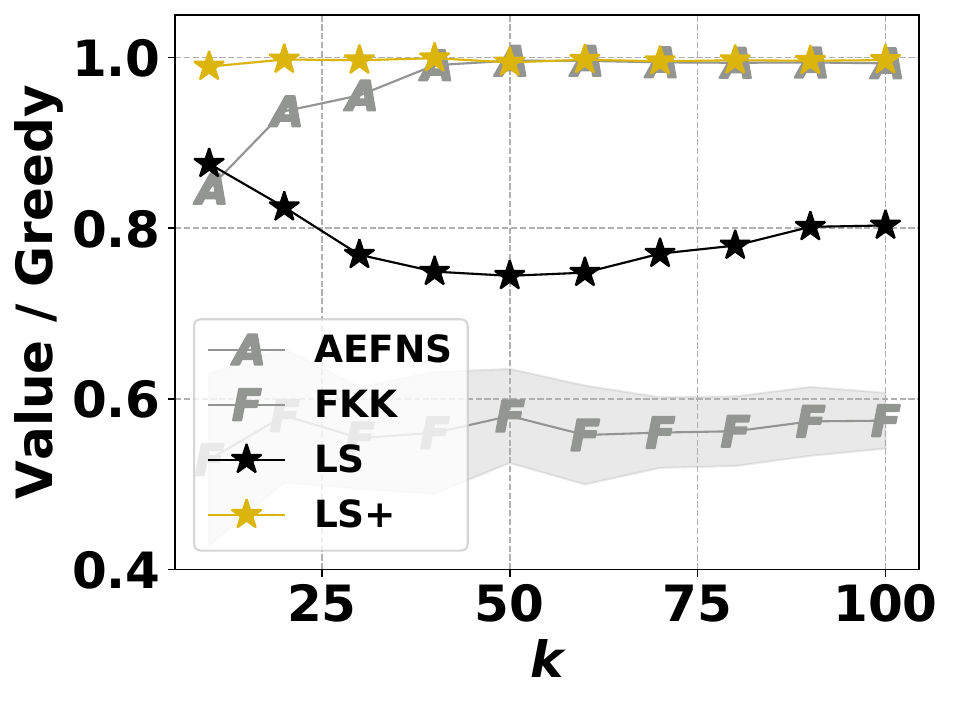}
  }
  \subfigure[\pokec]{ 
    \includegraphics[width=0.30\textwidth,height=0.18\textheight]{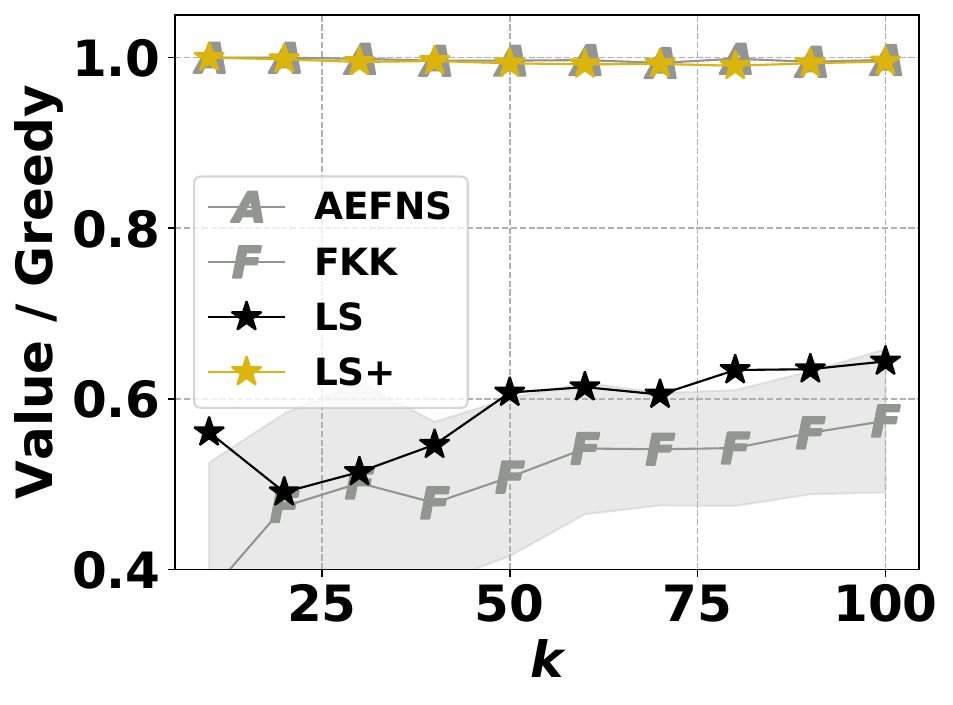}
  }
  \caption{Solution value vs. $k$ for single-pass algorithms for the \maxcut application on each dataset.} \label{fig:maxcut-val}
\end{figure}
\begin{figure}
  \centering
  \subfigure[\er]{ 
    \includegraphics[width=0.30\textwidth,height=0.18\textheight]{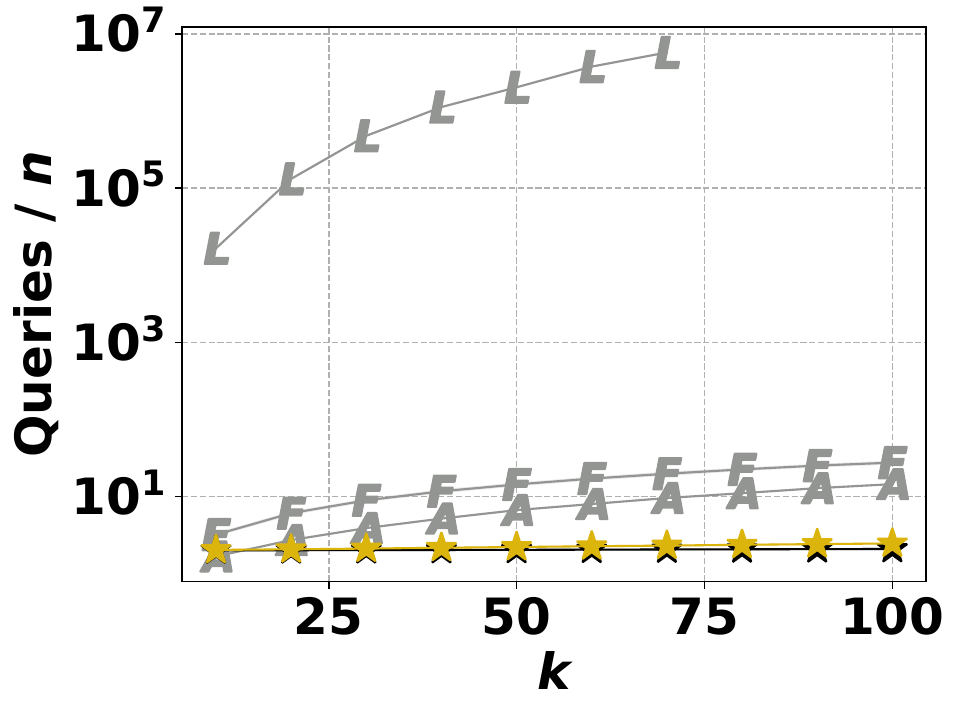}
  }
  \subfigure[\ba]{ 
    \includegraphics[width=0.30\textwidth,height=0.18\textheight]{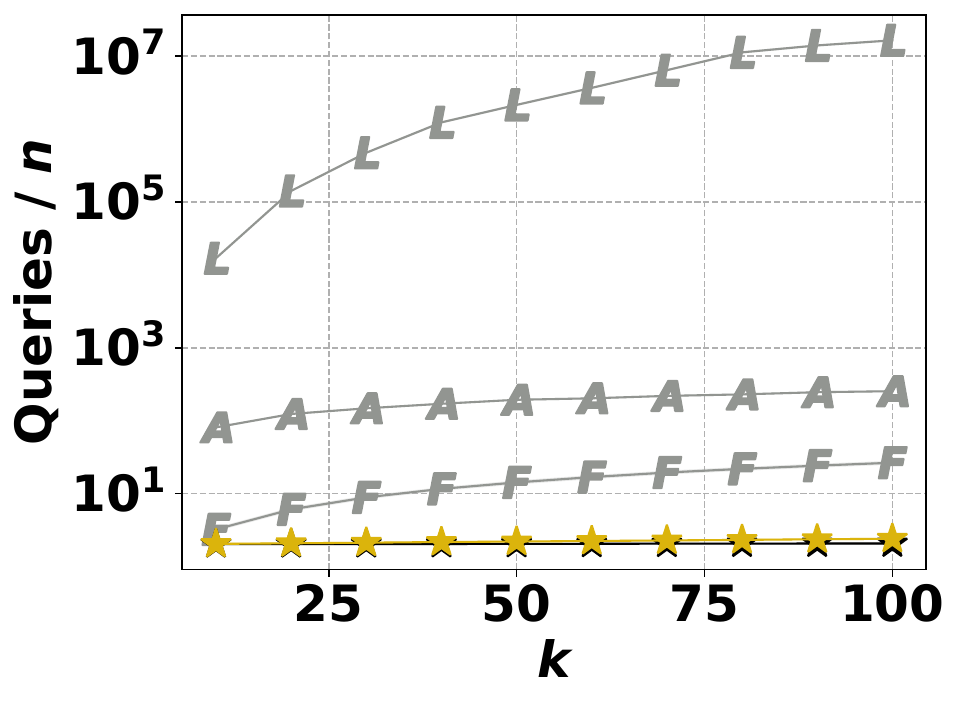}
  }
  \subfigure[\fb]{ 
    \includegraphics[width=0.30\textwidth,height=0.18\textheight]{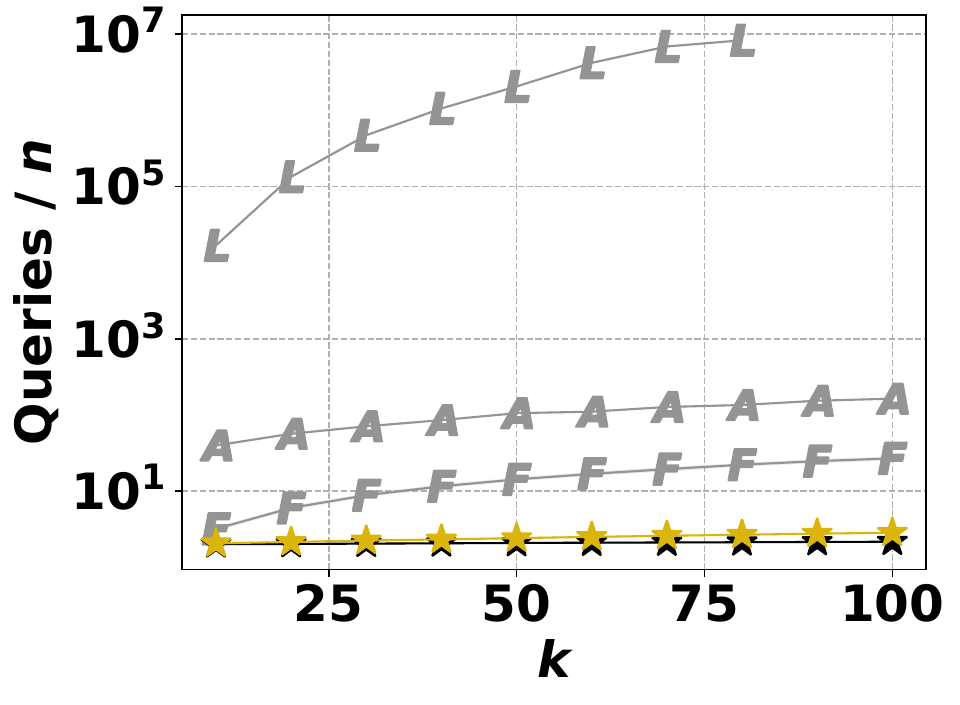}
  }
  
  \subfigure[\slashdot]{ 
    \includegraphics[width=0.30\textwidth,height=0.18\textheight]{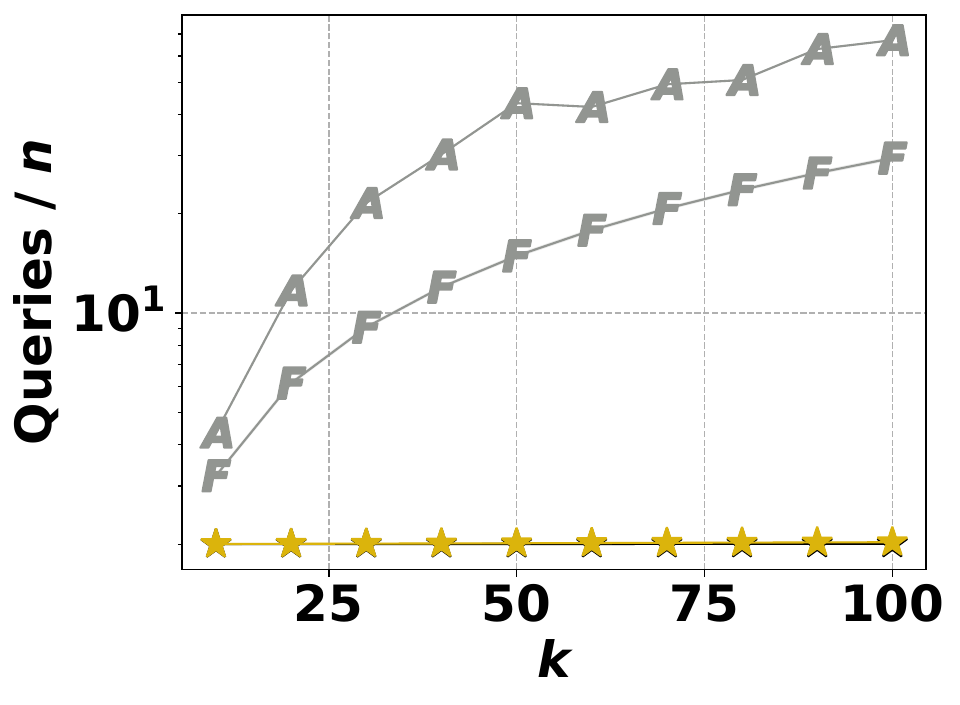}
  }
  \subfigure[\pokec]{ 
    \includegraphics[width=0.30\textwidth,height=0.18\textheight]{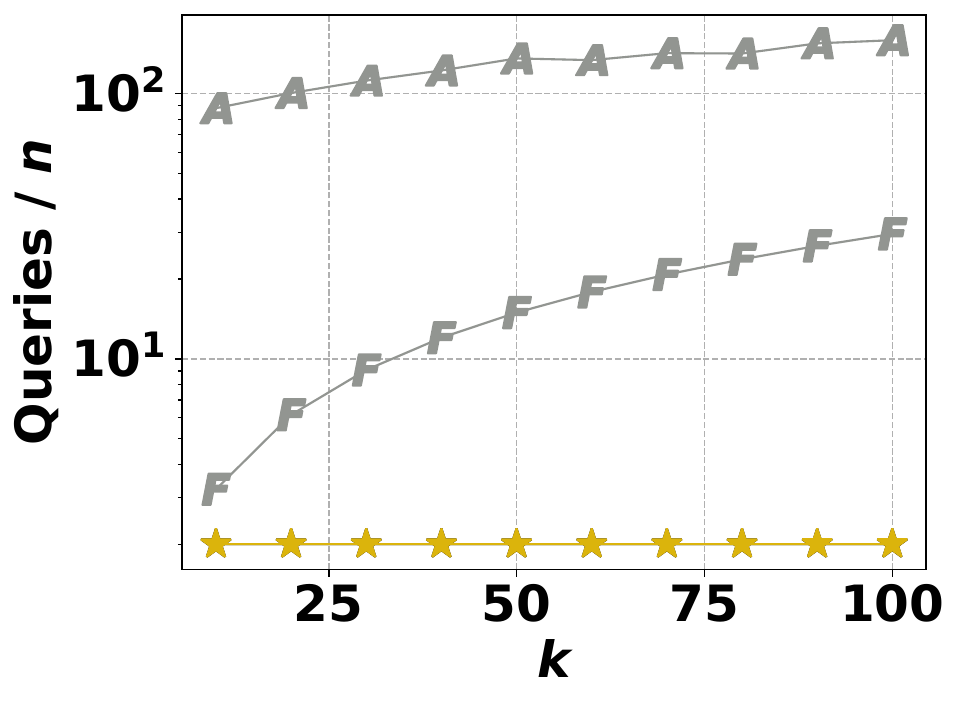}
  }
  \caption{Oracle queries vs. $k$ for single-pass algorithms for the \maxcut application on each dataset.} \label{fig:maxcut-query}
\end{figure}
\begin{figure}
  \centering
  \subfigure[\er]{ 
    \includegraphics[width=0.30\textwidth,height=0.18\textheight]{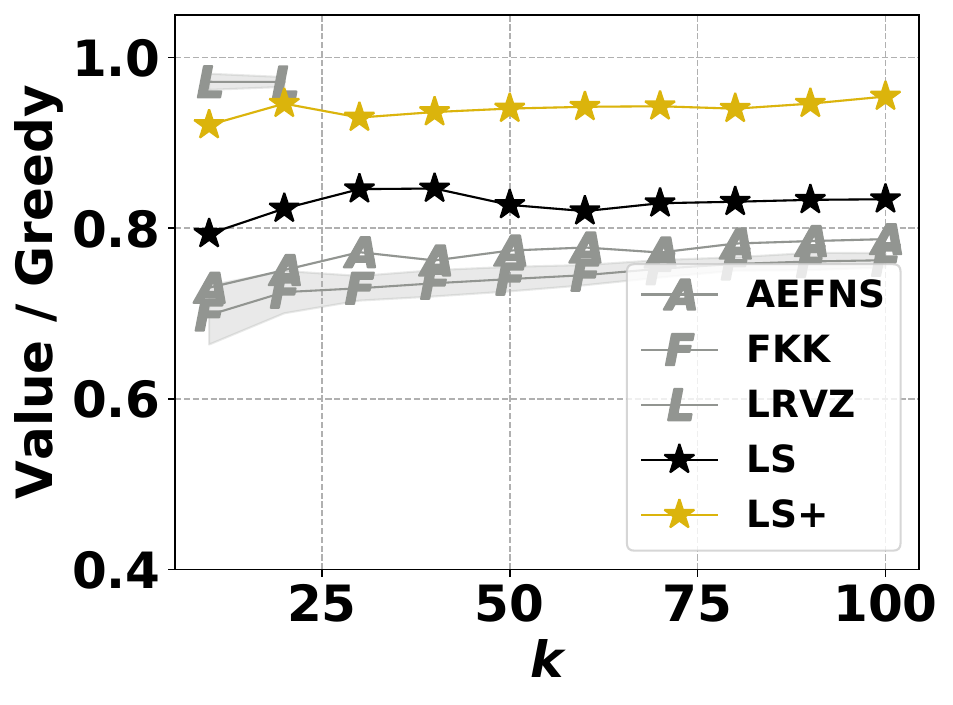}
  }
  \subfigure[\ba]{ 
    \includegraphics[width=0.30\textwidth,height=0.18\textheight]{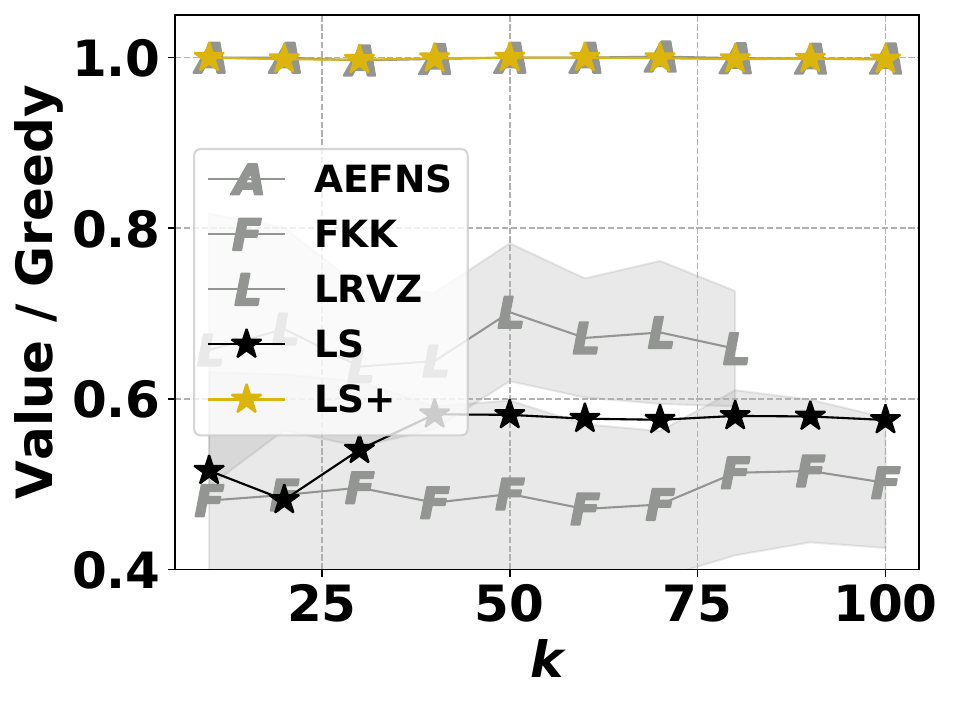}
  }
  \subfigure[\fb]{ 
    \includegraphics[width=0.30\textwidth,height=0.18\textheight]{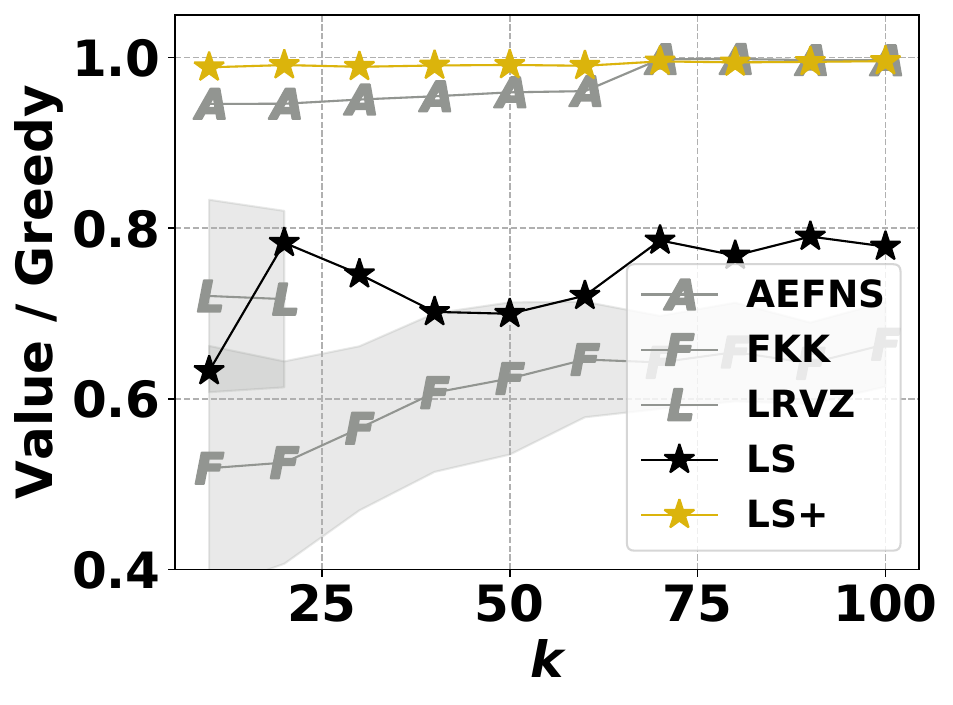}
  }
  
  \subfigure[\slashdot]{ 
    \includegraphics[width=0.30\textwidth,height=0.18\textheight]{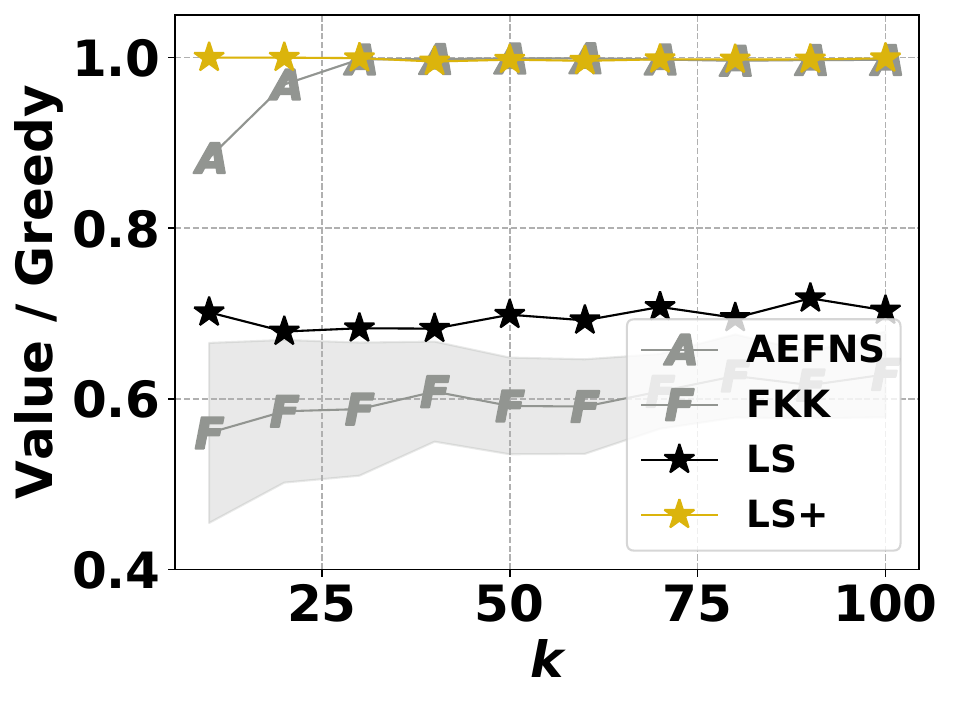}
  }
  \subfigure[\pokec]{ 
    \includegraphics[width=0.30\textwidth,height=0.18\textheight]{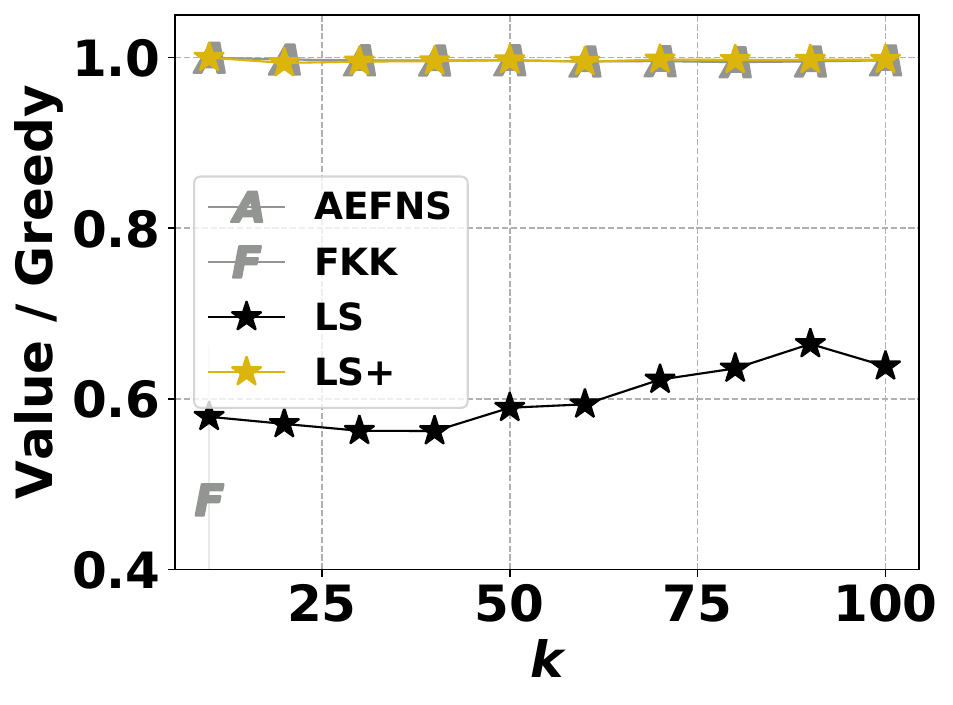}
  }
  \caption{Solution value vs. $k$ for single-pass algorithms for the \revmax application on each dataset.} \label{fig:revmax-val}
\end{figure}
\begin{figure}
  \centering
  \subfigure[\er]{ 
    \includegraphics[width=0.30\textwidth,height=0.18\textheight]{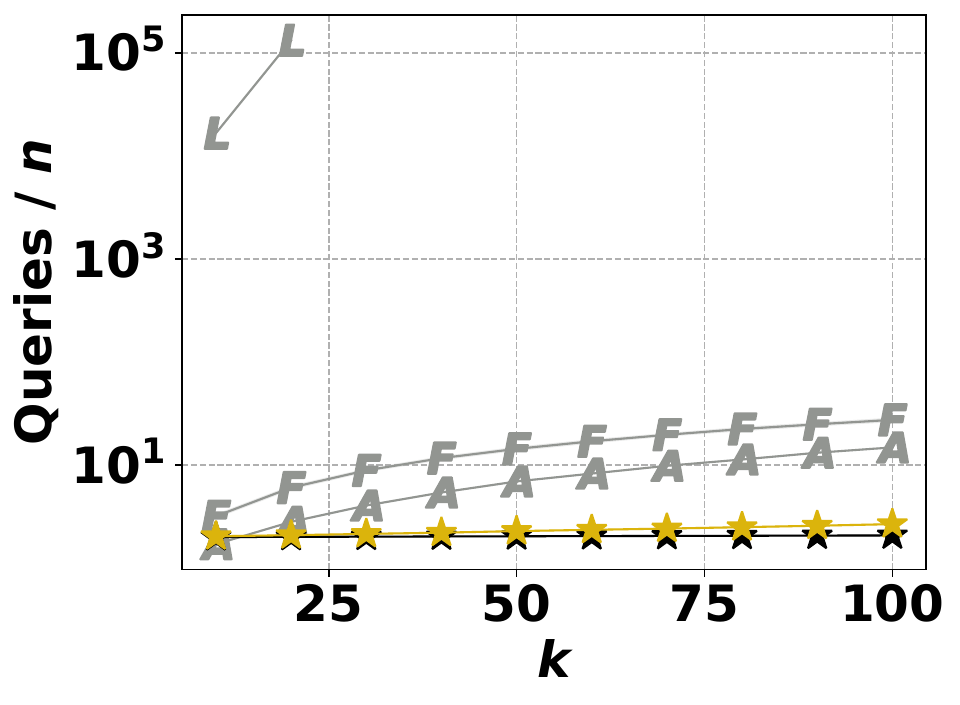}
  }
  \subfigure[\ba]{ 
    \includegraphics[width=0.30\textwidth,height=0.18\textheight]{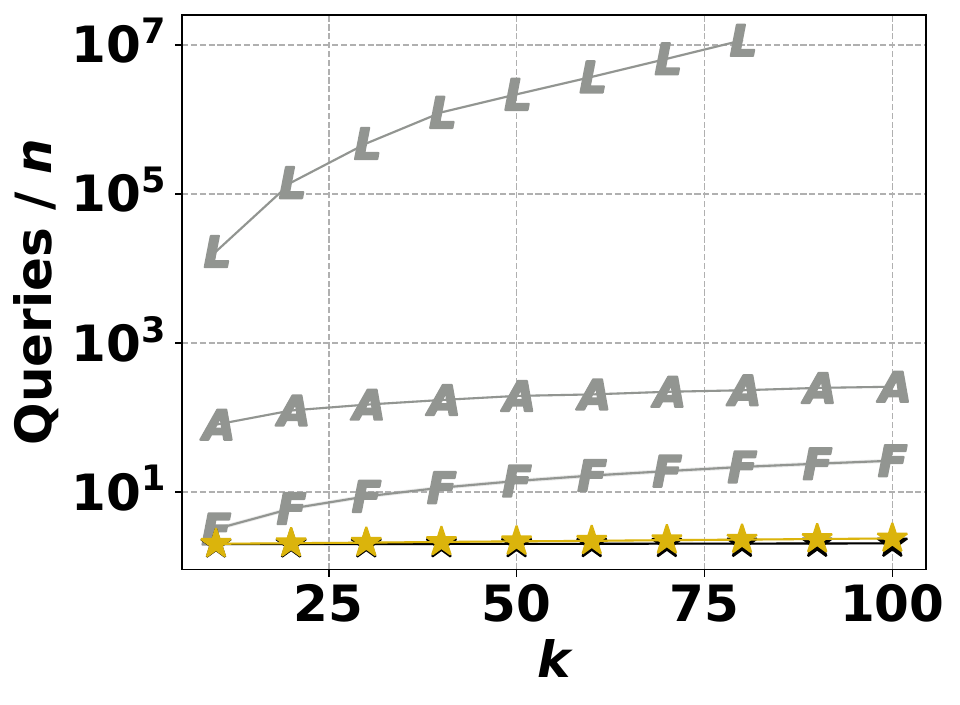}
  }
  \subfigure[\fb]{ 
    \includegraphics[width=0.30\textwidth,height=0.18\textheight]{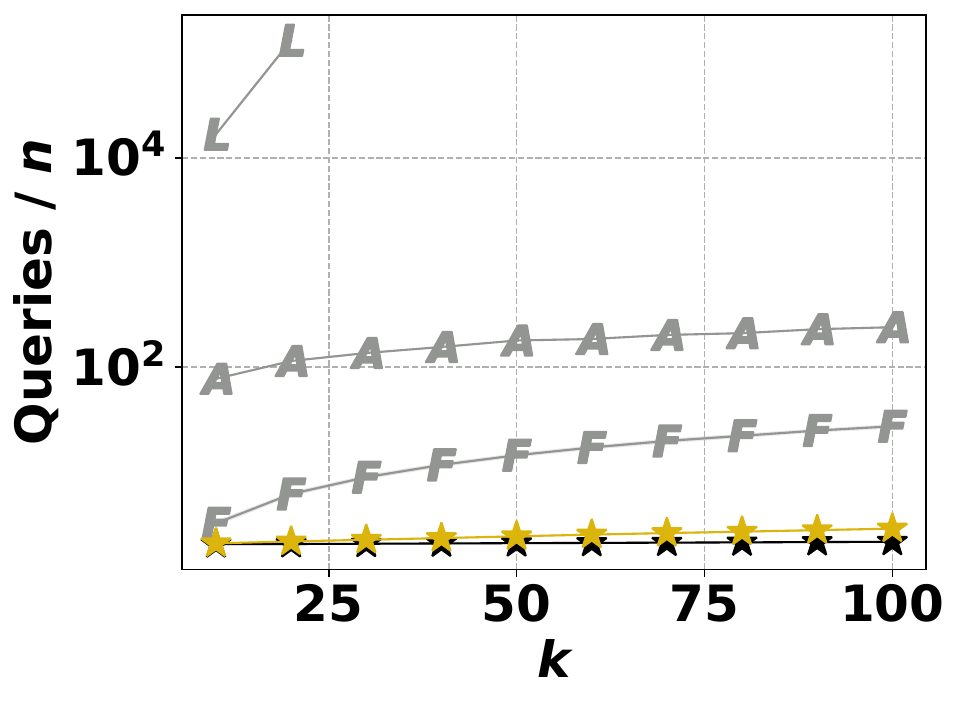}
  }
  
  \subfigure[\slashdot]{ 
    \includegraphics[width=0.30\textwidth,height=0.18\textheight]{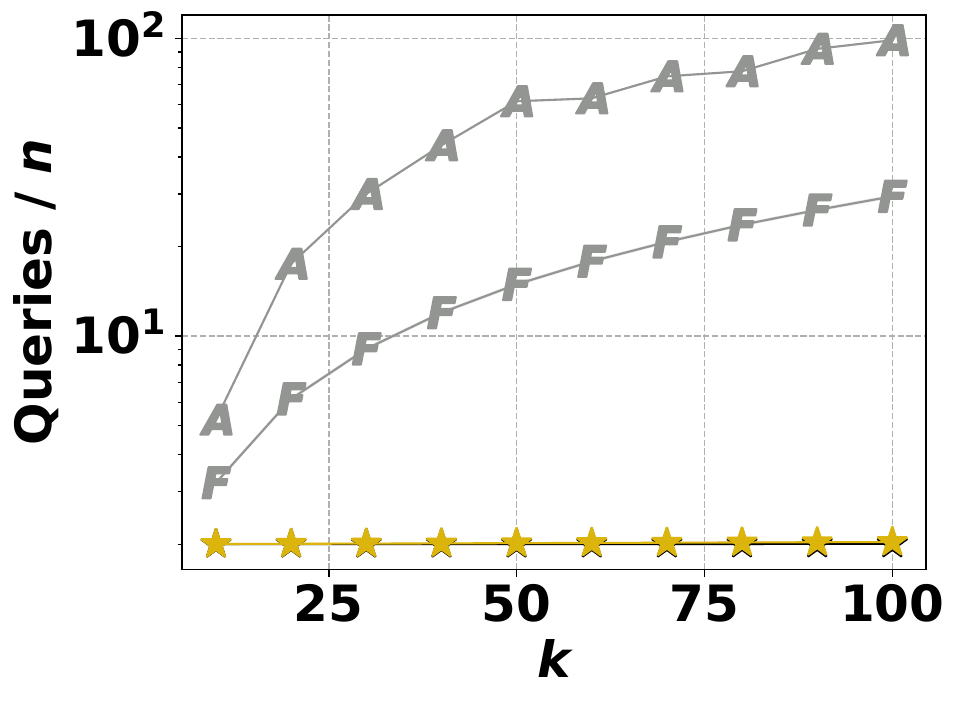}
  }
  \subfigure[\pokec]{ 
    \includegraphics[width=0.30\textwidth,height=0.18\textheight]{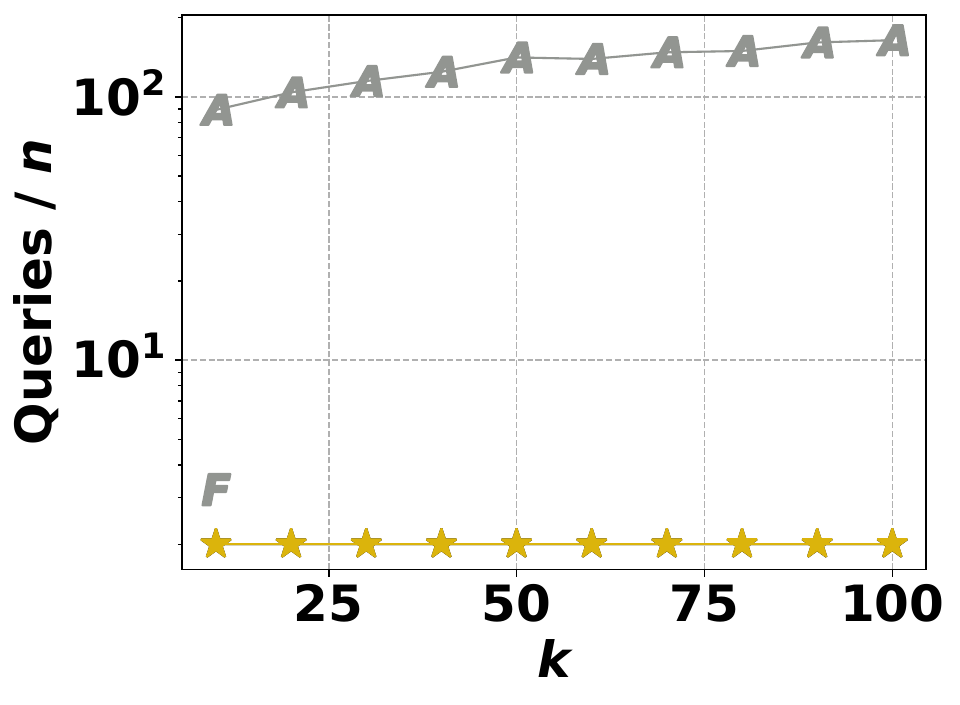}
  }
  \caption{Oracle queries vs. $k$ for single-pass algorithms for the \revmax application on each dataset.} \label{fig:revmax-query}
\end{figure}
\begin{figure}
  \subfigure[\maxcut value, \er]{ \label{mt:val-sd}
    \includegraphics[width=0.32\textwidth,height=0.2\textheight]{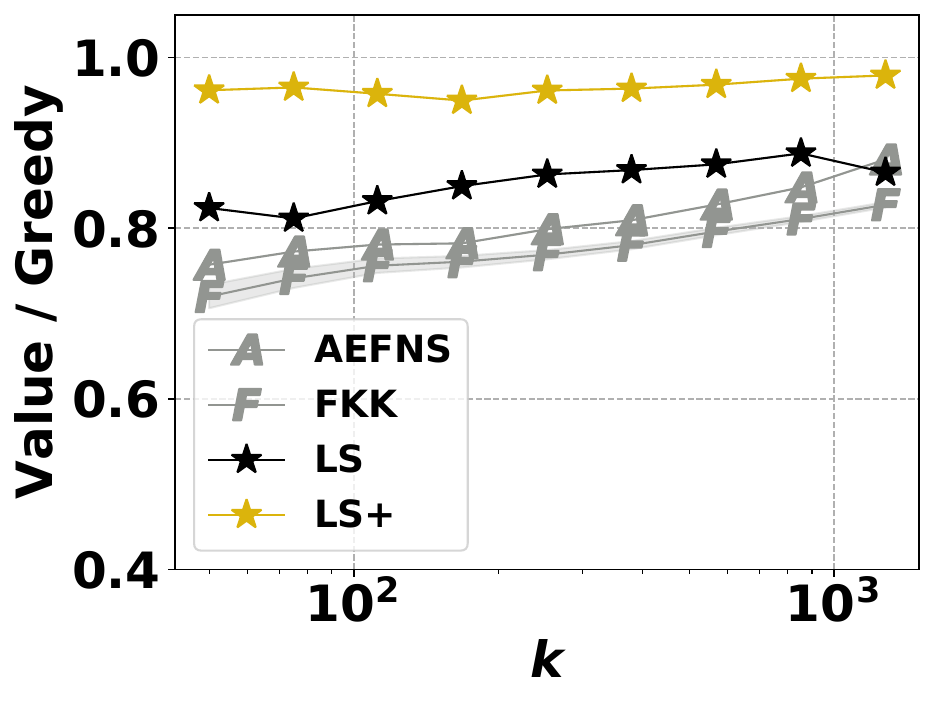}
  }
  \subfigure[\maxcut queries, \er]{ \label{mt:query-sd}
    \includegraphics[width=0.32\textwidth,height=0.2\textheight]{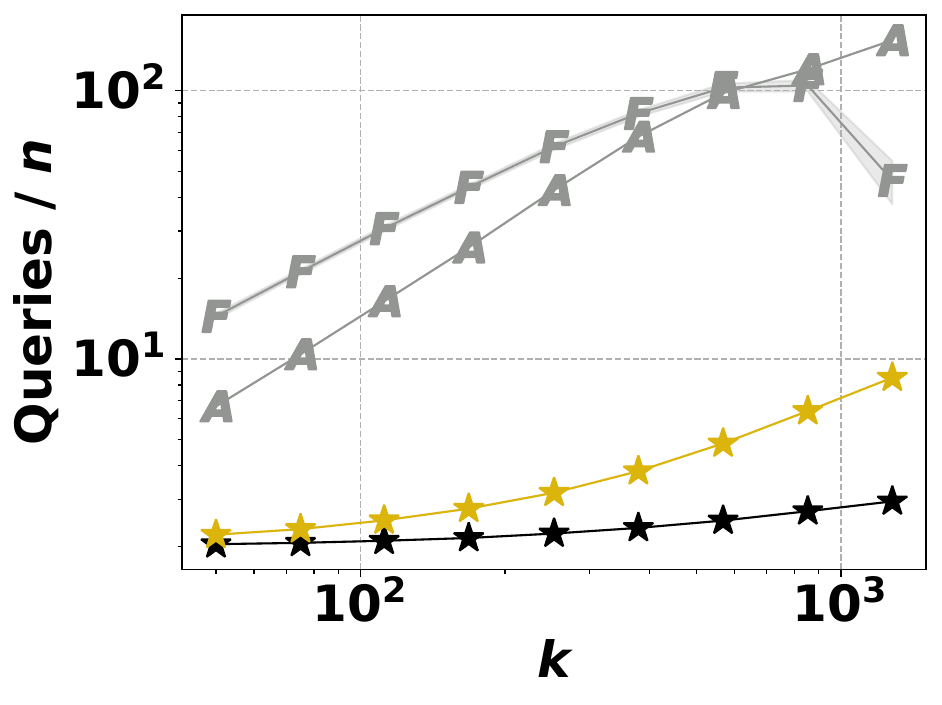}
  }
  \subfigure[\maxcut memory, \er]{ \label{mt:mem-sd}
    \includegraphics[width=0.32\textwidth,height=0.2\textheight]{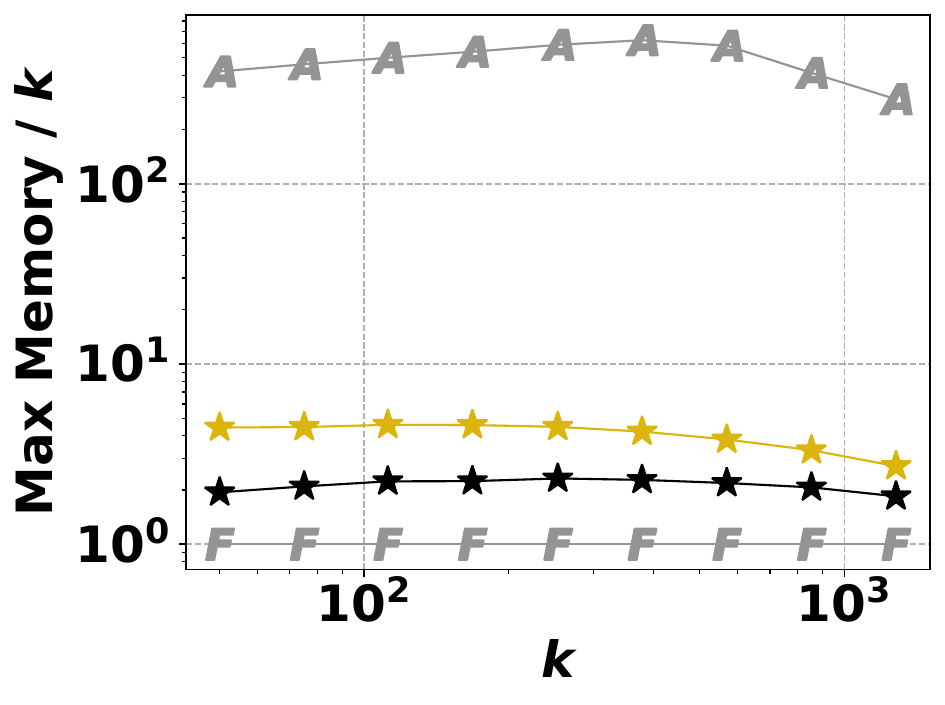}
  }
  \subfigure[\revmax value, \er]{ \label{mt:val-sd}
    \includegraphics[width=0.32\textwidth,height=0.2\textheight]{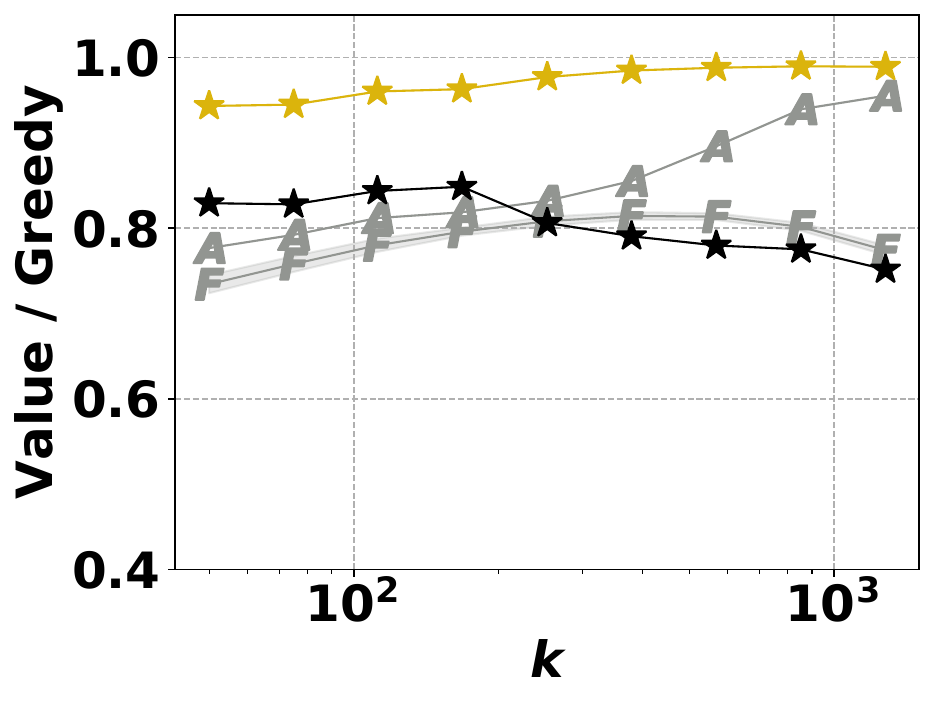}
  }
  \subfigure[\revmax queries, \er]{ \label{mt:query-sd}
    \includegraphics[width=0.32\textwidth,height=0.2\textheight]{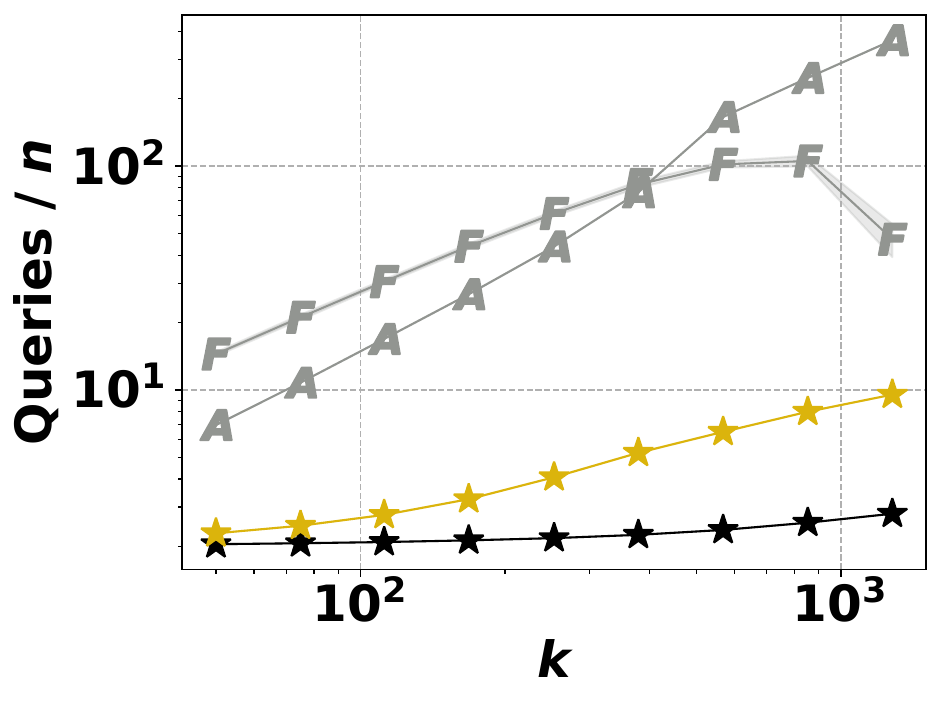}
  }
  \subfigure[\revmax memory, \er]{ \label{mt:mem-sd}
    \includegraphics[width=0.32\textwidth,height=0.2\textheight]{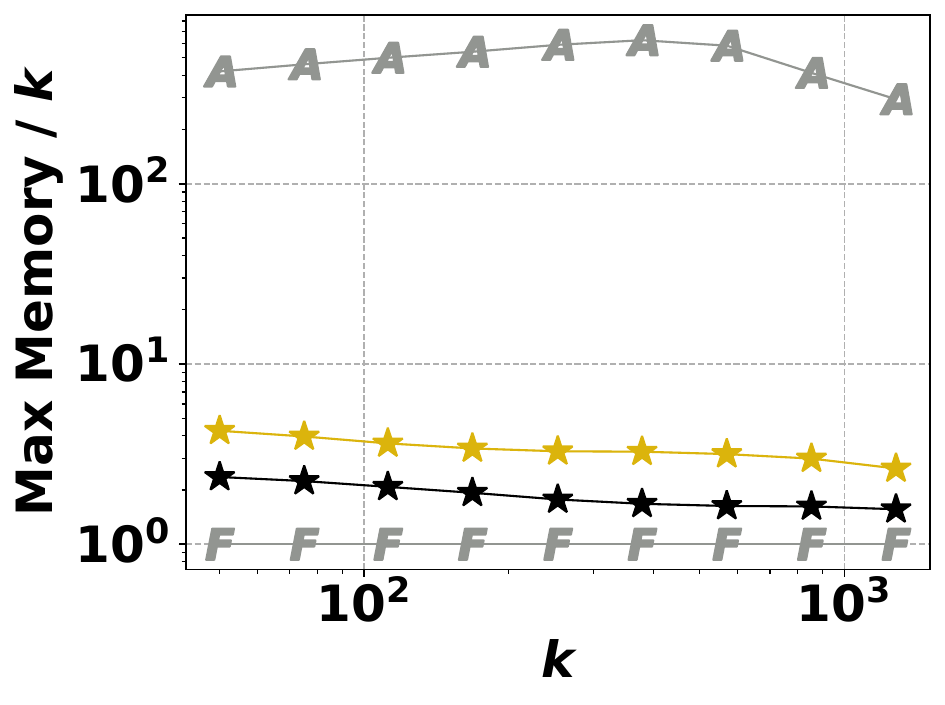}
  }
  \caption{Evaluation of single-pass streaming algorithms on the \er dataset (n=5,000) with the \maxcut and \revmax objectives and large $k$ values}
  \label{fig:er}
\end{figure}
\begin{figure}
  \subfigure[\maxcut value, \slashdot]{ \label{mt:val-sd}
    \includegraphics[width=0.32\textwidth,height=0.2\textheight]{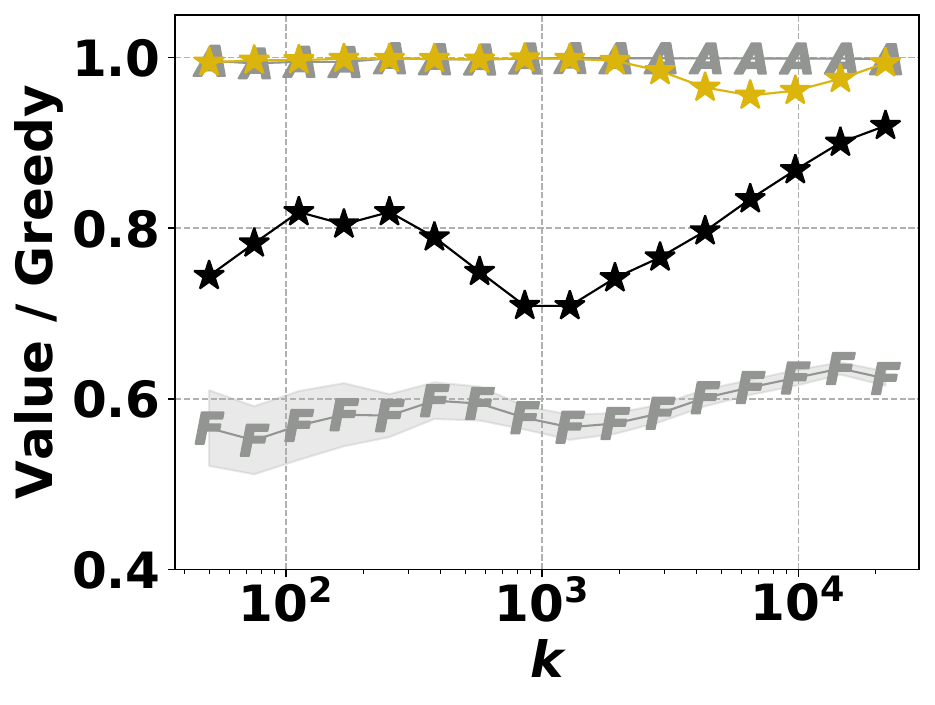}
  }
  \subfigure[\maxcut queries, \slashdot]{ \label{mt:query-sd}
    \includegraphics[width=0.32\textwidth,height=0.2\textheight]{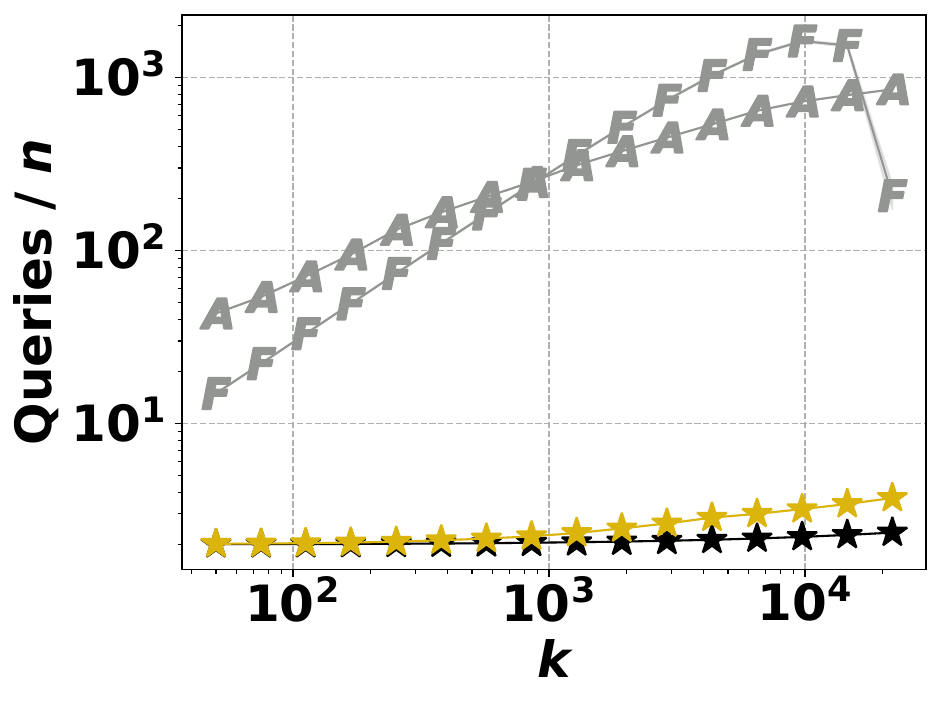}
  }
  \subfigure[\maxcut memory, \slashdot]{ \label{mt:mem-sd}
    \includegraphics[width=0.32\textwidth,height=0.2\textheight]{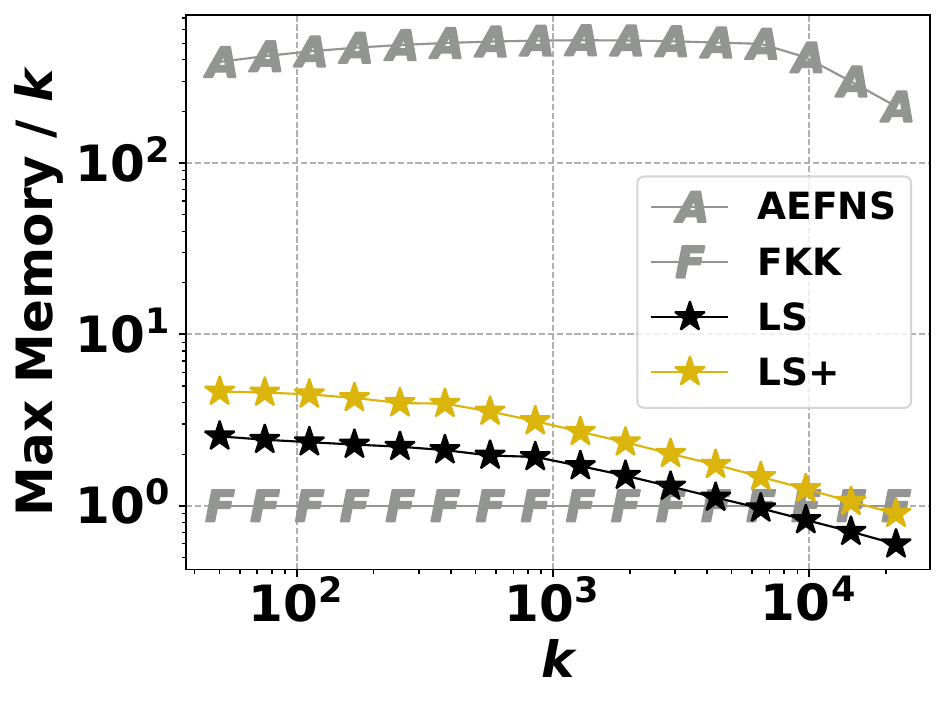}
  }
  \subfigure[\revmax value, \slashdot]{ \label{mt:val-sd}
    \includegraphics[width=0.32\textwidth,height=0.2\textheight]{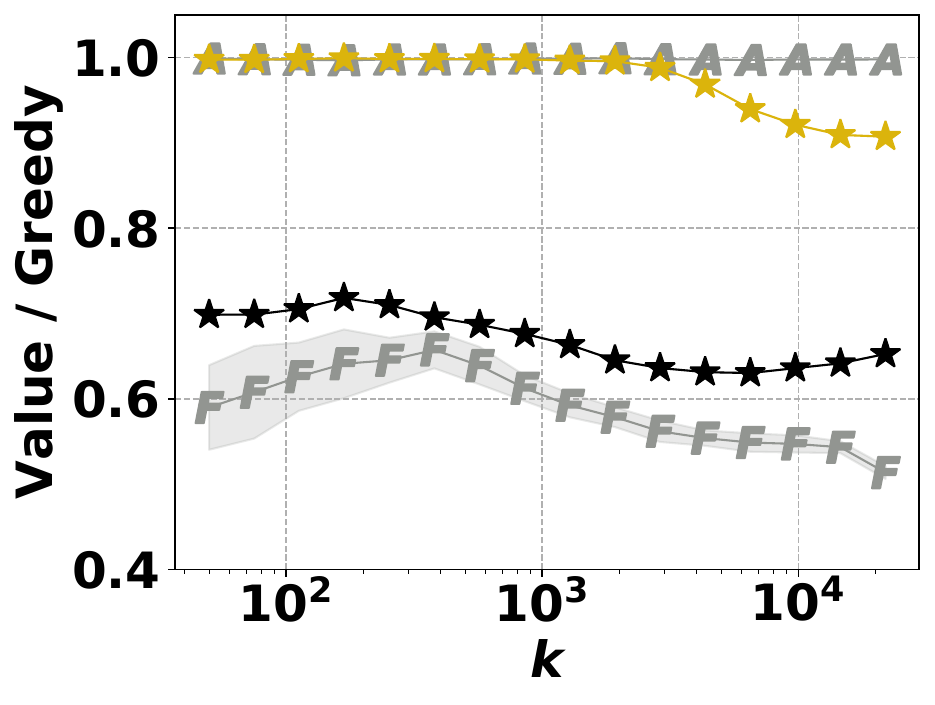}
  }
  \subfigure[\revmax queries, \slashdot]{ \label{mt:query-sd}
    \includegraphics[width=0.32\textwidth,height=0.2\textheight]{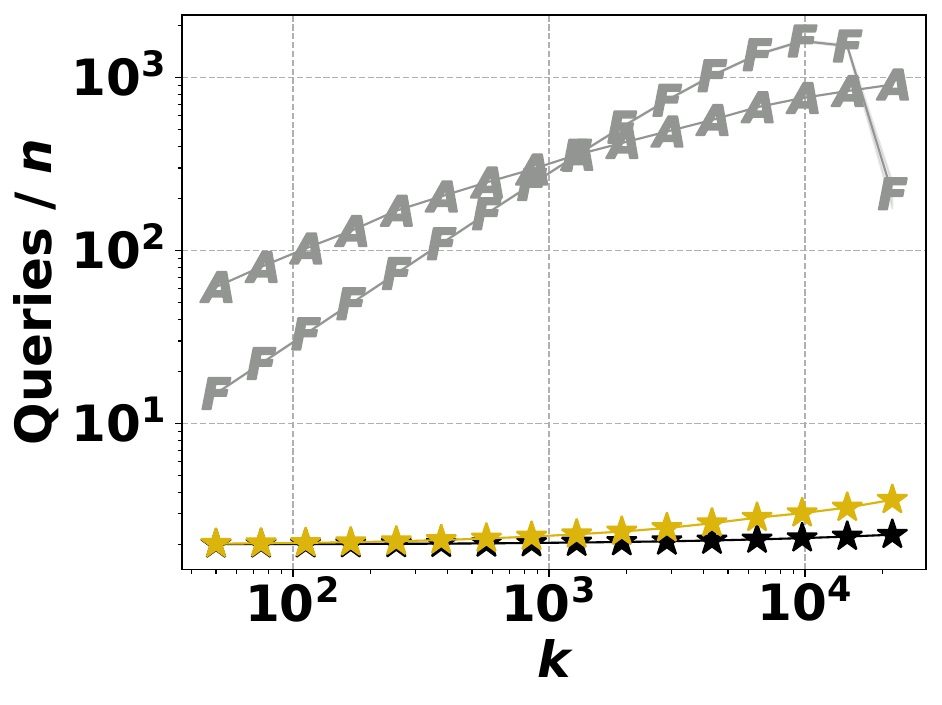}
  }
  \subfigure[\revmax memory, \slashdot]{ \label{mt:mem-sd}
    \includegraphics[width=0.32\textwidth,height=0.2\textheight]{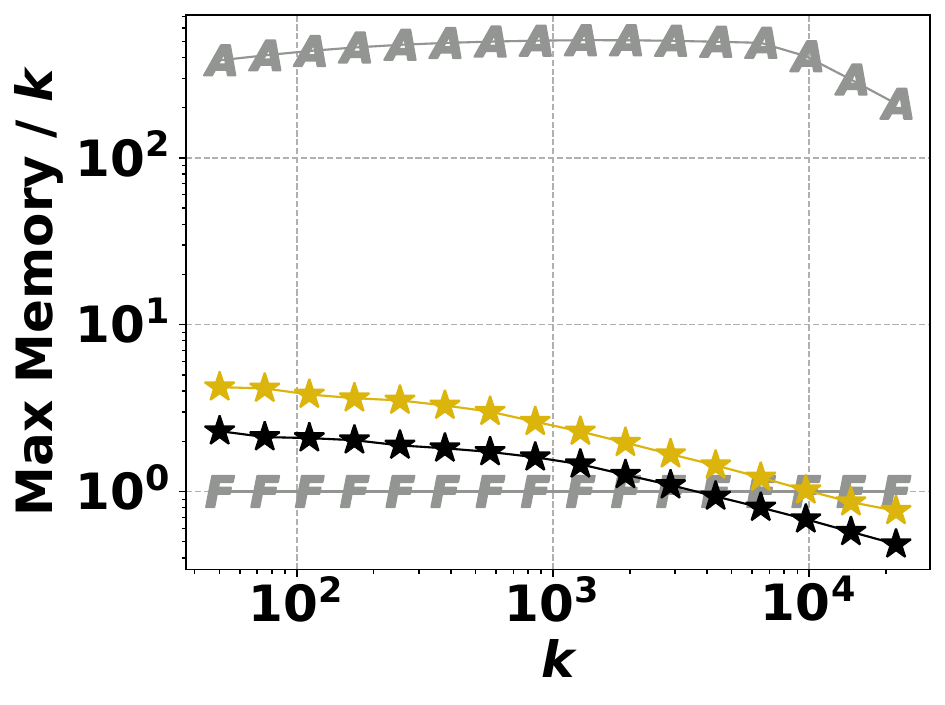}
  }
  \caption{Evaluation of single-pass streaming algorithms on the \slashdot dataset (n=77,360) with the \maxcut and \revmax objectives and large $k$ values}
  \label{fig:slashdot}
\end{figure}
\begin{figure}
  \subfigure[\maxcut value, \fb]{ \label{mt:val-sd}
    \includegraphics[width=0.32\textwidth,height=0.2\textheight]{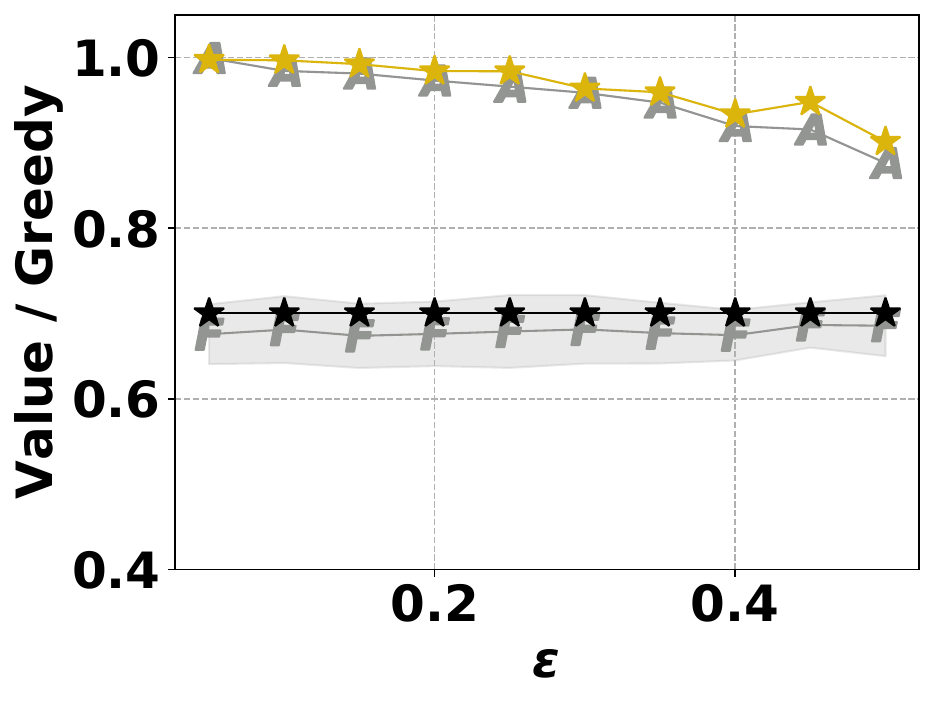}
  }
  \subfigure[\maxcut queries, \fb]{ \label{mt:query-sd}
    \includegraphics[width=0.32\textwidth,height=0.2\textheight]{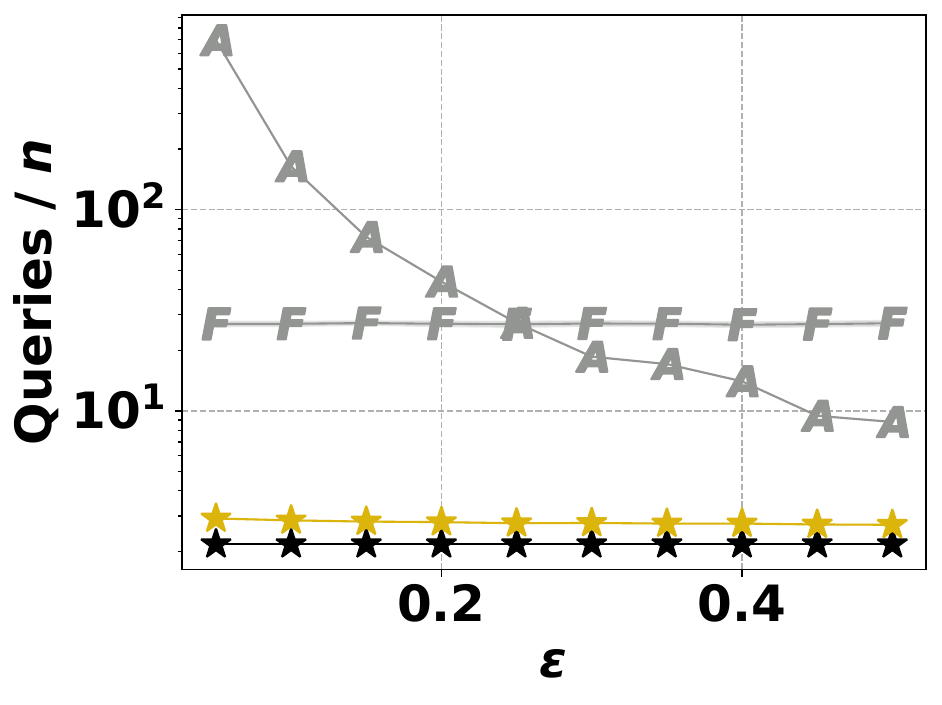}
  }
  \subfigure[\maxcut memory, \fb]{ \label{mt:mem-sd}
    \includegraphics[width=0.32\textwidth,height=0.2\textheight]{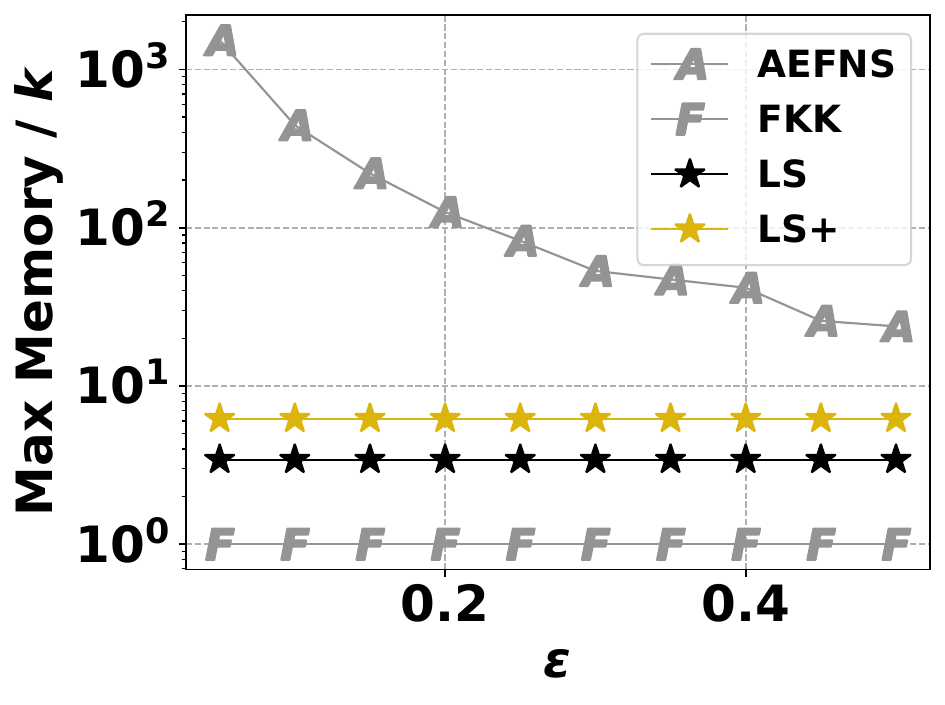}
  }
  \caption{Evaluation of single-pass streaming algorithms on the \fb dataset (n=4,039) with the \maxcut objective and different $\epsi$ values}
  \label{fig:epsi-fb}
\end{figure}
\begin{figure}
  \subfigure[\maxcut value, \slashdot]{ \label{mt:val-sd}
    \includegraphics[width=0.32\textwidth,height=0.2\textheight]{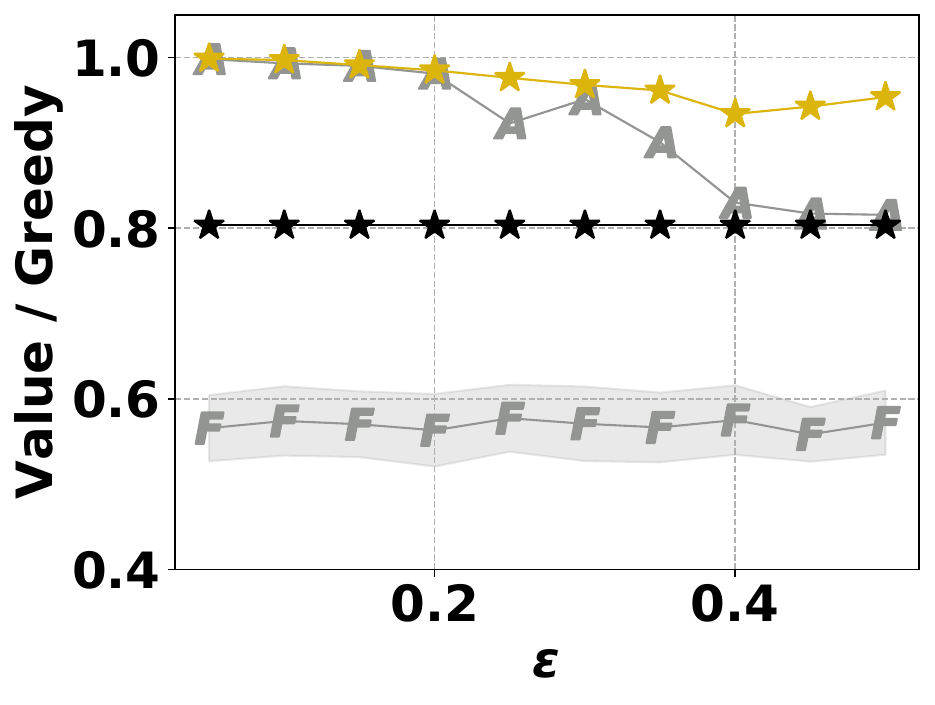}
  }
  \subfigure[\maxcut queries, \slashdot]{ \label{mt:query-sd}
    \includegraphics[width=0.32\textwidth,height=0.2\textheight]{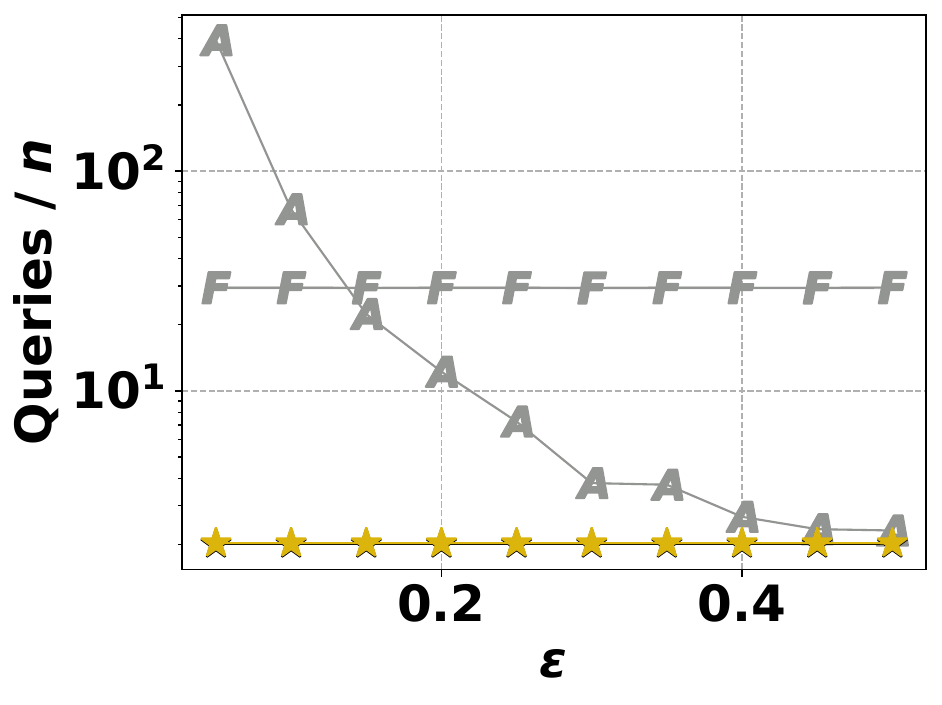}
  }
  \subfigure[\maxcut memory, \slashdot]{ \label{mt:mem-sd}
    \includegraphics[width=0.32\textwidth,height=0.2\textheight]{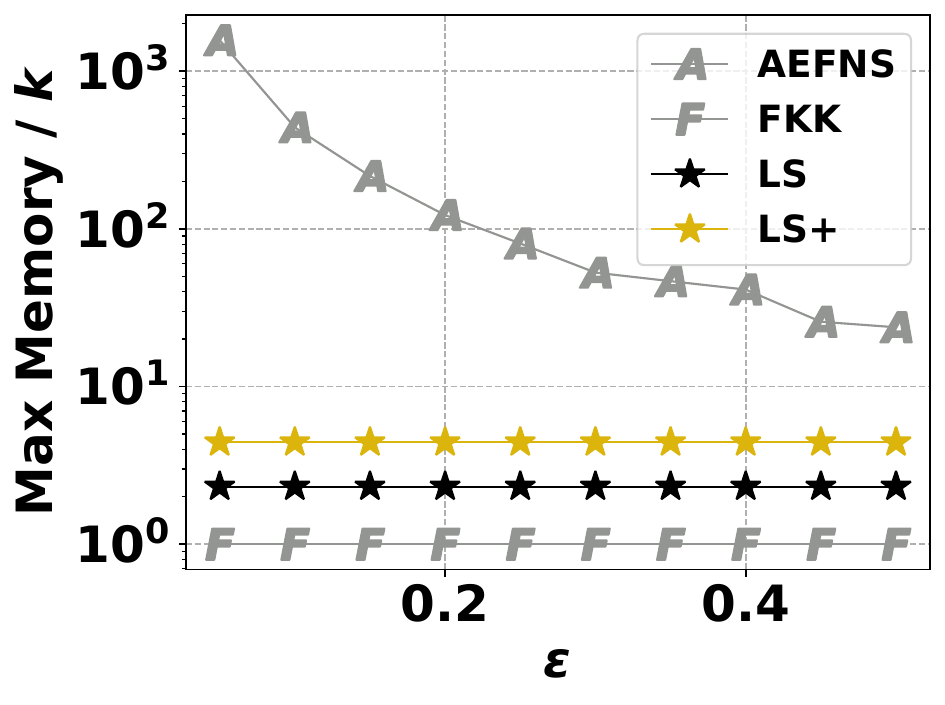}
  }
  \caption{Evaluation of single-pass streaming algorithms on the \slashdot dataset (n=77,360) with the \maxcut objective and different $\epsi$ values}
  \label{fig:epsi-sd}
\end{figure}
\subsection{Applications and Datasets}
The cardinality-constrained maximum cut function is defined as follows.
Given graph $G = (V, E)$, and nonnegative edge weight $w_{ij}$ on each edge
$(i,j) \in E$. For $S \subseteq V$, let
$$f(S) = \sum_{i \in V \setminus S} \sum_{j \in S} w_{ij}.$$
In general, this is a non-monotone, submodular function. 

The revenue maximization objective is defined as follows.
Let graph $G = (V, E)$ represent a social network, 
with nonnegative edge weight $w_{ij}$ on each edge
$(i,j) \in E$.
We use the concave graph model introduced by \citet{Hartline2008}.
In this model, each user $i \in V$ is associated with a non-negative,
concave function $f_i: \reals \to \reals$. The value $v_i(S) = f_i( \sum_{j \in S} w_{ij} )$
encodes how likely the user $i$ is to buy a product if the set $S$ has adopted it.
Then the total revenue for seeding a set $S$ is 
$$f(S) = \sum_{i \in V \setminus S} f_i\left( \sum_{j \in S} w_{ij} \right).$$
This is a non-monotone, submodular function. In our implementation,
each edge weight $w_{ij} \in (0,1)$ is chosen uniformly randomly; further,
$f_i( \cdot ) = ( \cdot )^{\alpha_i}$, where $\alpha_i \in (0,1)$ is chosen
uniformly randomly for each user $i \in V$.

The image summarization objective function is defined as follows.
Given a set $I$ of images, define nonnegative edge weight $s_{ij}$ as the cosine similarity
of the pixel vectors for images $i$ and $j$ in $I$.
\[f(S) = \sum_{i \in \uni} \max_{j \in S}s_{ij} - \frac{1}{n}\sum_{i \in S}\sum_{j \in S}s_{ij}.\]
The first term tries to ensure that the set $S$ is a good summary of the dataset,
while the second promotes diversity within the summary itself. 
This is a non-monotone, submodular objective function.
In this paper, we randomly select 500 images from CIFAR-10 where each image 
is represented by a pixel vector of length 3,072:
$32 \times 32$ pixels with red, green, and blue channels.

We evaluate the algorithms on sythetic random graphs as well as
real social network datasets from the Stanford Network Analysis
Project \citep{snapnets}. The specific datasets
used were as follows:
\begin{itemize}
\item \er, an Erd\H{o}s-Renyi random graph with
  number of nodes $n = 5000$ and edge probability
  $p = 0.01$.
\item \ba, a random graph in the
  Barabási-Albert preferential attachment
  model with parameter $n = 5000$ and
  initially $m_0 = 3$ nodes, and $3$ nodes
  added each iteration.
\item \fb, the ego-Facebook from \citet{snapnets}
  with $n=4039$, $m=88,234$.
\item \slashdot, the soc-Slashdot-0811
  social network from \citet{snapnets}
  with $n = 77,360$, $m=905,468$.
\item \pokec, the social network
  from \citet{snapnets} with $n=1,632,803$,
  and $m=30,622,564$.
\end{itemize}

\subsection{Additional Results}
Figs. \ref{fig:maxcut-val} and
\ref{fig:maxcut-query} show the
solution value and number of oracle
queries for the \maxcut application;
and Figs. \ref{fig:revmax-val} and
\ref{fig:revmax-query} show the same
for the \revmax application;
Figs.~\ref{fig:er} and~\ref{fig:slashdot}
show the results of solution value,
number of oracle queries, and memory usage
for \maxcut and \revmax applications with larger $k$ values;
Figs.~\ref{fig:epsi-fb} and~\ref{fig:epsi-sd}
show the results with different $\epsi$ values.

Observe that while occasionally
\qssp (gold star) obtains a lower
solution value than the other algorithms,
it more consistently returns high
solution values ($\ge 90\%$ of the
greedy algorithm)
across the five datasets and two
applications than the
other algorithms. Moreover, it uses
fewer queries, frequently by
more than an order of magnitude
over the next most efficient algorithm.

As for results comparing different $\epsi$ values,
in summary, LS, LS+ are very robust to changes in $\epsi$. 
This is because $\epsi$ impacts the frequency of deletion, which is a rare event. 
For LS+, it also impacts the post-processing procedure, which is why it does exhibit some dependence of objective value on $\epsi$.